%% file: verifiability.tex
\documentclass[english,11pt]{article}

\usepackage[sumlimits]{amsmath}
\usepackage{amsfonts,amssymb,dsfont}
\usepackage{times}
\usepackage{type1cm,graphicx}
\usepackage[usenames]{color}
\usepackage{bbm,mathtools}
\usepackage[left=2cm,top=2cm,right=2cm]{geometry}
\usepackage{hyperref}

\input{def}

\begin{document}

\author{Urmila Mahadev\footnote{California Institute of Technology. Supported by Templeton Foundation Grant 52536, ARO Grant W911NF-12-1-0541, NSF Grant CCF-1410022 and MURI Grant FA9550-18-1-0161. Email: mahadev@caltech.edu.} \footnote{A condensed version of this paper appeared in FOCS 2018; the FOCS 2018 version is essentially sections 1 and 2 of this paper.}}

\title{Classical Verification of Quantum Computations}

\maketitle

\setcounter{secnumdepth}{4}


\begin{abstract}
We present the first protocol allowing a classical computer to interactively verify the result of an efficient quantum computation. We achieve this by constructing a \textit{measurement protocol}, which enables a classical verifier to use a quantum prover as a trusted measurement device. The protocol forces the prover to behave as follows: the prover must construct an $n$ qubit state of his choice, measure each qubit in the Hadamard or standard basis as directed by the verifier, and report the measurement results to the verifier. The soundness of this protocol is enforced based on the assumption that the Learning with Errors problem is computationally intractable for efficient quantum machines.
\end{abstract}

\section{Introduction}
We propose a solution to the open question of verifying quantum computations through purely classical means. The question is as follows: is it possible for an efficient classical verifier (a \BPP\ machine) to verify the output of an efficient quantum prover (a \BQP\ machine)? This question was first raised by Daniel Gottesman in 2004 (\cite{aaronsonpost}). In the absence of any techniques for tackling this question, two weaker formulations were considered. In the first, it was shown that if the verifier had access to a small quantum computer, verification of all efficient quantum computations was possible (\cite{broadbent2008ubq}, \cite{fk2012}, \cite{abe2008}, \cite{abem}). The second formulation considered a classical polynomial time verifier interacting with two entangled, non-communicating quantum provers (rather than just one machine), and showed that in this setting it was possible to verify the result of an arbitrary quantum computation (\cite{ruv2012}). Although both lines of work initiated extensive research efforts, the question of classical verification by interaction with a single quantum computer has remained elusive.

In this paper, we show that a classical polynomial time verifier (a \BPP\ machine) can interact with an efficient quantum prover (a \BQP\ machine) in order to verify \BQP\ computations. We rely on the additional assumption that the verifier may use post-quantum cryptography that the \BQP\ prover cannot break; it remains an open question whether such a protocol is possible without assumptions. More specifically, we rely on the assumption that the Learning with Errors problem (\cite{regev2005}) cannot be solved by an efficient quantum machine. 

\subsection{Classical Commitment to Quantum States}\label{sec:intromeasprotocol}

The core of our construction is a \textit{measurement protocol}, an interactive protocol between an efficient quantum prover and a classical verifier which is used to force the prover to behave as the verifier's trusted measurement device, or, equivalently, to classically commit to a quantum state. At a very high level, the measurement protocol operates as follows: the prover will report a classical string $y$, which is meant to serve as a classical commitment to a quantum state. The verifier then selects a measurement basis (either Hadamard or standard) for each of the $n$ qubits of the quantum state. Finally, the prover is asked to report the chosen measurement results of the committed quantum state. The guarantee of the measurement protocol is that there must exist a quantum state underlying the measurement results reported by the prover. In other words, the classical string $y$ enforces non-adaptivity on the part of the prover; the prover is \textit{committed} to a single quantum state, and cannot change his mind when he reports the requested measurement results to the verifier. Therefore, the classical string $y$ serves as a commitment to a quantum state with an exponentially large description.

To formalize the idea of a measurement protocol, we now describe its completeness and soundness conditions, beginning with a small amount of necessary notation. In our measurement protocol, an honest prover constructs an $n$ qubit quantum state $\rho$ of his choice, and the verifier would like each qubit to be measured in either the standard basis or the Hadamard basis. Denote the choice of measurement basis by an $n$ bit string $h = (h_1,\ldots, h_n)$. For a prover $\Prov$ and a measurement basis choice $h = (h_1,\ldots, h_n)$, we define $D_{\Prov,h}$ to be the resulting distribution over the measurement result $m\in\{0,1\}^n$ obtained by the verifier. For an $n$ qubit state $\rho$, we define $D_{\rho,h}$ to be the distribution obtained by measuring $\rho$ in the basis corresponding to $h$. 

Our measurement protocol is complete, in the following sense: for all efficiently computable $n$ qubit states $\rho$, there exists a prover $\Prov$ such that $D_{\Prov,h}$ is approximately equal to $D_{\rho,h}$ for all $h$. Moreover, $\Prov$ is accepted by the verifier with all but negligible probability. Our soundness notion for the measurement protocol is slightly more complex, but a simplified form is as follows: if the prover $\Prov$ is accepted by the verifier with perfect probability, there exists an efficiently computable $n$ qubit quantum state $\rho$ underlying the distribution over $m$. More precisely, for all $h\in \{0,1\}^n$, $D_{\rho,h}$ is computationally indistinguishable from $D_{\Prov,h}$. The full soundness guarantee is achieved by making this statement robust: we will show that if a prover $\Prov$ is accepted by the verifier on basis choice $h$ with probability $1 - p_h$, there exists a prover $\Prov'$ who is always accepted by the verifier and the statistical distance between $D_{\Prov,h}$ and $D_{\Prov',h}$ is approximately $\sqrt{p_h}$ for all $h$.

\subsection{Linking Commitment to Verification}\label{sec:introlink}
So far, we have described the goal of our measurement protocol, which is to force the prover to behave as the verifier's trusted standard/Hadamard basis measurement device. To link our measurement protocol to verification, we simply need to show that a classical verifier who also has access to such a trusted measurement device can verify the result of \BQP\ computations. 

To describe how such verification can be done, we briefly recall a method of verifying classical computations. Assume that, for a language $L\in \textsf{P}$ and an instance $x$, the verifier wishes to check whether $x\in L$. To do so, the verifier can reduce $x$ to a 3-SAT instance, ask the prover for a satisfying variable assignment, and verify that the assignment satisfies the instance. There is an analogous setting for quantum computations, in which, for the language $L\in $ \BQP, the instance $x$ can be reduced to a local Hamiltonian instance $H_x$, an $n$ bit variable assignment corresponds to an $n$ qubit quantum state $\ket{\psi}$, and the fraction of unsatisfied clauses corresponds to the energy of the Hamiltonian $H_x$ with respect to the state $\ket{\psi}$ (\cite{kitaev2002caq}). If $x\in L$, there exists a state with low energy with respect to $H_x$; if not, all states will have sufficiently high energy. We will rely on the fact that the energy of $\ket{\psi}$ with respect to $H_x$ can be estimated by performing only standard/Hadamard basis measurements (\cite{xzhamiltonian}). 

With this analogy, verification of a quantum computation can be performed by a classical verifier with access to a trusted standard/Hadamard basis measurement device as follows (\cite{mf2016}): the verifier first reduces the instance $x$ to be verified to a local Hamiltonian instance $H_x$, then requests an $n$ qubit state from the prover, and finally checks (via standard/Hadamard basis measurements) if the received state has low energy with respect to $H_x$. If so, the verifier is assured that $x\in L$. Note that the verifier's only quantum capability is the ability to perform standard/Hadamard basis measurements on a fixed quantum state. Therefore, in order to transition to a classical verifier, we only need to enable the verifier to use the quantum prover as a trusted measurement device; the verifier needs to be able to guarantee that there exists a quantum state underlying the measurement results reported by the prover. This is exactly the guarantee of the measurement protocol.

\subsection{Measurement Protocol Construction}

Now that we have seen how to link our measurement protocol to verification, we proceed to describing how our measurement protocol works; as stated earlier, it is essentially a protocol which allows a classical string to serve as a commitment to a quantum state. Ideally, this commitment would operate as follows: at the start of the protocol, the verifier asks the prover for a classical commitment to a quantum state. This commitment should guarantee that the prover is forced to perform either a Hadamard or standard basis measurement as directed by the verifier on the committed state. Such a commitment scheme is enough to satisfy the soundness condition of the measurement protocol described above. Note the difference between this notion of commitment and the standard cryptographic notion: our commitment differs in that it does not need to hide the quantum state, but is similar since it is binding the prover to measuring the state he has committed to. 

The actual commitment performed in our measurement protocol has a weaker guarantee than the ideal scheme described above: the prover is asked to commit to a state $\rho$, but his hands are not entirely tied when it comes to the measurement of the state. The prover must perform the standard basis measurement on the committed state $\rho$, but he can perform a specific type of deviation operator prior to Hadamard measurement: this operator must commute with standard basis measurement (for example, it could be a Pauli $Z$ operator). The key point here is that this weaker commitment protocol is still strong enough to guarantee the soundness of the measurement protocol. To see why, observe that the deviation operator could have been applied prior to the commitment, creating a different committed state $\rho'$. Due to the fact that the deviation operator commutes with standard basis measurement, the measurement distribution obtained by the verifier (for both the Hadamard and standard bases) would have been equivalent had the prover instead committed to the state $\rho'$, and honestly performed the measurement requested by the verifier on $\rho'$.

The commitment structure described above is obtained from a classical cryptographic primitive called a trapdoor claw-free function, a function $f$ which is two to one, easy to invert with access to a trapdoor, and for which it is computationally difficult to find any pair of preimages with the same image. Such a pair of preimages is called a claw, hence the name claw-free. These functions are particularly useful in the quantum setting, due to the fact that a quantum machine can create a uniform superposition over a random claw $(x_0,x_1)$: $\frac{1}{\sqrt{2}}(\ket{x_0} + \ket{x_1})$. Using this state, a quantum machine can obtain \textit{either} a string $d\neq 0$ such that $d\cdot(x_0\oplus x_1) = 0$ or one of the two preimages $x_0,x_1$. It seems difficult to obtain this information classically; if it were possible, due to rewinding, it would also be possible to hold both pieces of information at once, thereby implying that the classical machine could produce one preimage and a single bit of the other. In \cite{oneproverrandomness}, this advantage was introduced, proven to be hard for a classical machine (providing a classical test of quantum supremacy), and was also used to generate information theoretic randomness from a single quantum device.

Trapdoor claw-free functions are used in our measurement protocol as follows. The prover is first asked to commit to a state of his choice (assume the state is $\alpha_0\ket{0} + \alpha_1\ket{1}$) by entangling it with a random claw $(x_0,x_1)$, creating the state:
\begin{equation}\label{eq:introcommittedstate}
    \alpha_0\ket{0}\ket{x_0} + \alpha_1\ket{1}\ket{x_1}\nonumber
\end{equation}
The corresponding classical commitment, which is sent to the verifier, is the image $y = f(x_0) = f(x_1)$. The goal is to show that, when asked for measurement results, the prover actually measures the committed state in the basis requested by the verifier; the prover should not be able to adaptively change his state based on the verifier's measurement choice. If we imagine for a second that the prover did not have access to the second register of the state in \eqref{eq:introcommittedstate} (containing $x_0$ or $x_1$), we can see why this would be true: the prover's state is entangled with a register which he has no control over, thereby preventing him from making changes to the state. This intuition is actually what underlies the soundness proof: the committed state is entangled with a claw, and the prover knows at most one member of the claw. This greatly limits the prover's ability to modify the state after commitment. In order to force the prover to perform the desired measurement on the committed state, the claw-free property is strengthened in two different ways and is then used to randomize any operator applied by the prover which is a deviation from the requested measurement, rendering the deviation of the prover essentially useless. We provide the definition and construction of an \textit{extended trapdoor claw-free family} which satisfies the strengthened claw-free properties. This family is an extension of the family given in \cite{oneproverrandomness}. Like the construction in \cite{oneproverrandomness}, our construction relies on the computational assumption that a \BQP\ machine cannot solve the Learning with Errors problem with superpolynomial noise ratio (\cite{regev2005}). 

The main result of our paper is stated informally as follows (stated formally as Theorem \ref{thm:lwetoqpip}):
\begin{thm}\label{thm:oneprover}{\textbf{(Informal)}}
Assuming the existence of an extended trapdoor claw-free family, all decision problems which can be efficiently computed in quantum polynomial time (the class \BQP) can be verified by an efficient classical machine through interaction (formally, \BQP\ = \QPIP$_0$).
\end{thm}
In Section \ref{sec:lweextended}, we provide a Learning with Errors based construction of an extended trapdoor family, providing a proof of the following theorem (stated formally as Theorem \ref{thm:lwetcfextended} in Section \ref{sec:lweextended}):
\begin{thm}\label{thm:lwetcfextendedinformal}{\textbf{(Informal)}}
Under the assumption that the Learning with Errors problem with superpolynomial noise ratio is computationally intractable for an efficient quantum machine, there exists an extended trapdoor claw-free family. 
\end{thm}

\subsection{Related Work}
The basic idea of using trapdoor claw-free functions to constrain the power of the prover in the context of randomness generation was introduced in \cite{oneproverrandomness}. This idea was further used in a computational setting to hide information from the prover (but not constrain the prover) in \cite{mahadev2017}. What is new to our paper is that it constrains the prover in the context of carrying out a particular computation. This requires the verifier to exert a much greater degree of control over the prover than in randomness generation. To accomplish this goal, we develop a new protocol and proof techniques for strongly characterizing the state of an untrusted prover.


\section{Overview}\label{sec:overview}

We now present an overview of this paper, proceeding as follows. Our measurement protocol relies on two cryptographic primitives which give a \BPP\ verifier some leverage over a \BQP\ prover. We begin by describing these primitives in Section \ref{sec:overviewprimitives}. Using these primitives, we can describe our measurement protocol in Section \ref{sec:overviewmeas}. In Sections \ref{sec:overviewexistence} - \ref{sec:overviewreductiontotrivial}, we show how the two cryptographic primitives can be used to guarantee soundness of the measurement protocol, in the sense that all provers must essentially be creating a state and measuring it in the basis chosen by the verifier. In Section \ref{sec:overviewmeastoqpip}, we show how to extend our measurement protocol to a verification protocol for all of \BQP, as described in Section \ref{sec:introlink}. 

There is an important caveat for the overview: in the overview, we build the measurement protocol by relying on idealized cryptographic primitives. Since we do not know how to construct these idealized primitives, the rest of the paper instead relies on approximate versions. These approximations come with added technical details, and we use idealized primitives in the overview in order to cover all of the ideas of the paper while avoiding these technical details; in other words, the measurement protocol and analysis presented in the overview are the same as those presented in the main body of the paper, up to the technical details which come with the approximate primitives. 

\subsection{Cryptographic Primitives}\label{sec:overviewprimitives}

 \subsubsection{Trapdoor Claw-Free Families}\label{sec:overviewtrapdoorclawfree}
The first cryptographic primitive we will use is a function family $\mathcal{F} = \{f_{k,b}: \sX\rightarrow \sY\}$ (for $b\in\{0,1\}$) called a \textit{trapdoor claw-free} function family. For convenience, we will assume in this overview that $\sX = \{0,1\}^w$. A trapdoor claw-free function family is a family of functions which are two to one (both $f_{k,0}(\cdot)$ and $f_{k,1}(\cdot)$ are injective and their images are equal) and for which it is computationally difficult to find a \textit{claw}, i.e. a pair of points $x_0$ and $x_1$ which have the same image ($f_{k,0}(x_0) = f_{k,1}(x_1)$). Given $y$ in the image of $f_{k,0}$ or $f_{k,1}$, the trapdoor $t_k$ of the functions $f_{k,0},f_{k,1}$ allows the recovery of both preimages of $y$. The trapdoor claw-free family also satisfies two hardcore bit properties, which are stronger versions of the claw-free property. At a high level, the first states that it is computationally difficult to hold both a preimage and a single bit of the other preimage, or, equivalently, it is difficult to simultaneously hold a preimage $x_b$, a non-zero string $d$ and the bit $d\cdot (x_0\oplus x_1)$, where $(x_0,x_1)$ form a claw. The second states that there exists strings $d$ for which finding just the bit $d\cdot (x_0\oplus x_1)$ is computationally difficult. These two properties are specified as needed (in Definitions \ref{def:informaladaptivehardcore} and \ref{def:informaladaptivehardcorecross}).

In this overview, we will assume the existence of the function family described above for simplicity. However, since we do not know how to construct such a family, in the rest of the paper we will instead rely on an approximate version of this family. We provide the definition of the function family we will use, which we call an extended trapdoor claw-free family, in Section \ref{sec:functiondefinitions}, Definition \ref{def:extendedtrapdoorclawfree}. The construction of this family (from Learning with Errors) is given in Section \ref{sec:lweextended}. Both the definition and construction are extensions of those given in \cite{oneproverrandomness}.

We now describe a \BQP\ process we call \textit{state commitment}, which requires a function key $k$ corresponding to functions $f_{k,0},f_{k,1}\in\mathcal{F}$ (we assume that computing the functions $f_{k,0},f_{k,1}$ only requires access to the function key $k$). The state commitment process is performed with respect to an arbitrary single qubit state $\ket{\psi}$:  
\begin{equation}\label{eq:proverstatepsi}
\ket{\psi} = \sum_{b\in\{0,1\}}\alpha_b\ket{b}
\end{equation}
The commitment process consists of two steps. First, the functions $f_{k,0}, f_{k,1}$ are applied in superposition, using $\ket{\psi}$ to determine whether to apply $f_{k,0}$ or $f_{k,1}$ and a uniform superposition over $x\in \sX$ as the input to $f_{k,0}$ or $f_{k,1}$:
\begin{equation}\label{eq:proverstateafterfunctionapp}
\frac{1}{\sqrt{|\sX|}}\sum_{b\in\{0,1\}}\sum_{x\in\sX}\alpha_b\ket{b}\ket{x}\ket{f_{k,b}(x)}
\end{equation}
Second, the final register of the resulting state is measured, obtaining $y\in \sY$. At this point, the state is:
\begin{equation}\label{eq:committedstateoverview}
\sum_{b\in\{0,1\}} \alpha_b\ket{b}\ket{x_{b,y}}
\end{equation}
where $x_{0,y}$ and $x_{1,y}$ are the two preimages of $y$. We will call the qubit containing $b$ the \textit{committed qubit}, the register containing $x_{b,y}$ the \textit{preimage register} and the string $y$ the \textit{commitment string}. The crucial point here is that, due to the claw-free nature of the functions $f_{k,0},f_{k,1}$, it is computationally difficult for a \BQP\ machine to compute both inverses $x_{0,y}$ and $x_{1,y}$ given only $y$. However, with access to the trapdoor $t_k$, both inverses can be computed from $y$. If we think of the state commitment process in an interactive setting, in which the verifier selects the function key and the trapdoor and the \BQP\ prover performs the commitment process (sending the commitment $y$ to the verifier), the \BQP\ prover cannot compute both inverses, but the verifier can. This gives the verifier some leverage over the prover's state.

A key property of the committed state in \eqref{eq:committedstateoverview} is that it allows a logical Hadamard measurement up to an $X$ Pauli operator, which is performed as follows. First, a Hadamard transform is applied to both the committed qubit and preimage register of the state in \eqref{eq:committedstateoverview}:
\begin{equation}\label{eq:afterhadamardoverview}
\frac{1}{\sqrt{|\sX|}}\sum_{d\in\sX} X^{d\cdot (x_{0,y}\oplus x_{1,y})}H\ket{\psi}\otimes Z^{x_{0,y}}\ket{d}
\end{equation}
The next step in applying the logical Hadamard measurement is to measure the second (preimage) register, obtaining $d\in\sX$. The state at this point is:
\begin{equation}\label{eq:hadamardbeforemeasurementoverview}
 X^{d\cdot (x_{0,y}\oplus x_{1,y})}H\ket{\psi}
\end{equation}
To obtain the Hadamard measurement of $\ket{\psi}$, the operator $X^{d\cdot (x_{0,y}\oplus x_{1,y})}$ (which we call the decoding operator and requires the trapdoor) is first applied, followed by a standard basis measurement of $H\ket{\psi}$. Note that these two operations commute: it is equivalent to first perform a standard basis measurement of the state in \eqref{eq:hadamardbeforemeasurementoverview} followed by applying the $X$ decoding operator. The $X$ decoding operator applied after measurement is simply the classical XOR operation.

We can again think of this logical Hadamard transform in the interactive setting, in which the \BQP\ prover applies the Hadamard transform to obtain the state in \eqref{eq:afterhadamardoverview} and then measures the committed qubit and preimage register, sending the measurement results $b'\in\{0,1\}$ and $d\in\{0,1\}^w$ to the verifier. The verifier decodes the measurement $b'$ by XORing it with $d\cdot (x_{0,y}\oplus x_{1,y})$ (which can be computed using the trapdoor) to obtain the bit $m$, which the verifier stores as the result of the Hadamard basis measurement. 

\subsubsection{Trapdoor Injective Function Families}\label{sec:overviewtrapdoorinjective}
The second primitive is a function family $\mathcal{G} = \{g_{k,b}: \sX\rightarrow \sY\}$ (for $b\in\{0,1\}$) called a \textit{trapdoor injective} function family. A trapdoor injective function family is a family of injective functions such that the images of $g_{k,0}(\cdot)$ and $g_{k,1}(\cdot)$ are disjoint. Given $y = g_{k,b}(x_{b,y})$, the trapdoor $t_k$ of the functions $g_{k,0},g_{k,1}$ allows recovery of $b,x_{b,y}$. We will also require that the trapdoor injective family is computationally indistinguishable from the trapdoor claw-free family: given a function key $k$, it must be computationally difficult to determine whether $k$ belongs to an injective or claw-free family. As in the case of trapdoor claw-free families, we will assume the existence of the function family described above for the purpose of this overview, but in the rest of the paper we will rely on an approximate version of this function family. We define the approximate version in Section \ref{sec:functiondefinitions} (Definition \ref{def:trapdoorinjective}) and construct it from Learning with Errors in Section \ref{sec:lweinjective}.  

The state commitment process described in Section \ref{sec:overviewtrapdoorclawfree} can also be performed with a function key $k$ corresponding to functions $g_{k,0},g_{k,1}\in\mathcal{G}$. At the stage of \eqref{eq:proverstateafterfunctionapp}, the following state has been created: 
\begin{equation}
\frac{1}{\sqrt{|\sX|}}\sum_{b\in\{0,1\}}\sum_{x\in\sX}\alpha_b\ket{b}\ket{x}\ket{g_{k,b}(x)}
\end{equation}
However, in this case when the last register is measured to obtain $y\in \sY$, the superposition over $b$ collapses. This is because the images of $g_{k,0}(\cdot)$ and $g_{k,1}(\cdot)$ are disjoint. It follows that with probability $|\alpha_b|^2$, $y\in g_{k,b}(\cdot)$. In this case, the state after measurement of $y$ is: 
\begin{equation}
\ket{b}\ket{x_{b,y}}
\end{equation}
where $x_{b,y}$ is the unique preimage of $y$. 

If we think of this process in the interactive setting (in which the \BQP\ prover performs the commitment and sends the verifier $y$), the verifier can use the trapdoor $t_k$ to extract $(b,x_{b,y})$ from $y$. Therefore, the verifier can obtain $b$, which is the result of standard basis measurement of the prover's state, simply by asking for the commitment $y$.

\subsection{Measurement Protocol}\label{sec:overviewmeas}
Given the two primitives described in Section \ref{sec:overviewprimitives}, we can now describe our measurement protocol (formally given in Protocol \ref{prot:measprotocol}). Before the protocol begins, the verifier selects the basis for which he would like a measurement result (either Hadamard or standard) for each of $n$ qubits. We will represent this basis choice by a string $h\in \{0,1\}^n$. If $h_i = 0$, it indicates that the standard basis was chosen for qubit $i$. 

The protocol begins with the commitment round. For all $i$ for which $h_i = 1$, the verifier samples a key $k_i$ and a trapdoor $t_{k_i}$ corresponding to a pair of trapdoor claw-free functions $f_{k_i,0}, f_{k_i,1}\in \mathcal{F}$. For all $i$ for which $h_i = 0$, the verifier samples a key $k_i$ and a trapdoor $t_{k_i}$ corresponding to a pair of trapdoor injective functions $g_{k_i,0},g_{k_i,1}\in\mathcal{G}$. The verifier sends the prover the function keys $k_1,\ldots,k_n$. The honest prover first constructs an $n$ qubit state of his choice and then commits to each qubit of the state using the corresponding function key, as described in Section \ref{sec:overviewprimitives}. The verifier receives the measurement results $y_1,\ldots,y_n\in \sY$ from the prover. 

Once the verifier receives $y_1,\ldots,y_n$, he chooses at random either to run a test round or a Hadamard round. In the case of a test round, the verifier requests standard basis measurements of the committed qubit and preimage register for all $n$ qubits. He receives $b_i'\in \{0,1\}$ and $x_i'\in \sX$ from the prover and checks if the measurement is a preimage of $y_i$. 

In the case of a Hadamard round, the verifier requests a Hadamard basis measurement of the committed qubit and preimage register for all $i$. The verifier receives $b'_i\in \{0,1\}$ and non-zero $d_i\in \sX$ for all $i$ (the verifier rejects otherwise). For all $i$ for which $h_i = 0$, the verifier ignores the measurement results $b'_i$ and $d_i$ and uses the trapdoor of $g_{k_i}$ to invert $y_i$. The verifier stores the first bit of the inverted value as the standard basis measurement result $m_i$. For all $i$ for which $h_i = 1$, the verifier first decodes $b'_i$ by XORing it with $d_i\cdot (x_{0,y_i}\oplus x_{1,y_i})$ (this can equivalently be thought of as applying the decoding operator $X^{d\cdot (x_{0,y_i}\oplus x_{1,y_i})}$- see the end of Section \ref{sec:overviewtrapdoorclawfree}). The verifier stores the result $m_i = b'_i\oplus d_i\cdot (x_{0,y_i}\oplus x_{1,y_i})$ as the Hadamard basis measurement result. 

Completeness (as defined in the introduction) of our measurement protocol follows immediately from the description of the state commitment process given in Sections \ref{sec:overviewtrapdoorclawfree} and \ref{sec:overviewtrapdoorinjective}.

\subsection{Measurement Protocol Soundness}\label{sec:overviewexistence}
We now give an overview of our soundness proof: we describe how to show that for $n = 1$ and a prover $\Prov$ who passes the test round perfectly, there exists a quantum state underlying the distribution over measurement results obtained by the verifier. The generalization to arbitrary $n$ (given in Section \ref{sec:measprotocolsoundness}) follows easily due to the fact that all $n$ function keys are drawn independently. The generalization to provers $\Prov$ who do not pass perfectly (also given in Section \ref{sec:measprotocolsoundness}) is straightforward as well: it is done by conditioning $\Prov$ on acceptance in a test round, thereby creating an efficient prover who passes the test round perfectly as long as $\Prov$ is accepted with non negligible probability. 

As the soundness proof is quite involved, we will begin by providing a brief sketch of the proof; the remainder of this section and the next (Section \ref{sec:overviewreductiontotrivial}) will fill in the details. To start, we will state the soundness guarantee we wish to prove a bit more precisely. Recall the description in Section \ref{sec:intromeasprotocol}. There, we defined $D_{\Prov,h}$ to be the distribution over the measurement result $m\in\{0,1\}^n$ obtained by the verifier when interacting with the prover $\Prov$ with measurement basis choice $h = (h_1,\ldots, h_n)$, and $D_{\rho,h}$ to be the distribution obtained by measuring an $n$ qubit state $\rho$ in the basis corresponding to $h$. With this notation, arguing that there exists a quantum state underlying the distribution over measurement results obtained by the verifier is equivalent to showing that there exists a quantum state $\rho$ such that for all $h\in\{0,1\}^n$, $D_{\Prov,h}$ and $D_{\rho,h}$ are computationally indistinguishable. As stated in the previous paragraph, we will focus on the case of $n = 1$.

The most intuitive path forward would be to use the actions of $\Prov$ to define a state $\rho$, and then argue that $D_{\Prov,h}$ and $D_{\rho,h}$ are computationally indistinguishable, by contradiction: the goal would be to show that if they are not, there exists an efficient attacker who can break one of the two hardness guarantees underlying our function families (from Section \ref{sec:overviewprimitives}). These two hardness guarantees are the computational indistinguishability of the injective and claw-free families, and the claw-free property. The difficulty in proving such a guarantee is that computing the distribution $D_{\Prov,h}$ requires trapdoor information, since $D_{\Prov,h}$ is the distribution over measurement results obtained by the verifier. For example, in the case that $h = 1$, the verifier must invert the trapdoor claw-free functions to decode the Hadamard basis results (as described in Section \ref{sec:overviewmeas}). How can an efficient attacker use $D_{\Prov,h}$ as a tool in breaking one of our hardness guarantees, if even sampling from the distribution requires a trapdoor? At a very high level, this is possible because the verifier's decoding only involves a single bit, and we can construct two strengthened claw-free properties, called hardcore bit properties (briefly mentioned at the start of Section \ref{sec:overviewtrapdoorclawfree}), which hold even after this single bit is revealed. 

The soundness proof proceeds as follow. The first step is to characterize the state of the prover, which will allow us to both analyze the distribution $D_{\Prov,h}$ obtained by the verifier and to define the state $\rho$ appropriately. The test round serves exactly this purpose: if the prover $\Prov$ passes the test round perfectly, we know that at some point, the prover's state \textit{must} be in a superposition over the two preimages of $y$. This gives us a well-defined notion of a  committed state for an arbitrary, dishonest prover, as long as the prover passes the test round perfectly. We can then assume that when asked for measurement results in the test round, the prover simply measures the committed state in the standard basis, and when asked for measurement results in the Hadamard round, the prover applies some unitary operator $U$ to the committed state, and then measures in the Hadamard basis (for an honest prover, $U = \mcI$). The details of this characterization are fleshed out in Section \ref{sec:overviewproverbehavior}. 

We can use this characterization (the committed state as well as the operator $U$) along with the verifier's $X$ decoding operator to define a state $\rho$. The state $\rho$ is defined so that the distribution obtained by measuring $\rho$ in the Hadamard basis ($D_{\rho,1}$) is equal to the distribution $D_{\Prov,1}$ obtained by the verifier for basis choice $h = 1$. With such a definition, the only step needed to complete the proof is to argue that standard basis measurements of $\rho$ match up with the verifier's stored results (i.e. $D_{\rho,0}$ is computationally indistinguishable from $D_{\Prov,0}$). Observe that if $U$ commutes with standard basis measurement of the committed qubit, this statement follows: simply shift $U$ to be \textit{part} of the committed state. In other words, in this case, the prover may as well have committed to a different state, and then performed the Hadamard round measurements honestly. Such a statement clearly completes our proof of soundness, as there is a well defined committed state, which the prover measures in the basis requested by the verifier. This part of the argument is outlined in Section \ref{sec:overviewstateexistenceproof}.

We call an operator which commutes with standard basis measurement of the committed qubit $X$-trivial. The crux of the soundness proof is proving that, for an arbitrary prover $\Prov$ who passes the test round perfectly, we can indeed assume that $U$ is $X$-trivial, in the following sense: we can replace $U$ with an $X$-trivial operator $U'$, and we can prove that this replacement yields measurement distributions which are computationally indistinguishable from $D_{\Prov,h}$ for both choices of $h$. The reason that this computational indistinguishability holds is because $U$ is computationally randomized, by two sources: it acts on a state which is a superposition over two preimages, one of which must be unknown to the prover, and is followed by the verifier's decoding operator, which also appears computationally random to the prover. If we replace this computational randomness with information theoretic randomness, the resulting randomized operator $U'$ does in fact commute with standard basis measurement of the committed qubit. The two hardcore bit properties mentioned earlier (at the start of Section \ref{sec:overviewtrapdoorclawfree}) capture this computational randomness, and we argue that if the two distributions resulting from $U$ and $U'$ are computationally distinguishable, it must be possible to break one of these two properties. This argument is outlined in Section \ref{sec:overviewreductiontotrivial}.

We now proceed to filling in the details of the above outline, beginning by characterizing the behavior of a general prover. We then use this characterization to define a quantum state $\rho$, and show that if this characterization satisfies a certain requirement (the $X$-trivial requirement), the measurement results of the prover are indeed consistent with measuring the state $\rho$ (i.e. $D_{\Prov,h}$ and $D_{\rho,h}$ are computationally indistinguishable for both $h = 0$ and $h = 1$). In Section \ref{sec:overviewreductiontotrivial} (which is the crux of this paper), we show how to enforce this requirement on general provers.

\subsubsection{Prover Behavior}\label{sec:overviewproverbehavior}
The analysis of the measurement protocol is based on understanding and characterizing the prover's Hilbert space and operations. We will rely on the following principle behind interactive proofs between a \BQP\ prover and a \BPP\ verifier. A round of the protocol begins with the verifier's message and ends with the prover's message. A general prover is equivalent, from the verifier's perspective, to a prover who begins each round by applying an arbitrary unitary operator to his space and then behaves exactly the same as an honest prover, culminating in a measurement (the result of which is sent to the verifier). This principle implies that an arbitrary prover measures the same registers that an honest prover does in each round, which will be particularly useful in our protocol.

Let $\Prov_0$ be an honest prover in our measurement protocol and assume the unitary operator he applies in the commitment round is $U_{C,0}$, after which he measures the commitment string register in the standard basis. As described in Section \ref{sec:overviewprimitives}, $\Prov_0$ has three designated registers: the register containing the committed qubit, the preimage register, and the commitment string register. Each message of $\Prov_0$ to the verifier is the result of the measurement of one of these registers. 

It follows from the principle above that a general prover $\Prov$ has the same designated registers as $\Prov_0$ and is characterized by 3 unitary operators: the unitary $U_C$ applied in the commitment round, the unitary $U_T$ applied in the test round, and the unitary $U_H$ applied in the Hadamard round. We assume that both $U_T$ and $U_H$ do not act on the commitment string register, since it has already been measured; the measurement result could have been copied into the auxiliary space, on which $U_T$ and $U_H$ can act. 

We now use the structure of our protocol to simplify the general prover one step further. There are only two possible messages of the verifier for the second round of our protocol: the message indicates either a test or Hadamard round. Due to this property, we can assume that the test round attack $U_T$ is equal to the identity operator. To see this, we only need to make one observation: the attack $U_T$ applied in the test round commutes with the measurement of the commitment string. Therefore, it could have been applied prior to reporting the commitment string $y$.

It follows that the general prover $\Prov$ described above is identical, from the perspective of the verifier, to a prover who applies the unitary $U_0 = U_TU_C$ immediately prior to measuring the commitment string register and applies $U = U_HU_T^\dagger$ prior to performing Hadamard basis measurements of the committed qubit and preimage register in the Hadamard round. We will say that such a prover is $\textit{characterized}$ by $(U_0, U)$. For a formal statement and proof of the above argument, see Claim \ref{cl:generalprover}.

The characterization of all provers by two unitary attacks allows us to use the test round of the measurement protocol to enforce that the prover's state has a specific structure, which is derived from the cryptographic primitives in Section \ref{sec:overviewprimitives}. Let $\Prov$ be a prover who passes the test round perfectly. If $h = 1$, the state of $\Prov$ at the start of either the test or the Hadamard round (i.e. immediately after reporting $y$) can be written as follows (the two preimages of $y$ are $x_{0,y},x_{1,y}$): 
\begin{equation}\label{eq:stateformatoverview}
\sum_{b\in\{0,1\}} \ket{b}\ket{x_{b,y}}\ket{\psi_{b,x_{b,y}}}
\end{equation}
where $\ket{\psi_{b,x_{b,y}}}$ contains all additional qubits held by the prover. This is because the verifier checks, in a test round, if he receives a valid pre-image from the prover. Since the prover simply measures the requested registers when asked by the verifier in a test round (i.e. he does not apply an attack in the test round), these registers must be in a superposition over the two preimages of the reported measurement result $y$. 

If $h = 0$ and $\Prov$ reports $y$, there is only one inverse of $y$. If we assume this inverse is $x_{b,y}$ (i.e. $g_{k,b}(x_{b,y})= y$), the state of $\Prov$ at the start of the test or Hadamard round can be written as follows, due to the same reasoning used in \eqref{eq:stateformatoverview}:
\begin{equation}\label{eq:stateformatoverviewstandard}
\ket{b}\ket{x_{b,y}}\ket{\psi_{b,x_{b,y}}}
\end{equation}
This structure enforced by the test run is the key to proving the existence of an underlying quantum state, as we will see shortly. 

\subsubsection{Construction of Underlying Quantum State}\label{sec:overviewstateexistenceproof}
We begin by using the characterization of general provers in Section \ref{sec:overviewproverbehavior} to define a single qubit state $\rho$ corresponding to a prover $\Prov$ who is characterized by $(U_0,U)$. Recall that $\Prov$ has a well defined committed qubit, which he measures when the verifier asks for the measurement of a committed qubit. Let $\rho'$ be the state of the committed qubit prior to the prover's measurement in the Hadamard round in the case that $h = 1$. We can think of $\rho'$ as encoded by the operator $Z^{d\cdot (x_{0,y}\oplus x_{1,y})}$, which is determined by the prover's measurements $d$ and $y$. This $Z$ operator is derived from the verifier's $X$ decoding operator applied in the measurement protocol; we have used a $Z$ operator here since the Hadamard measurement has not yet been performed. The single qubit state $\rho$ will be the result of applying the $Z$ decoding operator to the committed qubit $\rho'$.

Define $X$-trivial operators to be those which commute with standard basis measurement of the committed qubit. We now show that if the prover's Hadamard round attack $U$ is an $X$-trivial operator, the distribution $D_{\Prov,h}$ obtained by the verifier in the measurement protocol is computationally indistinguishable from the distribution which is obtained by measuring $\rho$ in the basis specified by $h$. 

Recall that $D_{\rho,h}$ is the distribution obtained by measuring $\rho$ in the basis corresponding to $h$. By construction, $D_{\rho,1} = D_{\Prov,1}$. If $h = 0$, there are two differences between the distribution $D_{\rho,h}$ and the distribution $D_{\Prov,h}$. The first difference lies in the function sampling: in our measurement protocol, an injective function is sampled if $h = 0$, but in the state $\rho$, a claw-free function is sampled. The second difference comes from how the standard basis measurement is obtained: in $D_{\Prov,h}$ the standard basis measurement is obtained from the commitment string $y$, but in $D_{\rho,h}$ the standard basis measurement is obtained by measuring $\rho$ (the committed qubit) in the standard basis.  

We can handle the first difference by making two key observations. First, the $Z$ decoding operator has no effect if $h = 0$; in this case, the committed qubit will be measured in the standard basis immediately after application of $Z$ in order to obtain $D_{\rho,h}$. Second, if the $Z$ decoding operator is not applied, the trapdoor $t_k$ is no longer needed to construct the distribution $D_{\rho,h}$. If $D_{\rho,h}$ is only dependent on the function key $k$ (and not the trapdoor $t_k$), the function key $k$ can be replaced with a function key which corresponds to a pair of trapdoor injective functions, rather than a pair of trapdoor claw-free functions, to obtain a computationally indistinguishable distribution. This is due to the computational indistinguishability between keys drawn from the trapdoor claw-free family $\mathcal{F}$ and the trapdoor injective family $\mathcal{G}$.

Let $\rho_0$ be the committed qubit of the prover prior to measurement in the Hadamard round in the case that $h = 0$. Due to the argument above, the distribution $D_{\rho,0}$ is computationally indistinguishable from $D_{\rho_0,0}$. To address the second difference, we now show that measuring $\rho_0$ in the standard basis produces the same distribution obtained from extracting the standard basis measurement from the commitment string $y$. First, note that measuring the committed qubit prior to application of $U$ (i.e. at the start of the Hadamard round) results in the same measurement obtained from $y$; as seen in \eqref{eq:stateformatoverviewstandard}, the value of the committed qubit is equal to the value $m$ extracted from $y$, since the prover passes the test round perfectly. To complete our proof, recall that $U$ is $X$-trivial with respect to the committed qubit, and therefore commutes with standard basis measurement of the committed qubit.  

To recap, the argument above shows that there exists a quantum state underlying the distribution $D_{\Prov,h}$ as long as the prover's attack operator in the Hadamard round is an $X$-trivial operator. For a formal statement and complete proof of this argument, see Lemma \ref{lem:trivialprover}.

\subsection{Replacement of a General Attack with an X-Trivial Attack}\label{sec:overviewreductiontotrivial}
We can now proceed to the crux of this paper: assuming that $n = 1$ and $\Prov$ passes the test round perfectly, we show that there exists a prover $\Prov'$ such that $D_{\Prov,h}$ is computationally indistinguishable from $D_{\Prov',h}$ for both $h$ and $\Prov'$ attacks with an $X$-trivial operator in the Hadamard round. By the argument in Section \ref{sec:overviewstateexistenceproof} and the triangle inequality, this implies that there exists a state $\rho$ for which $D_{\Prov,h}$ and $D_{\rho,h}$ are computationally indistinguishable, thereby proving our soundness guarantee. 

Assume $\Prov$ is characterized by $(U_0,U)$. Then $\Prov'$ is characterized by $(U_0, \{U_x\}_{x\in\{0,1\}})$, where $\{U_x\}_{x\in\{0,1\}}$ is an $X$-trivial CPTP map\footnote{$U$ can be written in the form in \eqref{eq:upaulibasis} by decomposing its action on the first qubit in the Pauli basis; the matrix $U_{xz}$ is not necessarily unitary. For more detail, see Section \ref{sec:PauliNClifford}.}: 
\begin{eqnarray}
U &=& \sum_{x,z\in\{0,1\}} X^xZ^z \otimes U_{xz} \label{eq:upaulibasis}\\
U_x &=& \sum_{z\in\{0,1\}} Z^z\otimes U_{xz} \label{eq:udiagonalizedx}
\end{eqnarray}
Observe that if $h = 0$, $D_{\Prov,h} = D_{\Prov',h}$; this is simply because the standard basis measurement is obtained from the commitment $y$, which is measured prior to the Hadamard round attack $U$. This argument requires a bit more detail for $n>1$ and is given formally in Claim \ref{cl:reductiontosinglestandard}. We proceed to describing how to replace the attack $U$ in \eqref{eq:upaulibasis} with the CPTP map $\{U_x\}_{x\in \{0,1\}}$ in \eqref{eq:udiagonalizedx} in the case that the verifier chooses the Hadamard basis ($h = 1$). We will rely heavily on the structure of the prover's state, as written in \eqref{eq:stateformatoverview}. 

The replacement of $U$ with $\{U_x\}_{x\in \{0,1\}}$ will be done by using the $Z$ Pauli twirl (Corollary \ref{corol:diagonalizingmeas}). The $Z$ Pauli twirl is a technique which allows the replacement of $U$ with the CPTP map $\{U_x\}_{x\in\{0,1\}}$ by conjugating $U$ by a random $Z$ Pauli operator. More formally, Corollary \ref{corol:diagonalizingmeas} states that the following two CPTP maps are equivalent when followed by Hadamard basis measurement: 
\begin{equation}\label{eq:paulitwirloverview1}
\{\frac{1}{\sqrt{2}}(Z^r\otimes \mcI)U(Z^r\otimes \mcI)\}_{r\in \{0,1\}}
\end{equation}
\begin{equation}\label{eq:paulitwirloverview2}
\{U_x\}_{x\in \{0,1\}}
\end{equation}
To apply the $Z$ Pauli twirl in this setting, it suffices to show that replacing the prover's attack $U$ with the unitary attack $(Z\otimes \mcI)U (Z\otimes \mcI)$ results in a computationally indistinguishable distribution. 

To prove this statement, we will rely on the fact that there is already computational randomness, due to the trapdoor claw-free function, which is hiding both the $Z$ operator applied prior to $U$ and the $Z$ operator applied after. The computational randomness hiding the posterior $Z$ operator comes from the verifier's decoding operator $X^{d\cdot (x_{0,y}\oplus x_{1,y})}$ applied at the end of the measurement protocol (see Section \ref{sec:overviewmeas}); if this decoding operator is shifted prior to the Hadamard transform on the committed qubit, it acts as a $Z$ operator immediately after the attack $U$. The computational randomness hiding the anterior $Z$ operator results from the format of the prover's state. Recall that, since the prover is perfect, we can assume the prover begins the Hadamard round with a state of the form in \eqref{eq:stateformatoverview}:
\begin{equation}\label{eq:stateformatoverviewcompind}
\ket{\phi_y} = \sum_{b\in\{0,1\}}\ket{b}\ket{x_{b,y}}\ket{\psi_{b,x_{b,y}}}
\end{equation}
Intuitively, the prover's inability to determine the claw $(x_{0,y},x_{1,y})$ prevents him from being able to distinguish whether or not a $Z$ operator is applied to $\ket{\phi_y}$. More formally, we show in Section \ref{sec:overviewcross} that distinguishing between the states $\ket{\phi_y}$ and $(Z\otimes \mcI)\ket{\phi_y}$ boils down to the ability to determine the bit $d\cdot(x_{0,y}\oplus x_{1,y})$ for an arbitrary string $d$. 

In order to use these two sources of computational randomness to hide the difference between $U$ and $(Z\otimes \mcI)U(Z\otimes \mcI)$, it must be the case that the bit $d\cdot (x_{0,y}\oplus x_{1,y})$ is computationally indistinguishable from a uniformly random bit. Formalizing this requirement is a bit tricky, since $d$ is sampled from the state created by the prover. In the next section, we show how to prove computational indistinguishability between the distributions resulting from $U$ and $(Z\otimes \mcI) U (Z\otimes \mcI)$. As part of this process, we formalize the computational randomness requirement regarding $d\cdot (x_{0,y}\oplus x_{1,y})$ as two different hardcore bit conditions for the function pair $f_{k,0},f_{k,1}$.

\subsubsection{Computational Indistinguishability of Phase Flip}\label{sec:overviewcomppauli}
Let $\Prov$ be the prover characterized by $(U_0,U)$ and let $\hat{\Prov}$ be the prover characterized by $(U_0,(Z\otimes \mcI) U (Z\otimes \mcI))$. In this section, we will show that the distributions resulting from the two provers ($D_{\Prov,h}$ and $D_{\hat{\Prov},h}$) are computationally indistinguishable for all $h$. For convenience, we will instead refer to these two distributions as mixed states; let $\sigma_0$ be the mixed state corresponding to $D_{\Prov,h}$ and let $\sigma_1$ be the mixed state corresponding to $D_{\hat{\Prov},h}$, i.e.
\begin{equation}
    \sigma_0 = \sum_{m\in\{0,1\}}D_{\Prov,h}(m)\ket{m}\bra{m}
\end{equation}
\begin{equation}
    \sigma_1 = \sum_{m\in\{0,1\}}D_{\hat{\Prov},h}(m)\ket{m}\bra{m}
\end{equation}
To prove the computational indistinguishability of $\sigma_0$ and $\sigma_1$, each state is split into two terms (for $r\in\{0,1\}$):
\begin{eqnarray}
\sigma_r = \sigma_r^D + \sigma_r^C
\end{eqnarray}
By a straightforward application of the triangle inequality, we obtain that if $\sigma_0$ is computationally distinguishable from $\sigma_1$, either $\sigma_0^D$ and $\sigma_1^D$ are computationally distinguishable or $\sigma_0^C$ and $\sigma_1^C$ are. Note that even if the terms are not quantum states, the notion of computational indistinguishability (Definition \ref{def:compind}) is still well defined: to show that two terms, for example $\sigma_0^C$ and $\sigma_1^C$, are computationally indistinguishable, we need to show (informally) that there does not exist an efficiently computable CPTP map $\mathcal{S}$ such that the following expression is non negligible
\begin{equation}\label{eq:crosstermindist}
   | \tr((\ket{0}\bra{0}\otimes \mcI)\mathcal{S}(\sigma_0^C - \sigma_1^C)|
\end{equation}
In more detail, the density matrices $\sigma_0$ and $\sigma_1$ are created by beginning with the state $\ket{\phi_y}$ in \eqref{eq:stateformatoverviewcompind} and applying the operations of both the prover and verifier in the Hadamard round, followed by tracing out all but the first qubit. Therefore, to split $\sigma_0$ and $\sigma_1$ into two parts, we can equivalently split the density matrix of $\ket{\phi_y}$ into the following two parts, corresponding to the diagonal and cross terms: 
\begin{equation}\label{eq:stateformatoverviewdiagonalintro}
\sum_{b\in\{0,1\}}\ket{b}\bra{b}\otimes \ket{x_{b,y}}\bra{x_{b,y}}\otimes \ket{\psi_{b,x_{b,y}}}\bra{\psi_{b,x_{b,y}}}
\end{equation}
\begin{equation}\label{eq:stateformatoverviewcrossintro}
\sum_{b\in\{0,1\}}\ket{b}\bra{b\oplus 1}\otimes \ket{x_{b,y}}\bra{x_{b\oplus 1,y}}\otimes \ket{\psi_{b,x_{b,y}}}\bra{\psi_{b,x_{b\oplus 1,y}}}
\end{equation}
Let $\sigma_0^D$ and $\sigma_1^D$ be the result of applying the operations of both the prover and the verifier in the Hadamard round to \eqref{eq:stateformatoverviewdiagonalintro}, followed by tracing out all but the first qubit. Recall the difference between $\sigma_0^D$ and $\sigma_1^D$: in the latter, the prover's attack $U$ is conjugated by $(Z\otimes \mcI)$. Define $\sigma_0^C$ and $\sigma_1^C$ similarly, but replace \eqref{eq:stateformatoverviewdiagonalintro} with \eqref{eq:stateformatoverviewcrossintro}. In the following two sections, we show that both pairs of terms are computationally indistinguishable.

\paragraph{Diagonal Terms}\label{sec:overviewdiagonal}
In this section, we will show that if there exists a \BQP\ attacker $\mathcal{A}'$ who can distinguish between the terms $\sigma_0^D$ and $\sigma_1^D$, then there exists a \BQP\ attacker $\mathcal{A}$ who can violate the following informal hardcore bit property of the function family $\mathcal{F}$ (the formal statement is part of Definition \ref{def:trapdoorclawfree}):
\begin{deff}\label{def:informaladaptivehardcore}
Let $f_{k,0}$ and $f_{k,1}$ be sampled from a trapdoor claw-free family $\mathcal{F}$. Then there does not exist a \BQP\ attacker who, on input $k$, can produce $b\in \{0,1\}$, $x_b\in \sX$, $d\in\{0,1\}^w\setminus\{0^w\}$ and $c\in \{0,1\}$ such that $c = d\cdot (x_0\oplus x_1)$ where $f_{k,0}(x_0) = f_{k,1}(x_1)$\footnote{Note that the \BQP\ attacker only produces one preimage, $x_b$, and the other preimage $x_{b\oplus 1}$ is defined by the equality $f_{k,0}(x_0) = f_{k,1}(x_1)$.}. 
\end{deff}
We first describe the state $\sigma_0^D$, which is created by beginning with the state in \eqref{eq:stateformatoverviewdiagonalintro}, in more detail. Note that the state in \eqref{eq:stateformatoverviewdiagonalintro} can be efficiently created by following the prover's commitment process and then measuring the committed qubit and preimage register, thereby obtaining the measurement results $b,x_{b,y}$. To create $\sigma_0^D$, the attack $U$ is applied to the state in \eqref{eq:stateformatoverviewdiagonalintro}, followed by Hadamard measurement of the preimage register (obtaining a string $d$) and the committed qubit, and ending with the application of the verifier's $X$ decoding operator. Finally, all qubits but the first are traced out. $\sigma_1^D$ is almost the same as $\sigma_0^D$, except the attack $U$ is replaced with the attack $(Z\otimes \mcI) U (Z\otimes \mcI)$. Note that the initial phase operator has no effect, since it acts on the diagonal state in \eqref{eq:stateformatoverviewdiagonalintro}. The final phase flip, once it is shifted past the Hadamard transform, is equivalent to flipping the decoding bit of the verifier; it follows that $\sigma_1^D = X\sigma_0^D X$.

We now construct the \BQP\ attacker $\mathcal{A}$ who will use $\mathcal{A}'$ to violate the adaptive hardcore bit property in Definition \ref{def:informaladaptivehardcore}. Let $\sigma^D$ be the state $\sigma_0^D$ except for the verifier's decoding. Observe from the description in the previous paragraph that this state can be efficiently created, and as part of creating the state, the measurements $b,x_{b,y}$ and $d$ are obtained. The attacker $\mathcal{A}$ creates the state $\sigma^D$, obtaining the measurements $b,x_{b,y}$ and $d$ in the process. The decoding bit $d\cdot (x_{0,y}\oplus x_{1,y})$ determines which of the two states $\sigma_0^D$ and $\sigma_1^D$ $\mathcal{A}$ has created; if $d\cdot(x_{0,y}\oplus x_{1,y}) = r$, $\mathcal{A}$ has created $\sigma_r^D$. Now $\mathcal{A}$ can run $\mathcal{A}'$ on the resulting mixed state in order to learn $d\cdot (x_{0,y}\oplus x_{1,y})$. As a result, $\mathcal{A}$ holds the following information: $b,x_{b,y},d,$ and $d\cdot (x_0\oplus x_1)$, therefore violating Definition \ref{def:informaladaptivehardcore}.

\paragraph{Cross Terms}\label{sec:overviewcross}
In this section, we will show that the cross terms $\sigma_0^C$ and $\sigma_1^C$ are computationally indistinguishable. Since the cross terms are not quantum states, we first show below that if there exists a CPTP map $\mathcal{S}$ which distinguishes between $\sigma_0^C$ and $\sigma_1^C$, i.e. if the following expression is non negligible:
\begin{equation}
   | \tr((\ket{0}\bra{0}\otimes \mcI)\mathcal{S}(\sigma_0^C - \sigma_1^C)|
\end{equation}
then there exists an efficiently computable CPTP map $\mathcal{S}'$ such that the CPTP map $\mathcal{S}\mathcal{S}'$ distinguishes between the quantum states $\hat{\sigma}_0$ and $\hat{\sigma}_1$, defined as follows. The density matrix $\hat{\sigma}_r$ corresponds to the following pure state (recall $\ket{\phi_y}$ from \eqref{eq:stateformatoverviewcompind}):
\begin{equation}\label{eq:hatsigmaproof}
(Z^r\otimes \mcI)\ket{\phi_y} = (Z^r\otimes \mcI)(\sum_{b\in\{0,1\}}\ket{b}\ket{x_{b,y}}\ket{\psi_{b,x_{b,y}}})
\end{equation}
To prove this implication, it suffices to show that $\sigma_0^C - \sigma_1^C = \mathcal{S}'(\hat{\sigma}_0 - \hat{\sigma_1})$. This equality is straightforward for two reasons. First, $\frac{1}{2}(\hat{\sigma}_0 - \hat{\sigma_1})$ is equal to the cross term in \eqref{eq:stateformatoverviewcrossintro}. Second, both $\sigma_0^C$ and $\sigma_1^C$ also begin with \eqref{eq:stateformatoverviewcrossintro}, but followed by a CPTP map which is inefficient due to the verifier's use of the trapdoor in computing the decoding. To prove the existence of $\mathcal{S}'$, we show that taking the difference between $\sigma_0^C$ and $\sigma_1^C$ effectively removes the verifier's decoding, creating an efficient CPTP map $\mathcal{S}'$. 

Finally, we will show that an attacker who can distinguish between $\hat{\sigma}_0$ and $\hat{\sigma}_1$ can violate the following informal hardcore bit property of the function family $\mathcal{F}$ (the formal statement is part of Definition \ref{def:extendedtrapdoorclawfree}):
\begin{deff}\label{def:informaladaptivehardcorecross}
Let $f_{k,0}$ and $f_{k,1}$ be sampled from a trapdoor claw-free family $\mathcal{F}$. Then there exists $d\in\{0,1\}^w$ which satisfies two conditions. First, there exists a bit $c_k$ such that $d\cdot (x_0\oplus x_1) = c_k$ for all claws $(x_0,x_1)$ ($f_{k,0}(x_0) = f_{k,1}(x_1)$). Second, there does not exist a \BQP\ attacker who, on input $k$, can determine the bit $c_k$.  
\end{deff}
We begin by describing the cross term $\sigma_0^C$ (which is not a quantum state) in more detail. $\sigma_0^C$ is created by beginning with the expression in \eqref{eq:stateformatoverviewcrossintro}, copied here for reference:
\begin{equation}\label{eq:stateformatoverviewcross}
\sum_{b\in\{0,1\}}\ket{b}\bra{b\oplus 1}\otimes \ket{x_{b,y}}\bra{x_{b\oplus 1,y}}\otimes \ket{\psi_{b,x_{b,y}}}\bra{\psi_{b,x_{b\oplus 1,y}}}
\end{equation}
then applying the attack $U$, followed by Hadamard measurement of the committed qubit and preimage register and application of the verifier's $X$ decoding operator. Finally, all qubits but the first are traced out. $\sigma_1^C$ is almost the same, except the attack $U$ is replaced with the attack $(Z\otimes \mcI) U (Z\otimes \mcI)$. As in Section \ref{sec:overviewdiagonal}, the phase flip acting after $U$ is equivalent to flipping the decoding operator of the verifier (i.e. applying an $X$ operator to the matrix $\sigma_0^C$). The initial phase flip, which acts on the first qubit of \eqref{eq:stateformatoverviewcross}, results in a phase of -1. Combining these two observations yields the following equality:
\begin{equation}
    \sigma_1^C = -X\sigma_0^CX
\end{equation}
Taking the difference between $\sigma_0^C$ and $\sigma_1^C$ results in a matrix which has a uniform $X$ operator applied:
\begin{equation}\label{eq:crosstermuniformx}
    \sigma_0^C - \sigma_1^C = \sum_{r\in\{0,1\}}X^r\sigma_0^CX^r
\end{equation}
Observe that the CPTP map applied to \eqref{eq:stateformatoverviewcross} to create $\sigma_0^C$ is efficiently computable except for the verifier's $X$ decoding operator. In \eqref{eq:crosstermuniformx}, there is a uniform $X$ operator acting on $\sigma_0^C$, effectively replacing the verifier's decoding operator. Let $\mathcal{S}'$ be the resulting efficiently computable CPTP map. It follows immediately that $\sigma_0^C - \sigma_1^C = \mathcal{S}'(\hat{\sigma}_0 - \hat{\sigma}_1)$.

We now proceed to showing that an attacker $\mathcal{A}'$ who can distinguish between $\hat{\sigma}_0$ and $\hat{\sigma}_1$ can be used to violate the hardcore bit property specified in Definition \ref{def:informaladaptivehardcorecross}. Since the state $\hat{\sigma}_r$ is the state $\ket{\phi_y}$ from \eqref{eq:stateformatoverviewcompind} with the operator $Z^r$ applied to the committed qubit, an attacker who can distinguish between $\hat{\sigma}_0$ and $\hat{\sigma}_1$ can distinguish whether or not a $Z$ operator is applied to the committed qubit of $\ket{\phi_y}$. The following equality (which holds up to a global phase) shows that a $Z$ operator on the preimage register is equivalent to a $Z$ operator on the committed qubit: 
\begin{equation}\label{eq:zequality}
   (\mcI \otimes Z^d\otimes \mcI)(\sum_{b\in\{0,1\}}\ket{b}\ket{x_{b,y}}\ket{\psi_{b,x_{b,y}}}) = (Z^{d\cdot(x_{0,y}\oplus x_{1,y})}\otimes \mcI)(\sum_{b\in\{0,1\}}\ket{b}\ket{x_{b,y}}\ket{\psi_{b,x_{b,y}}}) 
\end{equation}
This equality, along with the attacker $\mathcal{A}'$, can be used to construct a \BQP\ attacker $\mathcal{A}$ who can determine $d\cdot (x_{0,y}\oplus x_{1,y})$ for an arbitrary fixed string $d$. $\mathcal{A}$ first constructs $\ket{\phi_y}$ (this is simply the prover's state after reporting the commitment string $y$, so it can be constructed efficiently). Next, $\mathcal{A}$ applies $Z^{d}$ to the preimage register of $\ket{\phi_y}$. Due to the equality in \eqref{eq:zequality}, this is equivalent to instead applying $Z^{d\cdot (x_{0,y}\oplus x_{1,y})}$ to the committed qubit. By running the attacker $\mathcal{A}'$, $\mathcal{A}$ can determine $d\cdot (x_{0,y}\oplus x_{1,y})$, therefore violating Definition \ref{def:informaladaptivehardcorecross}.

\subsection{Extension of Measurement Protocol to a Verification Protocol for \BQP}\label{sec:overviewmeastoqpip}
Our goal is to verify that an instance $x\in L$ for a language $L\in $ \BQP. Recall that each instance can be converted into a local Hamiltonian $H$ with the following property: if $x\in L$, $H$ has ground state energy at most $a$ and if $x\notin L$, $H$ has ground state energy at least $b$, where the gap $b-a$ is inverse polynomial. Therefore, to verify that an instance $x\in L$, a verifier with a quantum computer can simply ask the prover for the ground state and estimate the energy of the received state with respect to the Hamiltonian $H$. The soundness of such a protocol rests on the fact that if an instance $x\notin L$, all possible states sent by the prover will have energy $\geq b$. 

To use such a verification procedure in our setting, we need to rely on one more fact: the Hamiltonian $H$ can be written as a sum over terms which are each a product of $X$ and $Z$ operators \cite{xzhamiltonian}. Therefore, when the verifier is estimating the energy of a state sent by the prover, he only needs to perform Hadamard or standard basis measurements on each individual qubit. In \cite{mf2016}, the authors formalize the resulting protocol and use it to build a protocol in which a verifier with access to a single qubit register can verify the result of a \BQP\ computation. Their protocol achieves a completeness/ soundness gap which is negligibly close to 1 by performing polynomially many repetitions of the energy estimation process described above.  

In \cite{mf2016}, the prover sends single qubits to the verifier, who performs either Hadamard or standard basis measurements. To obtain a verification protocol for \BQP, we simply replace this step of their protocol with our measurement protocol. Completeness and soundness follow, since our measurement protocol allows the verifier to collect standard and Hadamard basis measurements of a given state, and our soundness claim guarantees that the distribution over measurement results obtained by the verifier comes from the measurement of an underlying quantum state. The extension of our measurement protocol to a verification protocol is described in Section \ref{sec:meastoqpip}.

\subsection{Paper Outline}
As noted in Section \ref{sec:overviewtrapdoorclawfree}, the trapdoor function families we use in the rest of the paper are not the ideal families used so far in the overview. We instead use approximations of these function families, which are defined in Section \ref{sec:functiondefinitions}. The protocol as described used several properties of the ideal families. We take care to define our approximate families to make sure that they satisfy these required properties, at the expense of additional notation. 

We begin with the definition of our extended trapdoor claw-free family in Section \ref{sec:functiondefinitions}. Section \ref{sec:measprotocol} presents the measurement protocol and proves completeness (Sections \ref{sec:overviewprimitives} to \ref{sec:overviewmeas} of the overview). Section \ref{sec:measprotocolsoundnessoverall} covers soundness of the measurement protocol: Sections \ref{sec:proverbehavior} and \ref{sec:underlyingstate} characterize the behavior of the prover (Section \ref{sec:overviewexistence} of the overview), Section \ref{sec:generaltotrivialhadamard} presents the argument outlined in Section \ref{sec:overviewreductiontotrivial}, in which we replace a general attack with an attack which commutes with standard basis measurement, and Section \ref{sec:measprotocolsoundness} provides a proof of soundness by combining the results from Sections \ref{sec:proverbehavior} to \ref{sec:generaltotrivialhadamard}. Finally, the extension of the measurement protocol to a \QPIP$_0$ (described in Section \ref{sec:overviewmeastoqpip}) is given in Section \ref{sec:meastoqpip}, providing the main result of this paper (the following informal statement is taken from the introduction and is stated formally as Theorem \ref{thm:lwetoqpip} in Section \ref{sec:meastoqpip}):

\vspace{2mm}
\begin{statement}{\textbf{Theorem \ref{thm:oneprover} \textit{(Informal)}}}
\textit{Assuming the existence of an extended trapdoor claw-free family as defined in Definition \ref{def:extendedtrapdoorclawfree}, \BQP\ $ = $ \QPIP$_0$.}
\end{statement}

In Section \ref{sec:lweextended}, we provide a Learning with Errors based construction of an extended trapdoor family, providing a proof of the following theorem (the following informal statement is taken from the introduction and is stated formally as Theorem \ref{thm:lwetcfextended} in Section \ref{sec:lweextended}):

\vspace{2mm}
\begin{statement}{\textbf{Theorem \ref{thm:lwetcfextendedinformal} \textit{(Informal)}}}
\textit{Under the assumption that the Learning with Errors problem with superpolynomial noise ratio is computationally intractable for an efficient quantum machine, there exists an extended trapdoor claw-free family. }
\end{statement}




\section{Preliminaries}\label{sec:verificationprelim}

\subsection{Notation}\label{sec:prelimnotation}
For all $q \in \mN$ we let $\mZ_q$ denote the ring of integers modulo $q$. We represent elements in $\mZ_q$ using numbers in the range $(-\frac{q}{2}, \frac{q}{2}] \cap \mZ$. We denote by $\trnq{x}$ the unique integer $y$ s.t.\ $y = x \pmod{q}$ and $y \in (-\frac{q}{2}, \frac{q}{2}]$. For $x\in\mZ_q$ we define $\abs{x}=|{\trnq{x}}|$. For a vector $\*u\in \mZ_q^n$, we write $\lVert \*u \rVert_{\infty}\leq \beta$ if each entry $u_i$ in $\*u$ satisfies $|u_i|\leq \beta$. Similarly, for a matrix $U\in \mZ_q^{n\times m}$, we write $\lVert U \rVert_{\infty}\leq \beta$ if each entry $u_{i,j}$ in $U$ satisfies $|u_{i,j}|\leq \beta$. When considering an $s\in \{0,1\}^n$ we sometimes also think of $s$ as an element of $\mZ_q^n$, in which case we write it as $\*s$. 

We use the terminology of polynomially bounded, super-polynomial, and negligible functions. A function $n: \mN \to \mR_+$ is \emph{polynomially bounded} if there exists a polynomial $p$ such that $n(\lambda)\leq p(\lambda)$ for all $\lambda \in \mN$. A function $n: \mN \to \mR_+$ is \emph{negligible} (resp. \emph{super-polynomial}) if for every polynomial $p$, $p(\lambda) n(\lambda)\to_{\lambda\to\infty} 0$ (resp. $ n(\lambda)/p(\lambda)\to_{\lambda\to\infty} \infty$).

We generally use the letter $D$ to denote a distribution over a finite domain $X$, and $f$ for a density on $X$, i.e. a function $f:X\to[0,1]$ such that $\sum_{x\in X} f(x)=1$. We often use the distribution and its density interchangeably. We write $U$ for the uniform distribution. We write $x\leftarrow D$ to indicate that $x$ is sampled from distribution $D$, and $x\leftarrow_U X$ to indicate that $x$ is sampled uniformly from the set $X$. 
We write $\mathcal{D}_X$ for the set of all densities on $X$.
For any $f\in\mathcal{D}_X$, $\supp(f)$ denotes the support of $f$,
\begin{equation*}
    \supp(f) \,=\, \big\{x\in X \,|\; f(x)> 0\big\}\;.
\end{equation*}
For two densities $f_1$ and $f_2$ over the same finite domain $X$, the square of the Hellinger distance between $f_1$ and $f_2$ is
\begin{equation}\label{eq:bhatt}
H^2(f_1,f_2) \,=\, 1- \sum_{x\in X}\sqrt{f_1(x)f_2(x)}\;.
\end{equation}
and the total variation distance between $f_1$ and $f_2$ is:
\begin{equation}\label{eq:deftotalvariation}
\TV{f_1}{f_2} \,=\, \frac{1}{2} \sum_{x\in X}|f_1(x) - f_2(x)|\;.
\end{equation}
The following lemma will be useful:
\begin{lem}\label{lem:projectionprob}
Let $D_0, D_1$ be distributions over a finite domain $X$. Let $X'\subseteq X$. Then:
\begin{equation}
    \Big|\Pr_{x\leftarrow D_0}[x\in X'] - \Pr_{x\leftarrow D_1}[x\in X']\Big| \,\leq\, \TV{D_0}{D_1}.
\end{equation}
\end{lem}
We require the following definition:
\begin{deff}{\textbf{Computational Indistinguishability of Distributions}}\label{def:compinddist}
Two families of distributions $\{D_{0,\lambda}\}_{\lambda\in\mN}$ and $\{D_{1,\lambda}\}_{\lambda\in\mN}$ (indexed by the security parameter $\lambda$) are computationally indistinguishable if for all quantum polynomial-time attackers $\mathcal{A}$ there exists a negligible function $\mu(\cdot)$ such that for all $\lambda\in\mN$
\begin{equation}
\Big|\Pr_{x\leftarrow D_{0,\lambda}}[\mathcal{A}(x) = 0] - \Pr_{x\leftarrow D_{1,\lambda}}[\mathcal{A}(x) = 0]\Big| \,\leq\, \mu(\lambda)\;.
\end{equation}
\end{deff}

\subsection{Learning with Errors and Discrete Gaussians}\label{sec:lweprelim}

For a positive real $B$ and a positive integer $q$, the truncated discrete Gaussian distribution over $\mZ_q$ with parameter $B$ is supported on $\{x\in\mZ_q:\,\|x\|\leq B\}$ and has density
\begin{equation}\label{eq:d-bounded-def}
 D_{\mZ_q,B}(x) \,=\, \frac{e^{\frac{-\pi\lVert x\rVert^2}{B^2}}}{\sum\limits_{x\in\mZ_q,\, \|x\|\leq B}e^{\frac{-\pi\lVert x\rVert^2}{B^2}}} \;.
\end{equation}
We note that for any $B>0$, the truncated and non-truncated distributions have statistical distance that is exponentially small in $B$~\cite[Lemma 1.5]{banaszczyk1993new}. For a positive integer $m$, the truncated discrete Gaussian distribution over $\mZ_q^m$ with parameter $B$ is supported on $\{x\in\mZ_q^m:\,\|x\|\leq B\sqrt{m}\}$ and has density
\begin{equation}\label{eq:d-bounded-def-m}
\forall x = (x_1,\ldots,x_m) \in \mZ_q^m\;,\qquad D_{\mZ_q^m,B}(x) \,=\, D_{\mZ_q,B}(x_1)\cdots D_{\mZ_q,B}(x_m)\;.
\end{equation}
The proof of the following lemma can be found in \cite{oneproverrandomness} (Lemma 2.4). 
\begin{lem}[\cite{oneproverrandomness}]\label{lem:distributiondistance}
Let $B$ be a positive real number and $q,m$ be positive integers. Let $\*e \in \mZ_q^m$. The Hellinger distance between the distribution $D = D_{\mZ_q^{m},B}$ and the shifted distribution $D+\*e$ satisfies
\begin{equation}
H^2(D,D+\*e) \,\leq\, 1- e^{\frac{-2\pi \sqrt{m}\|\*e\|}{B}}\;,
\end{equation}
and the statistical distance between the two distributions satisfies
\begin{equation}
\big\| D - (D+\*e) \big\|_{TV}^2 \,\leq\, 2\Big(1 - e^{\frac{-2\pi \sqrt{m}\|\*e\|}{B}}\Big)\;.
\end{equation}
\end{lem}

\begin{deff}[\cite{oneproverrandomness}]\label{def:lwehardness}
For a security parameter $\lambda$, let $n,m,q\in \mN$ be integer functions of $\lambda$. Let $\chi = \chi(\lambda)$ be a distribution over $\mZ_q$. The $\lwe_{n,m,q,\chi}$ problem is to distinguish between the distributions $(\*A, \*A\*s + \*e \pmod{q})$ and $(\*A, \*u)$, where $\*A$ is uniformly random in $\mZ_q^{m \times n}$, $\*s$ is a uniformly random row vector in $\mZ_q^n$, $\vc{e}$ is a  row vector drawn at random from the distribution $\chi^m$, and $\*u$ is a uniformly random vector in $\mZ_q^m$. Often we consider the hardness of solving $\lwe$ for {any} function $m$ such that $m$ is at most a polynomial in $n \log q$. This problem is denoted $\lwe_{n,q,\chi}$. When we write that we make the $\lwe_{n,q,\chi}$ assumption, our assumption is that no quantum polynomial-time procedure can solve the $\lwe_{n,q,\chi}$ problem with more than a negligible advantage in $\lambda$. 
\end{deff}

We use two additional properties of the LWE problem. The first is that it is possible to generate LWE samples $(\*A,\*A\*s+\*e)$ such that there is a trapdoor allowing recovery of $\*s$ from the samples.

\begin{thm}[Theorem 5.1 in~\cite{miccancio2012}]\label{thm:trapdoor}
Let $n,m\geq 1$ and $q\geq 2$ be such that $m = \Omega(n\log q)$. There is an efficient randomized algorithm $\GenTrap(1^n,1^m,q)$ that returns a matrix $\*A \in \mZ_q^{m\times n}$ and a trapdoor $t_{\*A}$ such that the distribution of $\*A$ is negligibly (in $n$) close to the uniform distribution. Moreover, there is an efficient algorithm $\Invert$ that, on input $\*A, t_{\*A}$ and $\*A\*s+\*e$ where $\|\*e\| \leq q/(C_T\sqrt{n\log q})$ and $C_T$ is a universal constant, returns $\*s$ and $\*e$ with overwhelming probability over $(\*A,t_{\*A})\leftarrow \GenTrap$. \end{thm}

The second property is the existence of a ``lossy mode'' for LWE. 
\begin{deff}[Definition 3.1 in~\cite{lwr}]\label{def:lossy}
Let $\chi = \chi(\lambda)$ be an efficiently sampleable distribution over $\mZ_q$. Define a lossy sampler $\tilde{\*A} \leftarrow \lossy(1^n,1^m,1^\ell,q,\chi)$ by  $\tilde{\*A} = \*B\*C +\*F$, where $\*B\leftarrow_U \mZ_q^{m\times \ell}$, $\*C\leftarrow_U \mZ_q^{\ell \times n}$, $\*F\leftarrow \chi^{m\times n}$. 
\end{deff}

\begin{thm}[Lemma 3.2 in~\cite{lwr}]\label{thm:lossy}
Under the $\lwe_{\ell,q,\chi}$ assumption, the distribution of $\tilde{\*A} \leftarrow \lossy(1^n,1^m,1^\ell,q,\chi)$ is computationally indistinguishable from $\*A\leftarrow_U \mZ_q^{m\times n}$. 
\end{thm}

\subsection{Quantum Computation Preliminaries}

\subsubsection{Quantum Operations}\label{sec:PauliNClifford}
We will use the $X, Y$ and $Z$ Pauli operators: $X = \left(\begin{array}{ll}
0&1\\1&0\end{array}\right)$, $Z=\left(\begin{array}{ll}
1&0\\0&-1\end{array}\right)$ and $Y = iXZ$. The $l$-qubit Pauli group consists of all elements of the form
$P=P_1\otimes P_2\odots  P_l$ where $P_i \in \{\mcI,X,Y,Z\}$, together with the multiplicative factors $-1$ and $\pm i$. We will use a subset of this group, which we denote as $\mbP_l$, which includes all operators $P = P_1\otimes P_2\odots P_l$ but not the multiplicative factors. We will use the fact that Pauli operators anti commute; $ZX = - XZ$. The subset $\mbP_l$ is a basis of the matrices acting on $l$ qubits. We can write any matrix $U$ over a vector space $A\otimes B$ (where $A$ is the space of $l$ qubits) as $\sum_{P\in \mbP_l}P\otimes U_P$ where $U_P$ is some (not necessarily unitary) matrix on $B$. 

Let $\mfC_l$ denote the $l$-qubit Clifford group. Recall that it is a finite subgroup of unitaries acting on $l$ qubits generated by the Hadamard matrix $H=\frac{1}{\sqrt{2}}\left(\begin{array}{ll}
1&1\\1&-1\end{array}\right)$, by $K=\left(\begin{array}{ll}
1&0\\0&i\end{array}\right)$, and by controlled-NOT (CNOT) which maps $\ket{a,b}$ to $\ket{a,a\oplus b}$ (for bits $a,b$). The Clifford group is characterized by the property that it maps the
Pauli group $\mbP_l$ to itself, up to a phase $\alpha\in\{\pm 1,\pm i\}$. That is:
$\forall C\in\mfC_l ,  \forall P\in \mbP_l, \exists \alpha\in\{\pm 1,\pm i\}: ~\alpha CPC^\dagger \in \mbP_l$.

The Toffoli gate $T$ maps  $\ket{a,b,c}$ to $\ket{a,b,c\oplus ab}$ (for $a,b,c\in\{0,1\}$). We will use the fact that the set consisting of the Toffoli gate and the Hadamard gate is a universal gate set for quantum circuits (\cite{shiuniversal}).

We will use completely positive trace preserving (CPTP) maps to represent general quantum operations. A CPTP map $\mathcal{S}$ can be represented by its Kraus operators, $\{B_{\tau}\}_{\tau}$. The result of applying $\mathcal{S}$ to a state $\rho$ is:  
\begin{equation}
\mathcal{S}(\rho) = \sum_{\tau} B_{\tau} \rho B_{\tau}^\dagger
\end{equation}
We say that two CPTP maps $\mathcal{S}$ and $\mathcal{S}'$ are equal if, for all density matrices $\rho$, $\mathcal{S}(\rho) = \mathcal{S}'(\rho)$. 

\subsubsection{Trace Distance}\label{sec:prelimtrace}
For density matrices $\rho,\sigma$, the trace distance $\T{\rho}{\sigma}$ is equal to:
\begin{eqnarray}\label{eq:deftracedistance}
\T{\rho}{\sigma} = \frac{1}{2}\tr(\sqrt{(\rho - \sigma)^2})
\end{eqnarray}
We will use the following fact (\cite{tracedistance}):
\begin{equation}
\T{\rho}{\sigma} = \max_P \tr(P(\rho - \sigma))
\end{equation}
where the maximization is carried over all projectors $P$. We will also use the fact that the trace distance is contractive under completely positive trace preserving maps (\cite{tracedistance}). The following lemma relates the Hellinger distance as given in \eqref{eq:bhatt} and the trace distance of superpositions (the proof of the lemma follows immediately from the definitions of Hellinger distance and trace distance):
\begin{lem}\label{lem:hellingertotrace}
Let $X$ be a finite set and $f_1,f_2\in\mathcal{D}_x$. Let 
$$ \ket{\psi_1}=\sum_{x\in X}\sqrt{f_1(x)}\ket{x}\qquad\text{and}\qquad  \ket{\psi_2}=\sum_{x\in X}\sqrt{f_2(x)}\ket{x}\;.$$
 Then 
 $$\T{\ket{\psi_1}\bra{\psi_1}}{\ket{\psi_2}\bra{\psi_2}}\,=\, \sqrt{ 1 - (1-H^2(f_1,f_2))^2}\;.$$
\end{lem}
We require the following definition, which extends Definition \ref{def:compinddist} to quantum states:
\begin{deff}{\textbf{Computational Indistinguishability of Quantum States}}\label{def:compind}
Two families of density matrices $\{\rho_{0,\lambda}\}_{\lambda\in\mN}$ and $\{\rho_{1,\lambda}\}_{\lambda\in\mN}$ (indexed by the security parameter $\lambda$) are computationally indistinguishable if for all efficiently implementable CPTP maps $\mathcal{S}$ there exists a negligible function $\mu(\cdot)$ such that for all $\lambda\in\mN$:
\begin{equation}
\Big| \tr((\ket{0}\bra{0}\otimes \mcI)\mathcal{S}(\rho_{0,\lambda} - \rho_{1,\lambda}))\Big| \,\leq\, \mu(\lambda)\;.
\end{equation}
\end{deff}

\subsubsection{Pauli Twirl}\label{sec:paulitwirl}
We call the conjugation of a unitary operator (or a CPTP map) by a random Pauli a Pauli twirl (\cite{paulitwirl}). The twirled version of a unitary $U$ is the CPTP map $\{(X^xZ^z)^\dagger U(X^xZ^z)\}_{x,z}$. If the Pauli is a random $X$ (or $Z$) Pauli operator, we call the conjugation an $X$ (or $Z$) Pauli twirl. A $Z$ Pauli twirl has the following effect:
\begin{lem}\label{lem:diagonalizing}{\textbf{Z Pauli Twirl}}
For a CPTP map with Kraus operators $\{B_{\tau}\}_{\tau}$, the following two CPTP maps are equal:
\begin{equation}\label{eq:equivsuperoperatortwirlproof}
    \bigg\{\frac{1}{\sqrt{2}}(Z^r\otimes \mcI)B_{\tau}(Z^r\otimes \mcI)\bigg\}_{r\in\{0,1\},\tau} = \{(X^x\otimes \mcI)B'_{x,\tau}\}_{x\in\{0,1\},\tau}
\end{equation}
where $B_{\tau} = \sum\limits_{x,z\in \{0,1\}} X^xZ^z \otimes B_{xz\tau}$ and the CPTP map $\{B'_{x,\tau}\}_{x\in\{0,1\},\tau}$ is defined as $B'_{x,\tau} = \sum\limits_{z\in \{0,1\}} Z^z \otimes B_{xz\tau}$.
\end{lem}
\begin{proof}
To prove the lemma, we show that applying either of the two CPTP maps in equation \eqref{eq:equivsuperoperatortwirlproof} on an arbitrary density matrix $\rho$ results in the same state. We begin with the CPTP map on the left of \eqref{eq:equivsuperoperatortwirlproof}:
\begin{equation}\label{eq:diagproof1}
\frac{1}{2}\sum_{r\in\{0,1\}, \tau} (Z^r\otimes \mcI) B_{\tau} (Z^r\otimes \mcI) \rho (Z^r\otimes \mcI)^\dagger B_{\tau}^\dagger (Z^r \otimes \mcI)^\dagger
\end{equation}
Using the fact that $B_{\tau} = \sum\limits_{x,z\in \{0,1\}} X^xZ^z \otimes B_{xz\tau}$, we can rewrite the expression from \eqref{eq:diagproof1} as:
\begin{equation}
\frac{1}{2}\sum_{\substack{r\in\{0,1\},\tau\\ x,z,x',z'\in\{0,1\}}}(Z^r X^xZ^z Z^r \otimes B_{xz\tau})\rho (Z^r X^{x'}Z^{z'} Z^r \otimes B_{x'z'\tau})^\dagger
\end{equation}
Next, we use the anti commutation properties of Pauli operators (by commuting $Z^r$ with both $X^x$ and $X^{x'}$) to obtain the following state: 
\begin{equation}
\frac{1}{2}\sum_{\substack{r\in\{0,1\},\tau\\ x,z,x',z'\in\{0,1\}}}(-1)^{r\cdot (x\oplus x')}(X^x Z^z \otimes B_{xz\tau})\rho ( X^{x'}Z^{z'} \otimes B_{x'z'\tau})^\dagger
\end{equation}
At this point, we can sum over $r$ to obtain $x = x'$, resulting in the following expression: 
\begin{equation}
\sum_{\substack{\tau\\ x,z,z'\in\{0,1\}}}(X^xZ^z  \otimes B_{xz\tau})\rho (X^{x}Z^{z'} \otimes B_{xz'\tau})^\dagger
\end{equation}
\begin{equation}
= \sum_{x\in\{0,1\},\tau} (X^x\otimes \mcI)B'_{x,\tau}\rho ((X^x\otimes \mcI)B'_{x,\tau})^\dagger
\end{equation}
\end{proof}

The following corollary follows from Lemma \ref{lem:diagonalizing} and captures the effect of a $Z$ Pauli twirl followed by Hadamard basis measurement:
\begin{corol}\label{corol:diagonalizingmeas}{\textbf{Z Pauli Twirl with Measurement}}
For a CPTP map with Kraus operators $\{B_{\tau}\}_{\tau}$, the following two CPTP maps are equal:
\begin{equation}
    \bigg\{\frac{1}{\sqrt{2}}(\ket{b}\bra{b}HZ^r\otimes \mcI)B_{\tau}(Z^r\otimes \mcI)\bigg\}_{b,r\in\{0,1\},\tau} = \{(\ket{b}\bra{b}H\otimes \mcI)B'_{x,\tau}\}_{b,x\in\{0,1\},\tau}
\end{equation}
where $B_{\tau} = \sum\limits_{x,z\in \{0,1\}} X^xZ^z \otimes B_{xz\tau}$ and the CPTP map $\{B'_{x,\tau}\}_{x\in\{0,1\},\tau}$ is defined as $B'_{x,\tau} = \sum\limits_{z\in \{0,1\}} Z^z \otimes B_{xz\tau}$. 
\end{corol}
\begin{proof}
We begin by applying Lemma \ref{lem:diagonalizing} to obtain the following equality:
\begin{equation}
    \bigg\{\frac{1}{\sqrt{2}}(\ket{b}\bra{b}HZ^r\otimes \mcI)B_{\tau}(Z^r\otimes \mcI)\bigg\}_{b,r\in\{0,1\},\tau} = \{(\ket{b}\bra{b}HX^x\otimes \mcI)B'_{x,\tau}\}_{b,x\in\{0,1\},\tau}
\end{equation}
To prove the corollary, we show that
\begin{equation}\label{eq:diagonalizingmeasproof}
\{(\ket{b}\bra{b}HX^x\otimes \mcI)B'_{x,\tau}\}_{b,x\in\{0,1\},\tau} = \{(\ket{b}\bra{b}H\otimes \mcI)B'_{x,\tau}\}_{b,x\in\{0,1\},\tau}
\end{equation}
This is straightforward and only requires one intermediate step:
\begin{equation}\label{eq:diagonalizingmeasproof1}
\{(\ket{b}\bra{b}HX^x\otimes \mcI)B'_{x,\tau}\}_{b,x\in\{0,1\},\tau} = \{(\ket{b}\bra{b}Z^xH\otimes \mcI)B'_{x,\tau}\}_{b,x\in\{0,1\},\tau}
\end{equation}
The second CPTP map in \eqref{eq:diagonalizingmeasproof1} is equal to the second CPTP map in \eqref{eq:diagonalizingmeasproof}, since a $Z$ operator applied prior to standard basis measurement has no effect, and can be replaced with the identity operator. 
\end{proof}

\subsection{\QPIP\ Definition}
A \QPIP\ is defined as follows:
\begin{deff}\label{def:QPIP}{\textbf{\cite{abe2008}/ \cite{abem}}} A language $\mcL$ is said to have a 
Quantum Prover Interactive Proof (\QPIP$_{\tau}$) with completeness $c$ 
and soundness $s$ (where $c-s$ is at least a constant) if there exists a pair of algorithms $(\Prov,\mathds{V})$, where $\Prov$ is the prover and $\mathds{V}$ is the verifier, 
with the following properties:
\begin{itemize}
\item The prover $\Prov$ is a \BQP\ machine, which also has access to a quantum channel which can transmit $\tau$ qubits. 
\item The verifier $\mathds{V}$ is a hybrid quantum-classical machine. 
Its classical part is a \BPP\ machine. The quantum part is a register of $\tau$ qubits, on which the verifier can perform arbitrary quantum operations and which has access to a quantum channel which can transmit $\tau$ qubits. At any given time, the verifier is not allowed to possess more than $\tau$ qubits. The interaction between the quantum and classical parts of the verifier is the usual one: the classical part controls which operations are to be performed on the quantum register, and outcomes of measurements of the quantum register can be used as input to the classical machine.
\item There is also a classical communication channel between the 
prover and the verifier, which can transmit polynomially many bits at 
any step. 
\item At any given step, either the verifier or the prover perform computations on their registers and send bits and qubits through the relevant channels to the other party.
\end{itemize}
We require:
\begin{itemize}
\item \textbf{Completeness}: if $x\in\mcL$, then after interacting with $\Prov$, $\mathds{V}$ accepts with probability $\ge c$. 
\item \textbf{Soundness}: if $x\notin \mcL$, then the verifier rejects with probability $\ge 1-s$ regardless of the prover $\Prov'$ (who has the same description as $\Prov$) with whom he is interacting. 
\end{itemize} 
\end{deff}
Abusing notation, we denote the class of languages for which such a proof exists also by \QPIP$_{\tau}$.



\section{Function Definitions}\label{sec:functiondefinitions}
The following function definitions are motivated by the function constructions from Learning with Errors (given in Section \ref{sec:lweextended}). Unfortunately, these constructions do not satisfy the ideal function definitions given in the overview (Section \ref{sec:overviewprimitives}); they instead approximately satisfy these definitions. The technicalities of the definitions below arise from the approximate nature of the constructions. For a summary of the differences between the approximate and ideal function definitions, see Section 3 of \cite{oneproverrandomness}. For an intuitive description of the LWE-based function construction, see page 4 of \cite{oneproverrandomness}.

\subsection{Noisy Trapdoor Claw-Free Functions}

\begin{deff}[NTCF Family \cite{oneproverrandomness}]\label{def:trapdoorclawfree}
Let $\lambda$ be a security parameter. Let $\sX$ and $\sY$ be finite sets.
 Let $\mathcal{K}_{\mathcal{F}}$ be a finite set of keys. A family of functions 
$$\mathcal{F} \,=\, \big\{f_{k,b} : \sX\rightarrow \mathcal{D}_{\sY} \big\}_{k\in \mathcal{K}_{\mathcal{F}},b\in\{0,1\}}$$
is called a \textbf{noisy trapdoor claw-free (NTCF) family} if the following conditions hold:

\begin{enumerate}
\item{\textbf{Efficient Function Generation.}} There exists an efficient probabilistic algorithm $\textrm{GEN}_{\mathcal{F}}$ which generates a key $k\in \mathcal{K}_{\mathcal{F}}$ together with a trapdoor $t_k$: 
$$(k,t_k) \leftarrow \textrm{GEN}_{\mathcal{F}}(1^\lambda)\;.$$
\item{\textbf{Trapdoor Injective Pair.}} For all keys $k\in \mathcal{K}_{\mathcal{F}}$ the following conditions hold. 
\begin{enumerate}
\item \textit{Trapdoor}: For all $b\in\{0,1\}$ and $x\neq x' \in \sX$, $\supp(f_{k,b}(x))\cap \supp(f_{k,b}(x')) = \emptyset$. Moreover, there exists an efficient deterministic algorithm $\textrm{INV}_{\mathcal{F}}$ such that for all $b\in \{0,1\}$,  $x\in \sX$ and $y\in \supp(f_{k,b}(x))$, $\textrm{INV}_{\mathcal{F}}(t_k,b,y) = x$. 
\item \textit{Injective pair}: There exists a perfect matching $\sR_k \subseteq \sX \times \sX$ such that $f_{k,0}(x_0) = f_{k,1}(x_1)$ if and only if $(x_0,x_1)\in \sR_k$. \end{enumerate}

\item{\textbf{Efficient Range Superposition.}}
For all keys $k\in \mathcal{K}_{\mathcal{F}}$ and $b\in \{0,1\}$ there exists a function $f'_{k,b}:\sX\mapsto \mathcal{D}_{\sY}$ such that
\begin{enumerate} 
\item For all $(x_0,x_1)\in \mathcal{R}_k$ and $y\in \supp(f'_{k,b}(x_b))$, INV$_{\mathcal{F}}(t_k,b,y) = x_b$ and INV$_{\mathcal{F}}(t_k,b\oplus 1,y) = x_{b\oplus 1}$. 
\item There exists an efficient deterministic procedure CHK$_{\mathcal{F}}$ that, on input $k$, $b\in \{0,1\}$, $x\in \sX$ and $y\in \sY$, returns $1$ if  $y\in \supp(f'_{k,b}(x))$ and $0$ otherwise. Note that CHK$_{\mathcal{F}}$ is not provided the trapdoor $t_k$. 
\item For every $k$ and $b\in\{0,1\}$, 
$$ \Es{x\leftarrow_U \sX} \big[\,H^2(f_{k,b}(x),\,f'_{k,b}(x))\,\big] \,\leq\, \mu(\lambda)\;,$$
 for some negligible function $\mu(\cdot)$. Here $H^2$ is the Hellinger distance; see~\eqref{eq:bhatt}. Moreover, there exists an efficient procedure  SAMP$_{\mathcal{F}}$ that on input $k$ and $b\in\{0,1\}$ prepares the state
\begin{equation}\label{eq:sampledefinitionstate}
    \frac{1}{\sqrt{|\sX|}}\sum_{x\in \sX,y\in \sY}\sqrt{(f'_{k,b}(x))(y)}\ket{x}\ket{y}\;.
\end{equation}
The procedure SAMP$_{\mathcal{F}}$ can be implemented with a unitary operator which acts on $k, b\in\{0,1\}$ and the appropriate number of qubits initialized to $\ket{0}$, and produces the state in \eqref{eq:sampledefinitionstate} along with auxiliary qubits in the state $\ket{0}$.

\end{enumerate}

\item{\textbf{Adaptive Hardcore Bit.}}
For all keys $k\in \mathcal{K}_{\mathcal{F}}$ the following conditions hold, for some integer $w$ that is a polynomially bounded function of $\lambda$. 
\begin{enumerate}
\item For all $b\in \{0,1\}$ and $x\in \sX$, there exists a set $\dset_{k,b,x}\subseteq \{0,1\}^{w}$ such that $\Pr_{d\leftarrow_U \{0,1\}^w}[d\notin \dset_{k,b,x}]$ is negligible, and moreover there exists an efficient algorithm that checks for membership in $\dset_{k,b,x}$ given $k,b,x$ and the trapdoor $t_k$. 
\item There is an efficiently computable injection $\inj:\sX\to \{0,1\}^w$, such that $\inj$ can be inverted efficiently on its range, and such that the following holds. If
\begin{eqnarray*}\label{eq:defsetsH}
H_k &=& \big\{(b,x_b,d,d\cdot(\inj(x_0)\oplus \inj(x_1)))\,|\; b\in \{0,1\},\; (x_0,x_1)\in \mathcal{R}_k,\; d\in \dset_{k,0,x_0}\cap \dset_{k,1,x_1}\big\}\;,\text{\footnotemark}\\
\overline{H}_k &=& \{(b,x_b,d,c)\,|\; (b,x,d,c\oplus 1) \in H_k\big\}\;,
\end{eqnarray*}
\footnotetext{Note that although both $x_0$ and $x_1$ are referred to to define the set $H_k$, only one of them, $x_b$, is explicitly specified in any $4$-tuple that lies in $H_k$.}
then for any quantum polynomial-time procedure $\mathcal{A}$ there exists a negligible function $\mu(\cdot)$ such that 
\begin{equation}\label{eq:adaptive-hardcore}
\Big|\Pr_{(k,t_k)\leftarrow \textrm{GEN}_{\mathcal{F}}(1^{\lambda})}[\mathcal{A}(k) \in H_k] - \Pr_{(k,t_k)\leftarrow \textrm{GEN}_{\mathcal{F}}(1^{\lambda})}[\mathcal{A}(k) \in\overline{H}_k]\Big| \,\leq\, \mu(\lambda)\;.
\end{equation}
\end{enumerate}

\end{enumerate}
\end{deff}


\subsection{Extended Trapdoor Claw-Free Functions}
In this section, we define the extended trapdoor claw-free family we will use in this paper, which is an NTCF family (Definition \ref{def:trapdoorclawfree}) with two additional properties. In order to define an extended trapdoor claw-free family, we must first define a trapdoor injective family. A trapdoor injective family differs from an NTCF family in two ways: the function pairs have disjoint images (rather than perfectly overlapping images) and there is no adaptive hardcore bit condition.

\begin{deff}[Trapdoor Injective Function Family]\label{def:trapdoorinjective}
Let $\lambda$ be a security parameter. Let $\sX$ and $\sY$ be finite sets. Let $\mathcal{K}_{\mathcal{G}}$ be a finite set of keys. A family of functions 
$$\mathcal{G} \,=\, \big\{g_{k,b} : \sX\rightarrow \mathcal{D}_{\sY} \big\}_{b\in\{0,1\},k\in \mathcal{K}_{\mathcal{G}}}$$
is called a \textbf{trapdoor injective family} if the following conditions hold:

\begin{enumerate}
\item{\textbf{Efficient Function Generation.}} There exists an efficient probabilistic algorithm $\textrm{GEN}_{\mathcal{G}}$ which generates a key $k\in \mathcal{K}_{\mathcal{G}}$ together with a trapdoor $t_k$: 
$$(k,t_k) \leftarrow \textrm{GEN}_{\mathcal{G}}(1^\lambda)\;.$$

\item{\textbf{Disjoint Trapdoor Injective Pair.}} For all keys $k\in \mathcal{K}_{\mathcal{G}}$, for all $b, b'\in\{0,1\}$ and $x,x' \in \sX$, if $(b,x)\neq (b',x')$, $\supp(g_{k,b}(x))\cap \supp(g_{k,b'}(x')) = \emptyset$. Moreover, there exists an efficient deterministic algorithm $\textrm{INV}_{\mathcal{F}}$ such that for all $b\in \{0,1\}$,  $x\in \sX$ and $y\in \supp(g_{k,b}(x))$, $\textrm{INV}_{\mathcal{G}}(t_k,y) = (b,x)$.  

\item{\textbf{Efficient Range Superposition.}}
For all keys $k\in \mathcal{K}_{\mathcal{G}}$ and $b\in \{0,1\}$ 
\begin{enumerate} 
 
\item There exists an efficient deterministic procedure CHK$_{\mathcal{G}}$ that, on input $k$, $b\in \{0,1\}$, $x\in \sX$ and $y\in \sY$, outputs $1$ if  $y\in \supp(g_{k,b}(x))$ and $0$ otherwise. Note that CHK$_{\mathcal{G}}$ is not provided the trapdoor $t_k$. 
\item There exists an efficient procedure  SAMP$_{\mathcal{G}}$ that on input $k$ and $b\in\{0,1\}$ returns the state
\begin{equation}
    \frac{1}{\sqrt{|\sX|}}\sum_{x\in \sX,y\in \sY}\sqrt{(g_{k,b}(x))(y)}\ket{x}\ket{y}\;.
\end{equation}

\end{enumerate}

\end{enumerate}
\end{deff}

\begin{deff}[Injective Invariance]\label{def:injectiveinvariant}
A noisy trapdoor claw-free family $\mathcal{F}$ is \textbf{injective invariant} if there exists a trapdoor injective family $\mathcal{G}$ such that:
\begin{enumerate}
\item The algorithms CHK$_{\mathcal{F}}$ and SAMP$_{\mathcal{F}}$ are the same as the algorithms CHK$_{\mathcal{G}}$ and SAMP$_{\mathcal{G}}$.
\item For all quantum polynomial-time procedures $\mathcal{A}$, there exists a negligible function $\mu(\cdot)$ such that
\begin{equation}
\Big|\Pr_{(k,t_k)\leftarrow \textrm{GEN}_{\mathcal{F}}(1^{\lambda})}[\mathcal{A}(k) = 0] - \Pr_{(k,t_k)\leftarrow \textrm{GEN}_{\mathcal{G}}(1^{\lambda})}[\mathcal{A}(k) = 0]\Big|\leq \mu(\lambda)
\end{equation}
\end{enumerate}
\end{deff}

\begin{deff}[Extended Trapdoor Claw-Free Family]\label{def:extendedtrapdoorclawfree}
A noisy trapdoor claw-free family $\mathcal{F}$ is an \textbf{extended trapdoor claw-free family} if:
\begin{enumerate}
\item It is injective invariant.
\item For all $k\in \mathcal{K}_{\mathcal{F}}$, $d\in \{0,1\}^{w}$ and $J,\mathcal{R}_k$ as defined in Definition \ref{def:trapdoorclawfree}, let:
\begin{eqnarray}
H'_{k,d} &=& \{d\cdot (J(x_0)\oplus J(x_1))|  (x_0,x_1)\in\mathcal{R}_k \}
\end{eqnarray}
There exists a string $d\in \{0,1\}^{w}$ such that for all quantum polynomial-time procedures $\mathcal{A}$, there exists a negligible function $\mu(\cdot)$ such that 
\begin{equation}
\Big|\Pr_{(k,t_k)\leftarrow \textrm{GEN}_{\mathcal{F}}(1^{\lambda})}[\mathcal{A}(k) \in H'_{k,d}] - \frac{1}{2}\Big| \leq \mu(\lambda)
\end{equation}

\end{enumerate}
\end{deff}




\section{Measurement Protocol}\label{sec:measprotocol}
We begin by introducing the state commitment process (described in Section \ref{sec:overviewprimitives}) followed by the measurement protocol (given in Section \ref{sec:overviewmeas}). The presentation below is slightly more involved than the overview, since we are not using the perfect trapdoor claw-free/ injective families used in the overview; we are instead using the families given in Definitions \ref{def:trapdoorclawfree}, \ref{def:trapdoorinjective} and \ref{def:extendedtrapdoorclawfree}. After presenting the measurement protocol, we provide notation which will be used throughout the rest of the paper. We conclude this section by proving completeness of our measurement protocol in Lemma \ref{lem:measprotocolcorrectness}.

\subsection{How to Commit}\label{sec:statecommitment}
Here we describe the process of state commitment with an NTCF function (as defined in Definition \ref{def:trapdoorclawfree}). The state commitment process requires a function key $k\in\mathcal{K}_{\mathcal{F}}$ (which corresponds to functions $f_{k,0},f_{k,1}\in\mathcal{F}$) and is performed with respect to the first qubit of an arbitrary state $\ket{\psi}$:
\begin{equation}\label{eq:commitstartingstate}
\ket{\psi} = \sum_{b\in\{0,1\}} \alpha_b\ket{b}\ket{\psi_b}
\end{equation}
The first step of the commitment process is to apply the SAMP$_{\mathcal{F}}$ procedure in superposition, with $k$ and the first qubit containing $b$ as input:
\begin{eqnarray}\label{eq:commitbeforemeas0}
\frac{1}{\sqrt{|\sX|}}\sum_{\substack{b\in\{0,1\} \\x\in \sX,y\in \sY}} \alpha_b\sqrt{f'_{k,b}(x)(y)}\ket{b}\ket{x}\ket{\psi_b}\ket{y}\label{eq:superpositionoverT}
\end{eqnarray}
By condition 3(c) of Definition \ref{def:trapdoorclawfree} and Lemma \ref{lem:hellingertotrace} \eqref{eq:commitbeforemeas0} is within negligible trace distance of the following state:
\begin{eqnarray}\label{eq:commitbeforemeas}
\frac{1}{\sqrt{|\sX|}}\sum_{\substack{b\in\{0,1\} \\x\in \sX,y\in \sY}} \alpha_b\sqrt{f_{k,b}(x)(y)}\ket{b}\ket{x}\ket{\psi_b}\ket{y}
\end{eqnarray}
The second step of the commitment process is to measure the last register, obtaining the \textit{commitment string} $y\in \sY$. Let $x_{b,y} = \textrm{INV}_{\mathcal{F}}(t_k,b,y)$ ($t_k$ is the trapdoor corresponding to the key $k$). The remaining state at this point is within negligible trace distance of the following state
\begin{equation}\label{eq:finalcommitstate}
\sum_{b\in\{0,1\}} \alpha_b\ket{b}\ket{x_{b,y}} \ket{\psi_b}
\end{equation}
The fact that the superposition collapses in this manner is due to both the trapdoor and injective pair conditions in Definition \ref{def:trapdoorclawfree}. The trapdoor condition implies that for each $b$, there can be at most one remaining element $x\in \sX$ in the superposition after measuring $y$. The injective pair condition states that for all $x_{b,y}\in \sX$, there exists exactly one $x_{b\oplus 1,y}\in \sX$ such that $(x_{0,y},x_{1,y})\in\mathcal{R}_k$ (i.e. $f_{k,b}(x_{b,y}) = f_{k,b}(x_{b\oplus 1,y})$). Therefore, if $y\in \supp(f_{k,b}(x_{b,y}))$, it follows that $y\in \supp(f_{k,b\oplus 1}(x_{b\oplus 1,y}))$. We will call the first qubit of \eqref{eq:finalcommitstate} the \textit{committed qubit} and the second register (containing $x_{b,y}$) the \textit{preimage register}. 

\subsubsection{Hadamard Measurement of a Committed State}\label{sec:measofcommitted}

The first step in measuring a committed state in the Hadamard basis is to apply the unitary $U_{\inj}$, which uses the injective map $\inj$ defined in condition 4(b) of Definition \ref{def:trapdoorclawfree}:
\begin{equation}\label{eq:defuinj}
    U_{\inj}(\sum_{b\in\{0,1\}} \alpha_b\ket{b}\ket{x_{b,y}} \ket{\psi_b}\ket{0}^{e_{\inj}}) =  \sum_{b\in\{0,1\}} \alpha_b\ket{b}\ket{\inj(x_{b,y})} \ket{\psi_b}\ket{0}^{e'_{\inj}}
\end{equation}
where the number of auxiliary qubits ($e_{\inj}$ and $e'_{\inj}$) is determined by the map $\inj$. The map $U_{\inj}$ is unitary since $\inj$ is both efficiently computable and efficiently invertible. The second step is to apply the Hadamard transform $H^{\otimes w + 1}$ to the first two registers of the state in \eqref{eq:defuinj}. The resulting state is:
\begin{eqnarray}
\frac{1}{\sqrt{2^w}}\sum_{\substack{d\in\{0,1\}^w\\ b\in \{0,1\}}} \alpha_bX^{d\cdot(\inj(x_{0,y})\oplus \inj(x_{1,y}))}H \ket{b} \otimes (-1)^{d\cdot \inj(x_{0,y})}\ket{d} \otimes \ket{\psi_b}\otimes \ket{0}^{e'_{\inj}}
\end{eqnarray}
The third step is measurement of the preimage register, obtaining a string $d\in\{0,1\}^w$ and resulting in the following state (recall the state $\ket{\psi}$ from \eqref{eq:commitstartingstate}):
\begin{eqnarray}
(X^{d\cdot(\inj(x_{0,y})\oplus \inj(x_{1,y}))}H\otimes \mcI)\ket{\psi} \ket{0}^{e'_{\inj}}
\end{eqnarray}
The final step is measuring the committed qubit to obtain a bit $b'$. The Hadamard measurement result of the first qubit of $\ket{\psi}$ is $b'\oplus d\cdot (\inj(x_{0,y})\oplus \inj(x_{1,y}))\in \{0,1\}$ ($x_{0,y}$ and $x_{1,y}$ can be recovered from $y$ using the trapdoor $t_k$ and the function INV$_{\mathcal{F}}$). 

\subsubsection{How to Commit Using a Trapdoor Injective Family}\label{sec:injectivecommit}
The commitment process described in Section \ref{sec:statecommitment} can also be performed using a key $k\in\mathcal{K}_{\mathcal{G}}$ corresponding to trapdoor injective functions $g_{k,0}, g_{k,1}\in \mathcal{G}$ (see Definition \ref{def:trapdoorinjective}). Prior to measuring $y$ (at the stage of \eqref{eq:commitbeforemeas0}), the state is:
\begin{eqnarray}\label{eq:commitbeforemeasinjective}
\frac{1}{\sqrt{|\sX|}}\sum_{\substack{b\in\{0,1\} \\x\in \sX,y\in \sY}} \alpha_b\sqrt{g_{k,b}(x)(y)}\ket{b}\ket{x}\ket{\psi_b}\ket{y}
\end{eqnarray}
Now the last register is measured to obtain $y\in\sY$. Since the sets $\supp(g_{k,b}(x))$ and $\supp(g_{k,b'}(x'))$ are disjoint for all $(b,x)\neq (b',x')$ (see the trapdoor condition of Definition \ref{def:trapdoorinjective}), $y\in \mathop{\cup}\limits_{x\in \sX}\supp(g_{k,b}(x))$ with probability $|\alpha_b|^2$. Let $(b,x_{b,y}) =\textrm{INV}_{\mathcal{G}}(t_k,y)$. The remaining state after measurement is:
\begin{equation}
\ket{b}\ket{x_{b,y}}\ket{\psi_b}
\end{equation}
Therefore, measuring $y$ acts as a standard basis measurement of the first qubit of the state $\ket{\psi}$ in \eqref{eq:commitstartingstate}. The standard basis measurement $b$ can be obtained (with access to only the trapdoor $t_k$ of $g_{k,0},g_{k,1}$ and $y$) by running the function INV$_{\mathcal{G}}$.

\subsection{Measurement Protocol}\label{sec:lwemeasprotocol}
We now use the commitment process in Section \ref{sec:statecommitment} to construct our measurement protocol for $n$ qubits, where $n$ is polynomial in the security parameter $\lambda$. We require an extended trapdoor claw-free family $\mathcal{F}$ (Definition \ref{def:extendedtrapdoorclawfree}) as well as its corresponding trapdoor injective family $\mathcal{G}$ (Definition \ref{def:trapdoorinjective}). The measurement protocol depends on a string $h\in \{0,1\}^n$, called the basis choice, which represents the basis for which the verifier would like measurement results of the $n$ qubits; $h_i = 0$ indicates the standard basis and $h_i = 1$ indicates the Hadamard basis. We now provide the measurement protocol:
\begin{protocol}{\textbf{Measurement Protocol (for $h\in \{0,1\}^n$)}}\label{prot:measprotocol}
\begin{enumerate}
\item The verifier performs the following operations for $1\leq i\leq n$:
\begin{enumerate}

\item If the standard basis is chosen for qubit $i$ (i.e. $h_i = 0$), the verifier runs GEN$_{\mathcal{G}}(1^{\lambda})$ to produce a function key $k_i\in \mathcal{K}_{\mathcal{G}}$ and its corresponding trapdoor $t_{k_i}$.

\item If the Hadamard basis is chosen for qubit $i$ ($h_i$ = 1), the verifier runs GEN$_{\mathcal{F}}(1^{\lambda})$ to produce a function key $k_i\in \mathcal{K}_{\mathcal{F}}$ and its corresponding trapdoor $t_{k_i}$.
\end{enumerate}
Let $k' = (k_1,\ldots,k_n)$. The verifier sends the function choice $k'$ to the prover.
\item The verifier receives $y' = (y_1,\ldots,y_n)\in \sY^n $ from the prover.
\item The verifier chooses at random to run a test round or a Hadamard round (each is chosen with probability $\frac{1}{2}$).
\item For a test round: 
\begin{enumerate}
\item The verifier asks the prover for standard basis measurements of committed qubit $i$ and preimage register $i$, for $1 \leq i \leq n$. 

\item For $1\leq i\leq n$, the verifier receives a bit $b'_i$ and a string $x'_i\in \sX$. For all $i$ such that $h_i = 0$, the verifier rejects if CHK$_{\mathcal{G}}(k_i,b'_i,x'_i,y_i) = 0$. For all $i$ such that $h_i = 1$, the verifier rejects if CHK$_{\mathcal{F}}(k_i,b'_i,x'_i,y_i) = 0$. 

\end{enumerate}
\item For a Hadamard round:
\begin{enumerate}

\item The verifier asks the prover for Hadamard measurements of committed qubit $i$ and preimage register $i$ for $1\leq i \leq n$. 

\item For $1\leq i\leq n$, the verifier receives a bit $b'_i$ and a string $d_i\in \{0,1\}^w$. 

\item For qubits $i$ for which $h_i$ = 0, the results $(b'_i,d_i)$ are ignored. The verifier computes
\begin{eqnarray}
 (m_i,x_{m_i,y_i}) &=& \textrm{INV}_{\mathcal{G}}(t_{k_i},y_i)
\end{eqnarray} 
If the inverse does not exist, the verifier stores a random bit as the measurement result and rejects. Otherwise the verifier stores $m_i$ as the standard basis measurement result. 

\item For qubits $i$ for which $h_i = 1$, the verifier computes
\begin{eqnarray}
x_{0,y_i} &=& \textrm{INV}_{\mathcal{F}}(t_{k_i},0,y_i)\\
x_{1,y_i} &=&  \textrm{INV}_{\mathcal{F}}(t_{k_i},1,y_i) 
\end{eqnarray} 
If either of the inverses does not exist, the verifier stores a random bit as the measurement result and rejects. The verifier uses $t_{k_i}$ to check if $d_i \in \dset_{k_i,0,x_{0,y_i}}\cap \dset_{k_i,1,x_{1,y_i}}$. If not, the verifier stores a random bit as the measurement result and rejects. Otherwise, the verifier stores $m_i = b'_i\oplus d_i\cdot (\inj(x_{0,y_i})\oplus \inj(x_{1,y_i}))$ as the Hadamard basis measurement result. 

\end{enumerate}
\end{enumerate}
\end{protocol}

\subsubsection{Honest Prover}
We now provide an honest prover's behavior in Protocol \ref{prot:measprotocol}, assuming the prover would like to report measurement results of an $n$ qubit state $\rho$:
\begin{protocol}{\textbf{Honest Prover in Measurement Protocol (for an efficiently computable $n$ qubit state $\rho$)}}\label{prot:measprotocolprover}
\begin{enumerate}
\item The prover creates the state $\rho$. Upon receipt of $k'$ from the verifier, the prover commits to qubit $i$ of $\rho$ using $k_i$ as described in Section \ref{sec:statecommitment}. The prover reports the measurement results $y' = (y_1,\ldots,y_n) \in \sY^n$ obtained from each commitment process to the verifier. 
\item For a test round: 
\begin{enumerate}
\item The prover measures each of the $n$ committed qubits and preimage registers in the standard basis, sending the verifier the resulting bit $b_i'$ and string $x_i'\in \sX$ for $1\leq i\leq n$.
\end{enumerate}
\item For a Hadamard round:
\begin{enumerate}
\item The prover first applies the unitary $U_{\inj}$ to all $n$ preimage registers. The prover then measures each of the $n$ committed qubits and preimage registers in the Hadamard basis, sending the verifier the resulting bit $b_i'$ and string $d_i\in \{0,1\}^w$ for $1\leq i\leq n$.
\end{enumerate}
\end{enumerate}

\end{protocol}

\subsection{Notation}\label{sec:notation}
We now introduce some notation (and provide reminders of previously used notation) which will be useful throughout the rest of the paper.
\begin{enumerate}
\item The string $h\in\{0,1\}^n$ is called the basis choice; $h_i = 0$ indicates the standard basis and $h_i = 1$ indicates the Hadamard basis.
\item We will call $k'$ (produced in step 1 of Protocol \ref{prot:measprotocol}) the \textit{function choice} of the verifier. Let $D_{\Ver,h}$ be the distribution which it is sampled from (this is the distribution produced by GEN$_{\mathcal{F}}$ and GEN$_{\mathcal{G}}$). 
\item A \textit{perfect} prover is a prover who is always accepted by the verifier on the test round. 
\item For a density matrix $\rho$ on $n$ qubits and a string $h\in \{0,1\}^n$, let $D_{\rho,h}$ be the distribution over $\{0,1\}^n$ which results from measuring all qubits of $\rho$ in the basis specified by $h$. 
\item For every prover $\Prov$ and basis choice $h\in \{0,1\}^n$, let $D_{\Prov,h}$ be the distribution over measurement results $m\in\{0,1\}^n$ obtained by the verifier when interacting with $\Prov$ on basis choice $h$ in a Hadamard round. Let $D_{\Prov,h}^C$ be the same distribution, but conditioned on acceptance by the verifier (in a Hadamard round). Let $\sigma_{\Prov,h}$ be the density matrix corresponding to the distribution $D_{\Prov,h}$:
\begin{eqnarray}
\sigma_{\Prov,h} &\EqDef& \sum_{m\in \{0,1\}^n}D_{\Prov,h}(m)\ket{m}\bra{m}\label{eq:verifierstate}
\end{eqnarray}
We will frequently use the fact that for provers $\Prov$ and $\Prov'$, $\sigma_{\Prov,h}$ and $\sigma_{\Prov',h}$ are computationally indistinguishable if and only if $D_{\Prov,h}$ and $D_{\Prov',h}$ are computationally indistinguishable, by definition of computational indistinguishability of distributions (Definition \ref{def:compinddist}) and of density matrices (Definition \ref{def:compind}). Also note that by definition of trace distance in \eqref{eq:deftracedistance} and total variation distance in \eqref{eq:deftotalvariation}:
\begin{eqnarray}\label{eq:equivalenceofdistandstate}
\T{\sigma_{\Prov,h}}{\sigma_{\Prov',h}} = \TV{D_{\Prov,h}}{ D_{\Prov',h}}
\end{eqnarray}

\item As introduced in Section \ref{sec:statecommitment}, a \textit{committed qubit} is the qubit which is used to determine whether to apply $f_{k,0}$ or $f_{k,1}$ (or $g_{k,0}$ or $g_{k,1}$) and the \textit{preimage register} is the register which contains the inverse after the measurement; e.g. in the following state from \eqref{eq:finalcommitstate}:
\begin{equation}
\sum_b\alpha_b\ket{b}\ket{x_{b,y}}\ket{\psi_b}
\end{equation}
the first qubit is the committed qubit and the second register (containing $x_{b,y}$) is the preimage register. The \textit{commitment string} is the string $y\in \sY$.
\end{enumerate}

\subsection{Completeness of Measurement Protocol}\label{sec:measprotocolcorrectness}
\begin{lem}\label{lem:measprotocolcorrectness}{\textbf{Completeness of Measurement Protocol (Protocol \ref{prot:measprotocol})}}
For all $n$ qubit states $\rho$ and for all basis choices $h\in \{0,1\}^n$, the prover $\Prov$ described in Protocol \ref{prot:measprotocolprover} is a perfect prover ($\Prov$ is accepted by the verifier in a test round for basis choice $h$ with perfect probability). There exists a negligible function $\mu$ such that in the Hadamard round for basis choice $h$, the verifier accepts $\Prov$ with probability $\geq 1 - \mu$ and $\TV{D_{\Prov,h}^C}{D_{\rho,h}}\leq \mu$. 
\end{lem}
\begin{proofof}{\Le{lem:measprotocolcorrectness}}
First, assume that the prover could produce the ideal states in the commitment procedure, as written in \eqref{eq:commitbeforemeas} for the Hadamard basis and \eqref{eq:commitbeforemeasinjective} for the standard basis. Call such a prover $\Prov'$. The distribution over measurement results obtained by the verifier when interacting with $\Prov'$ (prior to conditioning on acceptance) is equal to the distribution over measurement results obtained by measuring $\rho$ in the basis specified by $h$, i.e.:  
\begin{equation}\label{eq:idealprovercomputationtrace}
    \TV{D_{\Prov',h}}{D_{\rho,h}} = 0
\end{equation} 
We now return to analyzing the prover $\Prov$ given in Protocol \ref{prot:measprotocolprover}. First note that $\Prov$ is a perfect prover: when measured, the superpositions created by $\Prov$ during the commitment process (in \eqref{eq:superpositionoverT} and \eqref{eq:commitbeforemeasinjective}) pass the CHK procedure perfectly, by definition (see Definition \ref{def:trapdoorclawfree} and Definition \ref{def:trapdoorinjective}). Moving on to the Hadamard round, $\Prov$ is rejected by the verifier only if there exists an $i$ such that the measurement result $d_i$ is not in the set $\dset_{k_i,0,x_{0,y_i}}\cap \dset_{k_i,1,x_{1,y_i}}$ (and $h_i = 1$). The adaptive hardcore bit clause of Definition \ref{def:trapdoorclawfree} (item 4(a)) implies that since $d_i$ is sampled uniformly, there exists a negligible function $\mu_H$ such that the probability that the verifier rejects $\Prov$ in the Hadamard round is at most $\mu_H$:
\begin{equation}\label{eq:honestproveracceptance}
    \TV{D_{\Prov,h}}{D^C_{\Prov,h}} \leq \mu_H
\end{equation} 
Next, observe that the prover $\Prov$ in Protocol \ref{prot:measprotocolprover} can produce the state in \eqref{eq:commitbeforemeasinjective} (which is used by the prover $\Prov'$), but can only create a state within negligible trace distance of the state in \eqref{eq:commitbeforemeas}. It follows that there exists a negligible function $\mu'$ such that:
\begin{equation}\label{eq:tracedistanceidealtohonest}
\TV{D_{\Prov,h}}{D_{\Prov',h}}   \leq \mu' 
\end{equation}
Using the triangle inequality, we obtain: 
\begin{eqnarray}
\TV{D_{\Prov,h}^C}{D_{\rho,h}} &\leq&  \TV{D_{\Prov,h}^C}{ D_{\Prov,h}} + \TV{D_{\Prov,h}}{D_{\Prov',h}} + \TV{D_{\Prov',h}}{D_{\rho,h}}
\end{eqnarray}
We complete the calculation by plugging in \eqref{eq:idealprovercomputationtrace}, \eqref{eq:honestproveracceptance} and \eqref{eq:tracedistanceidealtohonest}:
\begin{eqnarray}
\TV{D_{\Prov,h}^C}{D_{\rho,h}}  &\leq& \mu_H + \mu'
\end{eqnarray}
To complete the proof of the claim, set $\mu = \mu_H + \mu'$.
\end{proofof}

\section{Measurement Protocol Soundness}\label{sec:measprotocolsoundnessoverall}

 This section covers the soundness of the measurement protocol. We begin, in Section \ref{sec:proverbehavior}, by characterizing the behavior of a general prover, as in Section \ref{sec:overviewproverbehavior} of the overview. In Section \ref{sec:underlyingstate}, we show that if this characterization satisfies a certain requirement (i.e. the prover is \textit{trivial}), then the measurement results provided by the prover are actually consistent with an underlying quantum state, therefore satisfying the soundness requirement of a measurement protocol (this section corresponds to Section \ref{sec:overviewstateexistenceproof} of the overview). Section \ref{sec:generaltotrivialhadamard} deals with the crux of the soundness argument (as outlined in Section \ref{sec:overviewreductiontotrivial} of the overview): showing that a general attack of a prover can be replaced with an attack which commutes with standard basis measurement (thus representing a trivial prover). Finally, Section \ref{sec:measprotocolsoundness} provides a proof of soundness by combining the results from Sections \ref{sec:proverbehavior} to \ref{sec:generaltotrivialhadamard}.

\subsection{Prover Behavior}\label{sec:proverbehavior}
We now give a claim which characterizes the behavior of a general prover in Protocol \ref{prot:measprotocol} (the overview of this claim and its proof are given in Section \ref{sec:overviewproverbehavior}). The only difference between the following claim and the version given in the overview is the inclusion of the operator $U_{\inj}$ (as defined in \eqref{eq:defuinj}): 
\begin{claim}{\textbf{Prover Behavior}}\label{cl:generalprover}
For all \BQP\ provers $\Prov$ in Protocol \ref{prot:measprotocol}, there exist two efficiently computable unitary operators $U_0,U$ and a prover $\Prov'$ (described below) such that for all basis choices $h\in \{0,1\}^n$, $\Prov$ and $\Prov'$ are accepted by the verifier with the same probability in a test round and the distribution over measurement results $D_{\Prov,h}$  produced by the prover and verifier as a result of Protocol \ref{prot:measprotocol} is equal to the distribution $D_{\Prov',h}$ corresponding to the prover $\Prov'$. We say that $\Prov$ is characterized by $(U_0,U)$.  
\begin{enumerate}
\item $\Prov'$ designates his first $n$ qubits as committed qubits, the next $n$ registers as preimage registers and the final $n$ registers as commitment string registers. All other registers contain auxiliary space.
\item Upon receipt of the function choice $k'$ from the verifier, the prover $\Prov'$ applies $U_0$ to his initial state:
\begin{equation}
\ket{0}^{\otimes e}\otimes \ket{k'}
\end{equation}
where $e$ is determined by $(U_0, U)$ and $U_0$ uses the last register (containing $k'$) as a control register; i.e. there exists a unitary $U_{0,k'}$ such that
\begin{equation}\label{eq:u0trivialcontrol}
U_0(\ket{0}^{\otimes e}\otimes \ket{k'}) = U_{0,k'}(\ket{0}^{\otimes e}) \otimes \ket{k'}
\end{equation}
\item $\Prov'$ measures the  commitment string registers to obtain $y' = (y_1,\ldots,y_n)\in \sY^n$, which is sent to the verifier. 
\item For a Hadamard round:
\begin{enumerate}
    \item $\Prov'$ appends $e_{\inj}\cdot n$ auxiliary 0 qubits to his state and applies the unitary $U_{\inj}$ to all $n$ preimage registers, followed by application of the unitary $U$ to his entire state. 
    \item $\Prov'$ measures the $n$ committed qubits and preimage registers in the Hadamard basis. $\Prov'$ sends the verifier the resulting bit $b_i'$ and string $d_i\in \{0,1\}^w$ for $1\leq i\leq n$.
    \end{enumerate}
\item For a test round, $\Prov'$ measures each of the $n$ committed qubits and preimage registers in the standard basis, sending the verifier the resulting bit $b_i'$ and string $x_i'\in \sX$ for $1\leq i\leq n$.  

\end{enumerate}
\end{claim}
\begin{notation}\label{notation:superoperator}
We will also frequently say that a prover $\Prov$ is characterized by two CPTP maps $(\mathcal{S}_0,\mathcal{S})$. This means that for all basis choices $h\in \{0,1\}^n$, $\Prov$ and $\Prov'$ are accepted with the same probability in a test round and $D_{\Prov,h} = D_{\Prov',h}$, and the prover $\Prov'$ follows steps 1 - 5 in Claim \ref{cl:generalprover}, but uses the CPTP maps $\mathcal{S}_0,\mathcal{S}$ rather than the unitary operators $U_0,U$. 
\end{notation}
\begin{proofof}{\Cl{cl:generalprover}}
We will follow the principle given in Section \ref{sec:overviewproverbehavior}: a general prover is equivalent from the verifier's perspective to a prover $\Prov$ who begins each round by applying an arbitrary unitary attack and then behaves honestly. The first implication of the principle is that $\Prov$ measures the same registers as an honest prover; therefore, like the honest prover, $\Prov$ designates the first $n$ qubits as committed qubits, the next $n$ registers as preimage registers, and the final $n$ registers as commitment string registers. All other registers of $\Prov$ contain the auxiliary space. 

The second implication is that there exist unitary operators $U',U_T$ and $U_C$ such that $\Prov$ acts as follows. $\Prov$ begins with the initial state $\ket{0}^{\otimes e}\otimes \ket{k'}$ and then applies a unitary operator $U'$ to his state, followed by standard basis measurement of the commitment string registers to obtain $y'$. If the verifier chooses a test round, the prover applies another unitary $U_T$ followed by standard basis measurements of the committed qubit and preimage registers to obtain the requested measurement results. If the verifier chooses a Hadamard round, the prover first appends $e_{\inj}\cdot n$ auxiliary 0 qubits to his state. Next, the prover applies a unitary $U_C$ to his state. He then applies the unitary $U_{\inj}$ (see Section \ref{sec:measofcommitted}) to all $n$ preimage registers. Finally, the prover measures all $n$ committed qubits and preimage registers in the Hadamard basis to obtain the requested measurement results. We can assume both $U_T$ and $U_C$ do not act on the register containing $y'$. This is because $y'$ could have been copied into the prover's auxiliary space prior to measurement, and $U_T$ and $U_C$ can instead act on this space. It follows that both $U_T$ and $U_C$ commute with the measurement of $y'$. 

To obtain the attacks $U_0$ and $U$ which characterize $\Prov$, we make two changes. First, we use the fact that $U_T$ commutes with measurement of $y'$ to shift it prior to the measurement. Due to this change, we also need to append $U_T^\dagger$ to the start of the Hadamard round attack. Our second change is to shift the unitary $U_{\inj}$ so that it is prior to the Hadamard round attack; this can be done by conjugating the attack by $U_{\inj}$. It follows that if we let $U_0 = U_TU'$, $U = U_{\inj}^{\otimes n}U_CU_T^\dagger (U_{\inj}^{\otimes n})^\dagger$ and consider the prover $\Prov'$ described in the statement of Claim \ref{cl:generalprover} (with respect to $U_0$ and $U$), $\Prov$ and $\Prov'$ are accepted with the same probability in a test round and $D_{\Prov,h} = D_{\Prov',h}$ for all basis choices $h$.  

\end{proofof}

\subsection{Construction of Underlying Quantum State}\label{sec:underlyingstate}
We require the following definition: 
\begin{deff}{\textbf{Trivial Prover}}\label{def:trivial}
A perfect prover $\Prov$ in Protocol \ref{prot:measprotocol} characterized by $(U_0,\mathcal{S})$ (where $U_0$ is a unitary, $\mathcal{S}$ is a CPTP map and both are efficiently computable) is called trivial if $\mathcal{S}$ commutes with standard basis measurement on the first $n$ qubits.
\end{deff}
In this section, we prove that, for trivial provers, the distribution over measurement results produced in Protocol \ref{prot:measprotocol} actually corresponds to the measurement of a quantum state; in other words, there exists a quantum state underlying the measurement results (the overview of this claim is given in Section \ref{sec:overviewstateexistenceproof}):
\begin{lem}\label{lem:trivialprover}
For all trivial provers $\Prov$, there exists an $n$ qubit state $\rho$ (which can be created using a \BQP\ circuit) such that for all $h\in\{0,1\}^n$, the distribution over measurement results $D_{\Prov,h}$ produced in Protocol \ref{prot:measprotocol} with respect to $\Prov$ for basis choice $h$ is computationally indistinguishable from the distribution $D_{\rho,h}$ which results from measuring $\rho$ in the basis determined by $h$.  
\end{lem}

\begin{proof}
For a unitary $U_0$ and CPTP map $\mathcal{S}$, let the prover $\Prov$ be characterized by $(U_0,\mathcal{S})$. The state $\rho$ is constructed as follows: 
\begin{protocol}{\textbf{Construction of $\rho$ corresponding to $\Prov$}}\label{prot:rhoconstruction}
\begin{enumerate}
    \item For $1\leq i\leq n$: sample $(k_i, t_{k_i})\leftarrow \textrm{GEN}_{\mathcal{F}}(1^{\lambda})$.
    \item Follow steps 1-4(a) in Claim \ref{cl:generalprover} (with respect to $U_0, \mathcal{S}$).
    \item Measure all preimage registers in the Hadamard basis to obtain $d_1,\ldots,d_n\in\{0,1\}^w$.
    \item For $1\leq i\leq n$, use the trapdoor $t_{k_i}$ to apply $Z^{d_i\cdot (x_{0,y_i}\oplus x_{1,y_i})}$ to the $i^{th}$ committed qubit.
    \item Trace out all qubits except $n$ committed qubits. 
\end{enumerate}
\end{protocol}
We now argue that, for all $h\in\{0,1\}^n$, $D_{\rho,h}$ is computationally indistinguishable from $D_{\Prov,h}$. We will proceed through two families of hybrid states which are dependent on the basis choice $h$. In the first family $\{\rho_h^{(1)}\}_{h\in\{0,1\}^n}$, we simply remove the $Z$ decoding operator (step 4 of Protocol \ref{prot:rhoconstruction}) if $h_i = 0$. This also eliminates the need for the trapdoor $t_{k_i}$ if $h_i = 0$:
\begin{protocol}{\textbf{Construction of $\rho_h^{(1)}$ corresponding to $\Prov$}}\label{prot:rhoconstruction1}
\begin{enumerate}
    \item For $1\leq i\leq n$: sample $(k_i, t_{k_i})\leftarrow \textrm{GEN}_{\mathcal{F}}(1^{\lambda})$. If $h_i = 0$, discard the trapdoor $t_{k_i}$. 
    \item Apply steps 2-3 of Protocol \ref{prot:rhoconstruction}.
    \item For $1\leq i\leq n$, if $h_i = 1$, use the trapdoor $t_{k_i}$ to apply $Z^{d_i\cdot (x_{0,y_i}\oplus x_{1,y_i})}$ to the $i^{th}$ committed qubit.
    \item Trace out all qubits except the $n$ committed qubits. 
\end{enumerate}
\end{protocol}
The distributions $D_{\rho,h}$ and $D_{\rho_h^{(1)},h}$ differ only on $i$ for which $h_i = 0$. To address this difference, note that if $h_i = 0$, the $Z$ operator applied in step 4 of Protocol \ref{prot:rhoconstruction} has no effect on $D_{\rho,h}$: to obtain $D_{\rho,h}$ the $i^{th}$ committed qubit is measured in the standard basis immediately after application of the $Z$ operator. Therefore, $D_{\rho,h} = D_{\rho_h^{(1)},h}$ for all $h$.

Our next hybrid is:
\begin{protocol}{\textbf{Construction of $\rho_h^{(2)}$ corresponding to $\Prov$}}\label{prot:rhoconstruction2}
\begin{enumerate}
    \item For $1\leq i\leq n$: if $h_i = 1$, sample $(k_i, t_{k_i})\leftarrow \textrm{GEN}_{\mathcal{F}}(1^{\lambda})$. If $h_i = 0$, sample $(k_i, t_{k_i})\leftarrow \textrm{GEN}_{\mathcal{G}}(1^{\lambda})$ and discard the trapdoor $t_{k_i}$.
    \item Apply steps 2-4 of Protocol \ref{prot:rhoconstruction1}.
\end{enumerate}
\end{protocol}
The computational indistinguishability of $D_{\rho_h^{(1)},h}$ and $D_{\rho_h^{(2)},h}$ follows due to the injective invariance (Definition \ref{def:injectiveinvariant}) of $\mathcal{F}$ with respect to $\mathcal{G}$: as long as the trapdoor $t_{k_i}$ is unknown, a key $k_i$ sampled from $\mathcal{K}_{\mathcal{F}}$ is computationally indistinguishable from a key $k_i$ sampled from $\mathcal{K}_{\mathcal{G}}$. We can apply this argument for all $i$ such that $h_i = 0$ since the trapdoor $t_{k_i}$ was discarded for all such $i$ in Protocols \ref{prot:rhoconstruction1} and \ref{prot:rhoconstruction2}.

We have so far shown that $D_{\rho,h}$ is computationally indistinguishable from $D_{\rho_h^{(2)},h}$ for all $h\in\{0,1\}^n$. To complete our proof, we now show that $D_{\rho_h^{(2)},h} = D_{\Prov,h}$. The two distributions differ as follows: if $h_i = 0$, the distribution of the $i^{th}$ bit of $D_{\Prov,h}$ is obtained from the commitment string $y_i$ (see step 5(c) of Protocol \ref{prot:measprotocol}), but the distribution  of the $i^{th}$ bit of $D_{\rho_h^{(2)},h}$ is obtained from measuring the $i^{th}$ committed qubit of $\rho_h^{(2)}$ in the standard basis. 

To see that these two distributions are equal, we begin by observing that since the prover $\Prov$ is perfect, if $h_i = 0$, measuring the $i^{th}$ committed qubit prior to the attack $\mathcal{S}$ (i.e. at the start of the Hadamard round) results in the same outcome as extracting the measurement outcome from $y_i$. To complete our proof, recall that since the prover is trivial, the attack $\mathcal{S}$ commutes with standard basis measurement.  
\end{proof}




\subsection{Replacement of a General Attack with an X-Trivial Attack for Hadamard Basis}\label{sec:generaltotrivialhadamard}
In the previous section (Section \ref{sec:underlyingstate}), we showed that for trivial provers, there exists a quantum state underlying the measurement results produced in Protocol \ref{prot:measprotocol}; in other words, such provers satisfy the soundness condition of a measurement protocol. To prove soundness, we must prove the same result, i.e. the existence of an underlying quantum state, for all perfect provers (provers who passes the test round of Protocol \ref{prot:measprotocol} with perfect probability). To do so, we will show that for every perfect prover, there exists a trivial prover who produces computationally indistinguishable measurement results in Protocol \ref{prot:measprotocol}. This statement, described in Section \ref{sec:overviewreductiontotrivial} of the overview and stated formally in in Section \ref{sec:measprotocolsoundness}, Claim \ref{cl:perfecttotrivialprover}, is quite easy to see for standard basis measurement results, i.e. when $h = 0$ (this will also be shown in Section \ref{sec:measprotocolsoundness}, Claim \ref{cl:reductiontosinglestandard}). The difficulty lies in proving indistinguishablility in the case that $h = 1$, which is stated in the following lemma and will be the focus of this section: 
\begin{lem}{\textbf{General to $X$-Trivial Attack for Hadamard Basis}}\label{lem:singlesecurity}
Let $1\leq j\leq n$. Let $\mathcal{S} = \{B_{\tau}\}_{\tau}$ and $\mathcal{S}_j = \{B'_{j,x,\tau}\}_{x\in \{0,1\},\tau}$ be CPTP maps written in terms of their Kraus operators: 
\begin{equation}
B_{\tau} = \sum_{x,z\in\{0,1\}} X^xZ^z\otimes B_{jxz\tau}
\end{equation}
\begin{equation}
B'_{j,x,\tau} = \sum_{z\in\{0,1\}} Z^z\otimes B_{jxz\tau}
\end{equation}
where $B_{\tau}$ and $B'_{j,x,\tau}$ have been rearranged so that $X^xZ^z$ and $Z^z$ act on the $j^{th}$ qubit. For a unitary operator $U_0$, let $\Prov$ be a perfect prover characterized by $(U_0,\mathcal{S})$ (see Claim \ref{cl:generalprover} and notation \ref{notation:superoperator}). Let $\Prov_j$ be a perfect prover characterized by $(U_0,\mathcal{S}_j)$. If $h_j = 1$, $D_{\Prov,h}$ and $D_{\Prov_j,h}$ are computationally indistinguishable. 
\end{lem}
The overview of the proof of Lemma \ref{lem:singlesecurity} is given in Section \ref{sec:overviewreductiontotrivial}. Lemma \ref{lem:singlesecurity} is slightly more general than the statement in the overview: $n$ does not have to be equal to 1, and we are proving that we can replace the attack $\mathcal{S}$ with an attack which acts trivially on any one of the committed qubits $j$ for which $h_j = 1$. We begin by writing out the state $\sigma_{\Prov,h}$ which corresponds to the distribution $D_{\Prov,h}$  (as defined in \eqref{eq:verifierstate}). This requires some care, since we need to go through the steps of Protocol \ref{prot:measprotocol} in order to construct $\sigma_{\Prov,h}$. Once we write down the state $\sigma_{\Prov,h}$, we can proceed to proving computational indistinguishability between $\sigma_{\Prov,h}$ and $\sigma_{\Prov_j,h}$. As written below \eqref{eq:verifierstate}, proving computational indistinguishability between $\sigma_{\Prov,h}$ and $\sigma_{\Prov_j,h}$ is equivalent to proving indistinguishability between $D_{\Prov,h}$ and $D_{\Prov_j,h}$.

\begin{proofof}{\textbf{Lemma \ref{lem:singlesecurity}}} 
 We will assume for convenience that $j = 1$; the proof for all other values of $j$ is equivalent. To analyze the state $\sigma_{\Prov,h}$, we will assume that $\Prov$ follows steps 1-4 in Claim \ref{cl:generalprover}. We can do this since $\Prov$ is characterized by $U_0,\mathcal{S}$; therefore, the state $\sigma_{\Prov,h}$ can be obtained by following the steps in Claim \ref{cl:generalprover}.

We first provide some notation we will require. Let $k = k_1\in\mathcal{K}_{\mathcal{F}}$ be the first function key received by the prover $\Prov$ in Protocol \ref{prot:measprotocol}. Throughout this proof, we will only be focusing on the first committed qubit (since $j = 1$). Therefore, for notational convenience, we will drop the subscript of 1 for values pertaining to the first committed qubit (i.e. the basis choice, function key, commitment string, etc.). Let $h_{> 1} = (h_2,\ldots,h_n)$, let $k_{> 1} = (k_2,\ldots,k_n)$ and define $t_{k> 1}$ similarly. We will also require the following mixed state, which contains the distribution over all function keys and trapdoors except the first (as sampled by the verifier). 
\begin{equation}\label{eq:statewithotherkeys}
    \sum_{k_{> 1}}D_{\Ver,h_{> 1}}\ket{k_{> 1}}\bra{k_{> 1}} \otimes \ket{t_{k> 1}}\bra{t_{k> 1}}
\end{equation}
This mixed state is required to create $\sigma_{\Prov,h}$: the function keys are part of the prover's input and the trapdoors are used for the verifier's decoding. For convenience, let $\ket{\phi_{k>1}}$ be a purification of the above state; when analyzing the state $\sigma_{\Prov,h}$, we can consider a purification since we will eventually be tracing out all but the committed qubits. For $b\in\{0,1\}$, let $T_{k,b} = \mathop{\cup}\limits_{x\in\sX}\supp( f'_{k,b}(x))$ ($\supp( f'_{k,b}(x))$ is the support of the probability density function $f'_{k,b}(x)$ - see Definition \ref{def:trapdoorclawfree} for a reminder). Let $T_k = T_{k,0}\cup T_{k,1}$. 

We begin by writing the state of $\Prov$ after application of $U_0$. Recall from Claim \ref{cl:generalprover} that when the verifier requests test round measurement results from $\Prov$, $\Prov$ simply measures the requested registers in the standard basis and sends the results to the verifier. Since $\Prov$ is a perfect prover, it follows that the state of $\Prov$ after applying $U_0$ must yield measurement results which pass the test round perfectly. The state of $\Prov$ after applying $U_0$ can therefore be written as:  
\begin{equation}\label{eq:afterapplyingu0beforesimplify}
U_0\ket{0}^{\otimes e}\ket{k}\ket{\phi_{k > 1}} = \sum_{\substack{b\in \{0,1\}\\ y\in T_{k,b}}} \ket{b,x_{b,y}}\ket{\psi_{b,y,k}}\ket{y}
\end{equation}
where $x_{b,y}$ is the output of INV$_{\mathcal{F}}(t_k,b,y)$. We have suppressed the dependence of $x_{b,y}$ on $k$ for convenience. The state in \eqref{eq:afterapplyingu0beforesimplify} can be written in this format since if the prover returns $y\in \sY$ in the commitment stage (step 2 of Protocol \ref{prot:measprotocol}), in the test round he must return $b\in \{0,1\}$ and $x\in \sX$ such that CHK$_{\mathcal{F}}(t_k,b,x,y) = 1$. Conditions 3(a) and 3(b) of Definition \ref{def:trapdoorclawfree} imply that only $x_{b,y} = \textrm{INV}_{\mathcal{F}}(t_k,b,y)$ satisfies this condition. The auxiliary register represented by $\ket{\psi_{b,y,k}}$ includes the remaining $n - 1$ committed qubits, preimage registers and commitment strings as well as the state $\ket{\phi_{k > 1}}$. 

For convenience, we will instead write the state in \eqref{eq:afterapplyingu0beforesimplify} as follows:
\begin{equation}\label{eq:afterapplyingu0}
U_0\ket{0}^{\otimes e}\ket{k}\ket{\phi_{k > 1}} = \sum_{\substack{b\in \{0,1\}\\ y\in T_k}} \ket{b,x_{b,y}}\ket{\psi_{b,y,k}}\ket{y}
\end{equation}
The only change we have made is we have replaced the summation over $y\in T_{k,b}$ with the summation over $y\in T_k = T_{k,0}\cup T_{k,1}$. Note that the existence of two inverses of $y$ ($x_{0,y}$ and $x_{1,y}$) is guaranteed since $y\in T_k$ - see Definition \ref{def:trapdoorclawfree}, item 3(a). For $b,y$ for which $y\notin T_{k,b}$, let $\ket{\psi_{b,y,k}} = 0$. 

After the prover measures $y$ and sends it to the verifier, the state shared between the prover and verifier is: 
\begin{equation}
\sum_{y\in T_k} (\sum_{b\in \{0,1\}}\ket{b,x_{b,y}}\ket{\psi_{b,y,k}})(\sum_{b\in \{0,1\}}\ket{b,x_{b,y}}\ket{\psi_{b,y,k}})^\dagger\otimes \ket{y}\bra{y}
\end{equation}
and the last register (containing $y$) is held by the verifier. Next, the prover applies the injective map $U_{\inj}$ (see Section \ref{sec:measofcommitted}) to all $n$ preimage registers along with auxiliary 0 qubits, which we assume have already been included in the extra space $\ket{\psi_{b,y,k}}$. At this point, the state shared between the prover and verifier is:
\begin{equation}\label{eq:stateaftercommitment}
\rho_k = \sum_{y\in T_k} \rho_{y,k} 
\end{equation}
where
\begin{equation}\label{eq:stateforintersection}
\rho_{y,k} = (\sum_{b\in \{0,1\}}\ket{b,\inj(x_{b,y})}\ket{\psi'_{b,y,k}})(\sum_{b\in \{0,1\}}\ket{b,\inj(x_{b,y})}\ket{\psi'_{b,y,k}})^\dagger\otimes \ket{y}\bra{y}
\end{equation}
The auxiliary space $\ket{\psi_{b,y,k}}$ has changed to $\ket{\psi'_{b,y,k}}$ to account for the fact that the commitment strings corresponding to indices $i > 1$ were measured and the unitary $U_{\inj}$ was applied to the corresponding preimage registers in the auxiliary space (we are considering a purification of the auxiliary space for convenience). 

The prover then applies his CPTP map $\mathcal{S} = \{B_{\tau}\}_{\tau}$ followed by Hadamard basis measurement of the first committed qubit and preimage register of the state in \eqref{eq:stateaftercommitment}. The state shared between the prover and verifier at this point is:
\begin{equation}
   \sum_{\substack{b'\in \{0,1\}, \tau\\ d\in\{0,1\}^w}} (\ket{b'}\bra{b'}\otimes \ket{d}\bra{d}\otimes \mcI) (H^{\otimes l+1}\otimes \mcI)B_{\tau}\rho_{k}B_{\tau}^\dagger(H^{\otimes l+1}\otimes \mcI)^\dagger (\ket{b'}\bra{b'}\otimes \ket{d}\bra{d}\otimes \mcI)^\dagger
\end{equation}
Next, if the measurement result $d\in \dset_{k,0,x_{0,y}} \cap \dset_{k,1,x_{1,y}}$, the verifier decodes the first qubit by applying the operator $X^{d\cdot (\inj(x_{0,y})\oplus \inj(x_{1,y}))}$ (see step 5(d) of Protocol \ref{prot:measprotocol}). If not, the verifier stores a random bit as his measurement result; we can equivalently assume the verifier decodes the first qubit by applying a random $X$ operator. Note that the verifier's decoding (the application of the $X$ operator) commutes with the prover's measurement of the first qubit. Therefore, the entire state, including the verifier's decoding, can be written as:
\begin{equation}
\sigma_{0,k} = \sum_{\substack{b',c\in \{0,1\}, \tau\\ d\in\{0,1\}^w,y\in R_{c,d,k}}} \delta_{d,y}(\ket{b'}\bra{b'}X^{c}\otimes \ket{d}\bra{d}\otimes \mcI) (H^{\otimes l+1}\otimes \mcI)B_{\tau}\rho_{y,k}B_{\tau}^\dagger(H^{\otimes l+1}\otimes \mcI)^\dagger (\ket{b'}\bra{b'}X^{c}\otimes \ket{d}\bra{d}\otimes \mcI)^\dagger\label{eq:entire0}
\end{equation}
where $\delta_{d,y} = \frac{1}{2}$ if $d\notin \dset_{k,0,x_{0,y}} \cap \dset_{k,1,x_{1,y}}$ and $1$ if $d\in \dset_{k,0,x_{0,y}} \cap \dset_{k,1,x_{1,y}}$ and 
\begin{equation}\label{eq:defsetr}
R_{c,d,k} = \{y\in T_k | (d\in \dset_{k,0,x_{0,y}} \cap \dset_{k,1,x_{1,y}} \land d\cdot (\inj(x_{0,y})\oplus \inj(x_{1,y}))= c) \lor (d\notin \dset_{k,0,x_{0,y}} \cap \dset_{k,1,x_{1,y}})\}
\end{equation}
The set $R_{c,d,k}$ is defined to ensure that the verifier's decoding operator is correct. For ease of notation, we will instead write the state in \eqref{eq:entire0} as:
\begin{equation}\label{eq:entire}
\sigma_{0,k} = \sum_{\substack{b',c\in \{0,1\},\tau\\ d\in\{0,1\}^w,y\in R_{c,d,k}}}   \delta_{d,y}O_{b',c,d,\tau} \rho_{y,k}  O_{b',c,d,\tau}^\dagger
\end{equation}
where
\begin{eqnarray}\label{eq:defsuperoperatorO}
O_{b',c,d,\tau} &=& (\ket{b'}\bra{b'}X^c \otimes \ket{d}\bra{d}\otimes \mcI )(H^{\otimes l + 1}\otimes \mcI) B_{\tau}
\end{eqnarray}
Let $\mathcal{S}_{> 1}$ be the CPTP map which contains all operations done on the remaining $n - 1$ committed qubits and preimage registers after application of the attack $\mathcal{S}$: $\mathcal{S}_{> 1}$ consists of the Hadamard measurement of the remaining $n - 1$ committed qubits and preimage registers as well as the verifier decoding of those committed qubits. $\mathcal{S}_{> 1}$ is independent of the function key $k$ and trapdoor $t_k$; it is only dependent on the remaining $n - 1$ function keys and trapdoors, which are drawn independently and included in the auxiliary space of $\sigma_{0,k}$. Given this, the state $\sigma_{\Prov,h}$ is obtained by applying $\mathcal{S}_{> 1}$, and then tracing out all but the first $n$ qubits (the committed qubits): 
\begin{equation}\label{eq:firstproverextendedstate0}
\sigma_{\Prov,h} = \tr_{> n}(\mathcal{S}_{> 1}(\sum_{k\in\mathcal{K}_{\mathcal{F}}} D_{\Ver,h}(k)  \sigma_{0,k}))
\end{equation}
where $D_{\Ver,h}$ is the distribution over the set of function keys $\mathcal{K}_{\mathcal{F}}$ (since $h = 1$) produced by $\textrm{GEN}_{\mathcal{F}}$.

To prove the claim, we need to show that $\sigma_{\Prov,h}$ is computationally indistinguishable from $\sigma_{\Prov_1,h}$. Let
\begin{equation}\label{eq:firstproverextendedstate}
\sigma_{\Prov,h,E} = \sum_{k\in\mathcal{K}_{\mathcal{F}}} D_{\Ver,h}(k)  \sigma_{0,k}
\end{equation}
We will instead prove the stronger statement that $\sigma_{\Prov,h,E}$ is computationally indistinguishable from $\sigma_{\Prov_1,h,E}$ (to see why this is stronger, observe from \eqref{eq:firstproverextendedstate0} that $\sigma_{\Prov,h}$ can be obtained from $\sigma_{\Prov,h,E}$ by applying the efficiently computable superoperator $\mathcal{S}_{> 1}$ and tracing out all but the first $n$ qubits). Recall from the statement of the claim that $\Prov_1$ is characterized by $(U_0,\mathcal{S}_1)$. Since $\mathcal{S}_1$ is followed by Hadamard basis measurement, Corollary \ref{corol:diagonalizingmeas} implies that $\Prov_1$ is also characterized by $(U_0,\{\frac{1}{\sqrt{2}}(Z^r\otimes \mcI)\mathcal{S}(Z^r\otimes \mcI)\}_{r\in\{0,1\}})$. If we let $\hat{\Prov}_1$ be the prover characterized by $(U_0, (Z\otimes \mcI)\mathcal{S}(Z\otimes \mcI))$, it follows by linearity that
\begin{eqnarray}
\frac{1}{2}(\sigma_{\Prov,h,E} + \sigma_{\hat{\Prov}_1,h,E}) = \sigma_{\Prov_1,h,E}
\end{eqnarray}
Therefore, to complete the proof of Lemma \ref{lem:singlesecurity}, we can instead show that $\sigma_{\Prov,h,E}$ is computationally indistinguishable from $\sigma_{\hat{\Prov}_1,h,E}$, which implies that $\sigma_{\Prov,h,E}$ is computationally indistinguishable from $\sigma_{\Prov_1,h,E}$. 
\\~\\
\textbf{Computational Indistinguishability}
For convenience, let $\sigma_{\Prov,h,E} = \sigma_0$ and $\sigma_{\hat{\Prov}_1,h,E} = \sigma_1$. We now prove that $\sigma_0$ and $\sigma_1$ are computationally indistinguishable; our proof follows the outline given in Section \ref{sec:overviewcomppauli}. As given in \eqref{eq:firstproverextendedstate}:
\begin{equation}\label{eq:startingstatesum}
\sigma_r = \sum_{k\in\mathcal{K}_{\mathcal{F}}} D_{\Ver,h}(k)  \sigma_{r,k}
\end{equation}
and using \eqref{eq:entire}:
\begin{equation}\label{eq:startingstateindclaim}
\sigma_{r,k} = \sum_{\substack{b',c\in \{0,1\},\tau\\ d\in\{0,1\}^w,y\in R_{c,d,k}}} \delta_{d,y} O_{b',c\oplus r,d,\tau} (Z^r\otimes \mcI) \rho_{y,k}(Z^r\otimes \mcI)  O_{b',c\oplus r,d,\tau}^\dagger
\end{equation}
Note that, in the case that $r = 1$, the operator $(Z\otimes \mcI)$ acting after $B_{\tau}$ was absorbed into $O_{b',c, d,\tau}$ to create $O_{b',c\oplus r, d,\tau}$. Recall from \eqref{eq:stateforintersection} that:
\begin{equation}
\rho_{y,k} = (\sum_{b\in \{0,1\}}\ket{b,\inj(x_{b,y})}\ket{\psi'_{b,y,k}})(\sum_{b\in \{0,1\}}\ket{b,\inj(x_{b,y})}\ket{\psi'_{b,y,k}})^\dagger\otimes \ket{y}\bra{y}
\end{equation}
We can break down the state $\rho_{y,k}$ into two components: 
\begin{eqnarray}\label{eq:componentbreakdown}
\rho_{y,k} = \rho^D_{y,k} + \rho^{C}_{y,k}
\end{eqnarray} The components are as follows:
\begin{eqnarray}
\rho^D_{y,k} &=& \sum_{b\in \{0,1\}}\ket{b}\bra{b}\otimes \ket{\inj(x_{b,y})}\bra{\inj(x_{b,y})} \otimes \ket{\psi'_{b,y,k}} \bra{\psi'_{b,y,k}} \otimes \ket{y}\bra{y}\label{eq:diagonalcomponent} \\
\rho^{C}_{y,k} &=& \sum_{b\in \{0,1\}}\ket{b}\bra{b\oplus 1}\otimes \ket{\inj(x_{b,y})}\bra{\inj(x_{b\oplus 1, y})} \otimes \ket{\psi'_{b,y,k}} \bra{\psi'_{b\oplus 1,y}}\otimes \ket{y}\bra{y}\label{eq:crosscomponent}
\end{eqnarray}
Since $Z$ operators acting on the first qubit have no effect on \eqref{eq:diagonalcomponent} and add a phase of -1 to \eqref{eq:crosscomponent}, we can rewrite \eqref{eq:startingstateindclaim} as:
\begin{eqnarray}\label{eq:startingstatesimplified}
\sigma_{r,k} &=& \sum_{\substack{b',c\in \{0,1\},\tau\\ d\in\{0,1\}^w,y\in R_{c,d,k}}} \delta_{d,y}O_{b',c\oplus r,d,\tau}(Z^r\otimes \mcI)(\rho_{y,k}^D + \rho_{y,k}^C)(Z^r\otimes\mcI) O_{b',c\oplus r,d,\tau}^\dagger\\
&=& \sum_{\substack{b',c\in \{0,1\},\tau\\ d\in\{0,1\}^w,y\in R_{c,d,k}}} \delta_{d,y}O_{b',c\oplus r,d,\tau}(\rho_{y,k}^D + (-1)^r\rho_{y,k}^C)O_{b',c\oplus r,d,\tau}^\dagger\label{eq:sigmabreakdown}
\end{eqnarray}
To show that $\sigma_0$ and $\sigma_1$ are computationally indistinguishable, we reduce the problem to showing that the components corresponding to the diagonal and cross terms of the committed state $\rho_{y,k}$ are computationally indistinguishable:
\begin{claim}\label{cl:reductiontodiagonalcross}
If $\sigma_0$ is computationally distinguishable from $\sigma_1$, then one of the following must hold:
\begin{enumerate}
\item For $r\in\{0,1\}$, let
\begin{eqnarray}
\sigma_r^D &\EqDef& \sum_{k\in\mathcal{K}_{\mathcal{F}}} D_{\Ver,h}(k)  \sigma_{r,k}^D\label{eq:diagonaltermsumdef}\\
\sigma_{r,k}^D &\EqDef& \sum_{\substack{b',c\in \{0,1\},\tau\\ d\in\{0,1\}^w,y\in R_{c,d,k}}} \delta_{d,y}O_{b',c\oplus r,d,\tau}(\rho_{y,k}^D)O_{b',c\oplus r,d,\tau}^\dagger\label{eq:diagonaltermdef}
\end{eqnarray}
The density matrices $\sigma_0^D$ and $\sigma_1^D$ are computationally distinguishable. 

\item For $r\in\{0,1\}$, let
\begin{eqnarray}
\hat{\sigma}_r &\EqDef& \sum_{k\in\mathcal{K}_{\mathcal{F}}} D_{\Ver,h}(k)  \hat{\sigma}_{r,k}\label{eq:crosstermsumdef}\\
\hat{\sigma}_{r,k} &\EqDef& (Z^r\otimes \mcI)(\sum_{y\in T_k}\rho_{y,k})(Z^r\otimes \mcI)\label{eq:crosstermdef}
\end{eqnarray}

The density matrices $\hat{\sigma}_0$ and $\hat{\sigma}_1$ are computationally distinguishable.

\end{enumerate}
\end{claim}
It should be immediately apparent that the first pair of density matrices is equal to the terms of $\sigma_0$ and $\sigma_1$ resulting from the diagonal terms of the committed state. It turns out that the second pair of density matrices represents the cross terms; this will be shown in the proof of Claim \ref{cl:reductiontodiagonalcross}, which is a simple application of the triangle inequality and given in Section \ref{sec:proofofreductiontodiagonalcross}. We complete the proof of Lemma \ref{lem:singlesecurity} with the following two claims:
\begin{claim}\label{cl:stateind1}
If $\sigma_0^D$ is computationally distinguishable from $\sigma_1^D$, then there exists a \BQP\ attacker $\mathcal{A}$ who can violate the hardcore bit clause of Definition \ref{def:trapdoorclawfree}. 
\end{claim}
\begin{claim}\label{cl:stateind2}
If $\hat{\sigma}_0$ is computationally distinguishable from $\hat{\sigma}_1$, then there exists a \BQP\ attacker $\mathcal{A}$ who can violate the hardcore bit clause of Definition \ref{def:extendedtrapdoorclawfree}. 
\end{claim}
We prove Claims \ref{cl:stateind1} and \ref{cl:stateind2} in the next two sections. 

\end{proofof}

\subsubsection{Indistinguishability of Diagonal Terms (Proof of Claim \ref{cl:stateind1})}\label{sec:proofofstateind1}

The overview of the following proof is given in Section \ref{sec:overviewdiagonal}.

\begin{proofof}{\textbf{Claim \ref{cl:stateind1}}}
Recall that we would like to show that the density matrices $\sigma_0^D$ and $\sigma_1^D$ are computationally indistinguishable. We will proceed by contradiction. Assume $\sigma_0^D$ and $\sigma_1^D$ are distinguishable using an efficiently computable CPTP map $\mathcal{O}$. It follows by the definition of computational indistinguishability (Definition \ref{def:compind}) and by the expression for $\sigma_r^D$ in \eqref{eq:diagonaltermsumdef} that the following expression is non negligible:
\begin{equation}\label{eq:compinddefstateind1}
|\sum_{k\in\mathcal{K}_{\mathcal{F}}} D_{\Ver,h}(k)\cdot \tr((\ket{0}\bra{0}\otimes \mcI)\mathcal{O}(\sum_{r\in \{0,1\}} (-1)^r \sigma_{r,k}^D))|
\end{equation}
We will use the CPTP map $\mathcal{O}$ to construct an attacker $\mathcal{A}$ who violates the hardcore bit clause of Definition \ref{def:trapdoorclawfree}. 

Let
\begin{eqnarray}\label{eq:defsetrdiagonal}
R^D_{c,d,k} &=& \{y\in T_k | d\in \dset_{k,0,x_{0,y}} \cap \dset_{k,1,x_{1,y}} \land d\cdot (\inj(x_{0,y})\oplus \inj(x_{1,y}))= c\}
\end{eqnarray}
We will require the unnormalized state $\tilde{\sigma}_{r,k}^D$, which is the state $\sigma_{r,k}^D$ from \eqref{eq:diagonaltermdef} conditioned on obtaining measurements $y,d$ such that $d\in \dset_{k,0,x_{0,y}} \cap \dset_{k,1,x_{1,y}}$: 
\begin{eqnarray}\label{eq:defhatsigmaD}
\tilde{\sigma}_{r,k}^D &=& \sum_{\substack{b',c\in \{0,1\},\tau\\ d\in\{0,1\}^w,y\in R^D_{c,d,k}}} O_{b',c\oplus r,d,\tau}(\rho_{y,k}^D)O_{b',c\oplus r,d,\tau}^\dagger
\end{eqnarray}
Observe that for all $k\in\mathcal{K}_{\mathcal{F}}$
\begin{eqnarray}\label{eq:distaftercondition}
\sum\limits_{r\in\{0,1\}}(-1)^r\sigma_{r,k}^D = \sum\limits_{r\in\{0,1\}}(-1)^r\tilde{\sigma}_{r,k}^D
\end{eqnarray}
This is because $\sigma_{0,k}^D$ and $\sigma_{1,k}^D$ are identical when conditioned on $d\notin \dset_{k,0,x_{0,y}} \cap \dset_{k,1,x_{1,y}}$ (both have a uniform $X$ decoding operator applied). It follows that the expression in \eqref{eq:compinddefstateind1} is equal to:
\begin{eqnarray}
|\sum_{k\in\mathcal{K}_{\mathcal{F}}} D_{\Ver,h}(k)\cdot \tr((\ket{0}\bra{0}\otimes \mcI)\mathcal{O}(\sum_{r\in \{0,1\}} (-1)^r\tilde{\sigma}_{r,k}^D))|
\end{eqnarray}
The attacker $\mathcal{A}$ (on input $k\in\mathcal{K}_{\mathcal{F}}$) first constructs the following state from \eqref{eq:stateaftercommitment}: 
\begin{equation}
\sum_{y\in T_k}\rho_{y,k}
\end{equation}
He does this exactly as the \BQP\ prover $\Prov$ would have: by applying the prover's initial operator $U_0$ to the input state and measuring the last register to obtain $y$. Then $\mathcal{A}$ measures both the committed qubit and preimage register, obtaining $(b,x_{b,y})$. The resulting state is a summation over the state in \eqref{eq:diagonalcomponent}:
\begin{equation}
\sum_{y\in T_k}\rho_{y,k}^D
\end{equation}
$\mathcal{A}$ now continues as $\Prov$ would have: he applies the CPTP map $\{B_{\tau}\}_{\tau}$, followed by a Hadamard measurement of the committed qubit and preimage register, obtaining results $b'\in \{0,1\}$ and $d\in\{0,1\}^w$. $\mathcal{A}$ then chooses a bit $c'$ at random, stores the bit $c'$ in an auxiliary register, and applies $X^{c'}$ to the committed qubit (this operation commutes with measurement of the committed qubit). The unnormalized state created by $\mathcal{A}$ (conditioned on $d,y$ such that $d\in \dset_{k,0,x_{0,y}} \cap \dset_{k,1,x_{1,y}}$) is equal to: 
\begin{eqnarray}
\frac{1}{2}\sum_{\substack{b',c'\in \{0,1\},\tau,y\in T_k\\ d\in\dset_{k,0,x_{0,y}} \cap \dset_{k,1,x_{1,y}}}} O_{b',c',d,\tau}(\rho_{y,k}^D)O_{b',c',d,\tau}^\dagger\otimes \ket{c'}\bra{c'}\label{eq:secondtolastdistinguishingdiag}
\end{eqnarray}
We will partition the above state into components using the following projection (the set $R^D_{c,d,k}$ is defined in \eqref{eq:defsetrdiagonal}):
\begin{eqnarray}
P_{c,k}^D &=& \mcI\otimes \sum_{d\in \{0,1\}^w, y\in R^D_{c,d,k}}\ket{d}\bra{d}\otimes \mcI\otimes \ket{y}\bra{y}
\end{eqnarray}
The state of $\mathcal{A}$ in \eqref{eq:secondtolastdistinguishingdiag} can now be written in terms of the state $\tilde{\sigma}_{r,k}^D$, as defined in \eqref{eq:defhatsigmaD}:
\begin{eqnarray}
&=& \frac{1}{2}\sum_{c'\in \{0,1\}} P_{c',k}^D\tilde{\sigma}_{0,k}^D P_{c',k}^D \otimes \ket{c'}\bra{c'} + \frac{1}{2}\sum_{c'\in \{0,1\}} P_{c'\oplus 1,k}^D\tilde{\sigma}_{1,k}^D P_{c'\oplus 1,k}^D \otimes \ket{c'}\bra{c'}\\
&=& \frac{1}{2}\sum_{c,c'\in \{0,1\}} P_{c,k}^D\tilde{\sigma}_{c\oplus c',k}^D P_{c,k}^D \otimes \ket{c'}\bra{c'}\label{eq:finalstatexdistinguishing}
\end{eqnarray}
Finally, $\mathcal{A}$ applies the efficiently computable CPTP map $\mathcal{O}$ (which is used to distinguish between $\sigma_0^D$ and $\sigma_1^D$) to the state in \eqref{eq:finalstatexdistinguishing} and measures the first qubit. If the measurement result of the first qubit is $r\in\{0,1\}$, $\mathcal{A}$ outputs $b,x_{b,y},d,c'\oplus r$. 

In order to violate the hardcore bit clause of Definition \ref{def:trapdoorclawfree}, $\mathcal{A}$ must output $(b,x_b,d,d\cdot(\inj(x_0)\oplus \inj(x_1)))$ with non negligible advantage (over outputting $(b,x_b,d,d\cdot(\inj(x_0)\oplus \inj(x_1))\oplus 1)$). More formally, we need to show that the following advantage of $\mathcal{A}$ (taken from Definition \ref{def:trapdoorclawfree}) is non negligible:
\begin{equation}\label{eq:diagonalhardcoreviolation}
\Big|\Pr_{(k,t_k)\leftarrow \textrm{GEN}_{\mathcal{F}}(1^{\lambda})}[\mathcal{A}(k) \in H_k] - \Pr_{(k,t_k)\leftarrow \textrm{GEN}_{\mathcal{F}}(1^{\lambda})}[\mathcal{A}(k) \in\overline{H}_k]\Big|
\end{equation}
where
\begin{eqnarray}
H_k &=& \{(b,x_b,d,d\cdot(\inj(x_0)\oplus \inj(x_1)))| b\in \{0,1\},(x_0,x_1)\in \mathcal{R}_k, d\in \dset_{k,0,x_0}\cap \dset_{k,1,x_1}\}\nonumber\\
\overline{H}_k &=& \{(b,x_b,d,c)| (b,x_b,d,c\oplus 1)\in H_k\}\nonumber
\end{eqnarray}
$\mathcal{A}$ outputs a string in $H_k$ if, on components $P_{c,k}^D\tilde{\sigma}^D_{0,k}P_{c,k}^D$ and $P_{c,k}^D\tilde{\sigma}^D_{1,k}P_{c,k}^D$, the final bit of $\mathcal{A}$'s output is $c$. This occurs as long as the distinguishing operator $\mathcal{O}$ outputs $r$ on components $P_{0,k}^D\tilde{\sigma}^D_{r,k}P_{0,k}^D$ and $P_{1,k}^D\tilde{\sigma}^D_{r,k}P_{1,k}^D$. It follows that the probability that $\mathcal{A}$ outputs a string in $H_k$ is equal to:
\begin{eqnarray}\label{eq:probainhk}
\Pr_{(k,t_k)\leftarrow \textrm{GEN}_{\mathcal{F}}(1^{\lambda})}[\mathcal{A}(k) \in H_k] &=& \frac{1}{2}\sum_{\substack{k\in \mathcal{K}_{\mathcal{F}}\\ r\in\{0,1\}}} D_{\Ver,h}(k)\cdot \tr((\ket{r}\bra{r}\otimes \mcI) \mathcal{O} (\sum_{c\in \{0,1\}}P_{c,k}^D\tilde{\sigma}_{r,k}^DP_{c,k}^D))\nonumber\\
&=& \frac{1}{2}\sum_{\substack{k\in \mathcal{K}_{\mathcal{F}}\\ r\in\{0,1\}}} D_{\Ver,h}(k)\cdot \tr((\ket{r}\bra{r}\otimes \mcI) \mathcal{O} (\tilde{\sigma}_{r,k}^D))
\end{eqnarray}
By similar reasoning, 
\begin{equation}\label{eq:probainhkbar}
\Pr_{(k,t_k)\leftarrow \textrm{GEN}_{\mathcal{F}}(1^{\lambda})}[\mathcal{A}(k) \in \overline{H}_k] 
= \frac{1}{2}\sum_{\substack{k\in \mathcal{K}_{\mathcal{F}}\\ r\in\{0,1\}}} D_{\Ver,h}(k)\cdot \tr((\ket{r\oplus 1}\bra{r\oplus 1}\otimes \mcI) \mathcal{O} (\tilde{\sigma}_{r,k}^D))
\end{equation}
By combining \eqref{eq:probainhk} and \eqref{eq:probainhkbar} and then using the equality in \eqref{eq:distaftercondition}, we obtain that the advantage of $\mathcal{A}$ in \eqref{eq:diagonalhardcoreviolation} is equal to:
\begin{equation}
\Big|\sum_{k\in \mathcal{K}_{\mathcal{F}}} D_{\Ver,h}(k)\cdot \tr((\ket{0}\bra{0}\otimes \mcI) \mathcal{O} (\tilde{\sigma}_{0,k}^D - \tilde{\sigma}_{1,k}^D ) ) \Big|= \Big|\sum_{k\in \mathcal{K}_{\mathcal{F}}} D_{\Ver,h}(k)\cdot \tr((\ket{0}\bra{0}\otimes \mcI) \mathcal{O} (\sigma_{0,k}^D - \sigma_{1,k}^D ) )\Big |\label{eq:finalexpressdiag}
\end{equation}
The expression in \eqref{eq:finalexpressdiag} is non negligible, due to our initial assumption that the CPTP map $\mathcal{O}$ can distinguish between $\sigma_0^D$ and $\sigma_1^D$ (see \eqref{eq:compinddefstateind1}).

\end{proofof}

\subsubsection{Indistinguishability of Cross Terms (Proof of Claim \ref{cl:stateind2})}\label{sec:proofofstateind2}
The overview of the following proof is given in Section \ref{sec:overviewcross}.

\begin{proofof}{\textbf{Claim \ref{cl:stateind2}}}
Recall that we would like to show that the density matrices $\hat{\sigma}_0$ and $\hat{\sigma}_1$ are computationally indistinguishable. We will proceed by contradiction. Assume the two matrices $\hat{\sigma}_0$ and $\hat{\sigma}_1$ are computationally distinguishable using the efficiently computable CPTP map $\mathcal{O}$. It follows by the definition of computational indistinguishability (Definition \ref{def:compind}) and the expression for $\hat{\sigma}_r$ in \eqref{eq:crosstermsumdef} that the following expression is non negligible:
\begin{equation}\label{eq:compinddefstateind2}
|\sum_{k\in\mathcal{K}_{\mathcal{F}}} D_{\Ver,h}(k)\cdot \tr((\ket{0}\bra{0}\otimes \mcI)\mathcal{O}(\sum_{r\in \{0,1\}} (-1)^r\hat{\sigma}_{r,k}))|
\end{equation}
We will use the CPTP map $\mathcal{O}$ to construct a \BQP\ attacker $\mathcal{A}$ who will violate the hardcore bit clause of Definition \ref{def:extendedtrapdoorclawfree} (item 2). 

Let $d\in\{0,1\}^w$ be the string referred to in item 2 of Definition \ref{def:extendedtrapdoorclawfree}. The attacker $\mathcal{A}$ (on input $k\in\mathcal{K}_{\mathcal{F}}$) first constructs the following state from \eqref{eq:stateaftercommitment}: 
\begin{equation}
\sum_{y\in T_k}\rho_{y,k}
\end{equation}
He does this exactly as the \BQP\ prover $\Prov$ would have: by applying the prover's initial operator $U_0$ to the input state and measuring the last register to obtain $y$. Then $\mathcal{A}$ applies $Z^{d}$ to the preimage register. The resulting state is:  
\begin{equation}
(\mcI \otimes Z^{d} \otimes \mcI)(\sum_{y\in T_k}\rho_{y,k}) (\mcI \otimes Z^{d} \otimes \mcI)= \sum_{y\in T_k} (Z^{d\cdot(\inj(x_{0,y})\oplus \inj(x_{1,y}))} \otimes \mcI)\rho_{y,k} (Z^{d\cdot(\inj(x_{0,y})\oplus \inj(x_{1,y}))} \otimes \mcI)\label{eq:zoncommitment}
\end{equation}
The equality is due to the format of the state $\rho_{y,k}$, as written in \eqref{eq:stateforintersection}:
\begin{equation}
\rho_{y,k} = (\sum_{b\in \{0,1\}}\ket{b,\inj(x_{b,y})}\ket{\psi'_{b,y,k}})(\sum_{b\in \{0,1\}}\ket{b,\inj(x_{b,y})}\ket{\psi'_{b,y,k}})^\dagger\otimes \ket{y}\bra{y}
\end{equation}
If we let 
\begin{equation}
R_{c,d,k}^C = \{y\in T_k | d\cdot (\inj(x_{0,y})\oplus \inj(x_{1,y})) = c\}
\end{equation}
 the expression in \eqref{eq:zoncommitment} is equal to:
\begin{eqnarray}
 \sum_{c\in \{0,1\}, y\in R_{c,d,k}^C} (Z^c \otimes \mcI)\rho_{y,k} (Z^c \otimes \mcI)\label{eq:secondtolastdistinguishingcross}
\end{eqnarray}
Finally, $\mathcal{A}$ chooses a random bit $c'$, applies $Z^{c'}$ to the committed qubit and stores $c'$ in an auxiliary register. Continuing from \eqref{eq:secondtolastdistinguishingcross}, the state of $\mathcal{A}$ at this point is equal to: 
\begin{eqnarray}\label{eq:stateind2beforepartition}
\frac{1}{2}\sum_{c,c'\in \{0,1\}, y\in R_{c,d,k}^C} (Z^{c\oplus c'} \otimes \mcI)\rho_{y,k} (Z^{c\oplus c'} \otimes \mcI)\otimes \ket{c'}\bra{c'}
\end{eqnarray}
We partition the state in \eqref{eq:stateind2beforepartition} into components using the following projection:
\begin{eqnarray}
P_{c,k}^C &=& \mcI \otimes \sum_{y\in R^C_{c,d,k}}\ket{y}\bra{y}
\end{eqnarray}
The state in \eqref{eq:stateind2beforepartition} can now be written in terms of $\hat{\sigma}_{r,k}$, as defined in \eqref{eq:crosstermdef}
\begin{equation}
\hat{\sigma}_{r,k} = (Z^r\otimes \mcI)(\sum_{y\in T_k}\rho_{y,k})(Z^r\otimes \mcI)
\end{equation}
as follows:
\begin{eqnarray} 
&=& \frac{1}{2}\sum_{c,c'\in \{0,1\}} P_{c,k}^C(Z^{c\oplus c'} \otimes \mcI)(\sum_{y\in T_k} \rho_{y,k}) (Z^{c\oplus c'} \otimes \mcI)P_{c,k}^C\otimes \ket{c'}\bra{c'}\\
&=& \frac{1}{2}\sum_{c,c'\in \{0,1\}} P_{c,k}^C(\hat{\sigma}_{c\oplus c',k}) P_{c,k}^C\otimes \ket{c'}\bra{c'}\label{eq:finalstatezdistinguishing}
\end{eqnarray}
Finally, $\mathcal{A}$ applies the CPTP map $\mathcal{O}$ (which is used to distinguish between $\hat{\sigma}_{0}$ and $\hat{\sigma}_{1}$) to the state in \eqref{eq:finalstatezdistinguishing} and measures the first qubit. If the result of the measurement is $r\in \{0,1\}$, $\mathcal{A}$ outputs $c'\oplus r$.

In order to violate the hardcore bit clause of Definition \ref{def:extendedtrapdoorclawfree}, $\mathcal{A}$ must guess the value of $d\cdot (\inj(x_{0,y})\oplus \inj(x_{1,y}))$ with non negligible advantage. More formally, we need to show that the following advantage of $\mathcal{A}$ (taken from Definition \ref{def:extendedtrapdoorclawfree}) is non negligible:
\begin{equation}\label{eq:advantagecrossterms}
\Big|\Pr_{(k,t_k)\leftarrow \textrm{GEN}_{\mathcal{F}}(1^{\lambda})}[\mathcal{A}(k) \in H'_{k,d}] - \frac{1}{2}\Big|
\end{equation}
where
\begin{eqnarray}
H'_{k,d} &=& \{d\cdot (\inj(x_0)\oplus \inj(x_1))|  (x_0,x_1)\in\mathcal{R}_k \}
\end{eqnarray}
The bit output by $\mathcal{A}$ is in $H'_{k,d}$ if, on components $P_{c,k}^C\hat{\sigma}_{0,k}P_{c,k}^C$ and $P_{c,k}^C\hat{\sigma}_{1,k}P_{c,k}^C$, $\mathcal{A}$ outputs $c$. This occurs as long as the distinguishing operator $\mathcal{O}$ outputs $r$ on components $P_{0,k}^C\hat{\sigma}_{r,k}P_{0,k}^C$ and $P_{1,k}^C\hat{\sigma}_{r,k}P_{1,k}^C$. It follows that the probability that $\mathcal{A}$ outputs a value in $H'_{k,d}$ is equal to: 
\begin{eqnarray}
\Pr_{(k,t_k)\leftarrow \textrm{GEN}_{\mathcal{F}}(1^{\lambda})}[\mathcal{A}(k,d) \in H'_{k,d}] &=& \frac{1}{2}\sum_{\substack{k\in \mathcal{K}_{\mathcal{F}}\\ r\in\{0,1\}}} D_{\Ver,h}(k)\cdot \tr((\ket{r}\bra{r}\otimes \mcI)\mathcal{O}(\sum_{c\in \{0,1\}} P_{c,k}^C\hat{\sigma}_{r,k}P_{c,k}^C ))\nonumber\\
&=& \frac{1}{2}\sum_{\substack{k\in \mathcal{K}_{\mathcal{F}}\\ r\in\{0,1\}}} D_{\Ver,h}(k)\cdot \tr((\ket{r}\bra{r}\otimes \mcI)\mathcal{O}(\hat{\sigma}_{r,k} ))\nonumber\\
&=& \frac{1}{2}\sum_{k\in\mathcal{K}_{\mathcal{F}}} D_{\Ver,h}(k)\cdot \tr((\ket{0}\bra{0}\otimes \mcI)\mathcal{O}(\hat{\sigma}_{0,k} - \hat{\sigma}_{1,k} )) + \frac{1}{2} \nonumber\label{eq:attackerprobinhprime}
\end{eqnarray}
We can use the above expression to write the advantage of $\mathcal{A}$ from \eqref{eq:advantagecrossterms} as:
\begin{eqnarray}
\frac{1}{2}\Big|\sum_{k\in\mathcal{K}_{\mathcal{F}}} D_{\Ver,h}(k)\cdot \tr((\ket{0}\bra{0}\otimes \mcI)\mathcal{O}(\hat{\sigma}_{0,k} - \hat{\sigma}_{1,k}))\Big|\label{eq:finalexpresscross}
\end{eqnarray}
The expression in \eqref{eq:finalexpresscross} is non negligible, due to our initial assumption that the CPTP map $\mathcal{O}$ can distinguish between $\hat{\sigma}_0$ and $\hat{\sigma}_1$ (see \eqref{eq:compinddefstateind2}).

\end{proofof}

\subsubsection{Reduction to Diagonal/Cross Terms (Proof of \Cl{cl:reductiontodiagonalcross})}\label{sec:proofofreductiontodiagonalcross}
\begin{proof}
If $\sigma_0$ is computationally distinguishable from $\sigma_1$, by Definition \ref{def:compind} and the expression for $\sigma_r$ in \eqref{eq:startingstatesum} there exists an efficiently computable CPTP map $\mathcal{O}$ for which the following expression is non negligible:
\begin{equation}\label{eq:compdistneg}
\Big|\sum_{k\in\mathcal{K}_{\mathcal{F}}} D_{\Ver,h}(k)\cdot \tr((\ket{0}\bra{0}\otimes \mcI)\mathcal{O}(\sum_{r\in \{0,1\}} (-1)^r\sigma_{r,k}))\Big|
\end{equation}
We use the expression for $\sigma_{r,k}$ from  \eqref{eq:sigmabreakdown} and the matrix $\sigma_{r,k}^D$ from \eqref{eq:diagonaltermdef} to define the matrix $\sigma_{r,k}^C$:
\begin{eqnarray}
\sigma_{r,k} &=&  \sum_{\substack{b',c\in \{0,1\},\tau\\ d\in\{0,1\}^w,y\in R_{c,d,k}}} \delta_{d,y}O_{b',c\oplus r,d,\tau}(\rho_{y,k}^D + (-1)^r\rho_{y,k}^C)O_{b',c\oplus r,d,\tau}^\dagger\\
&=&  \sigma_{r,k}^D + \sigma_{r,k}^C\label{eq:sigmacrosstermproof}
\end{eqnarray}
By combining the following equality
\begin{equation}
\sum_{r\in \{0,1\}} (-1)^r\sigma_{r,k} = \sum_{r\in \{0,1\}} (-1)^r \sigma_{r,k}^D + \sum_{r\in \{0,1\}} (-1)^r\sigma_{r,k}^C
\end{equation}
with the triangle inequality, it follows that if the quantity in \eqref{eq:compdistneg} is non negligible, one of the following two quantities must be non negligible: 
\begin{eqnarray}
\Big|\sum_{k\in\mathcal{K}_{\mathcal{F}}} D_{\Ver,h}(k)\cdot \tr((\ket{0}\bra{0}\otimes \mcI)\mathcal{O} (\sum_{r\in \{0,1\}} (-1)^r\sigma_{r,k}^D))\Big|\label{eq:diagonaltermbeforesimplify}\\
\Big|\sum_{k\in\mathcal{K}_{\mathcal{F}}} D_{\Ver,h}(k)\cdot\tr((\ket{0}\bra{0}\otimes \mcI)\mathcal{O} (\sum_{r\in \{0,1\}} (-1)^r\sigma_{r,k}^C))\Big|\label{eq:crosstermbeforesimplify}
\end{eqnarray}
If the quantity in \eqref{eq:diagonaltermbeforesimplify} is non negligible, $\sigma_0^D$ is computationally distinguishable from $\sigma_1^D$. To complete the proof of Claim \ref{cl:reductiontodiagonalcross}, we will show that if the quantity in \ref{eq:crosstermbeforesimplify} is non negligible, $\hat{\sigma}_0$ is computationally distinguishable from $\hat{\sigma}_1$ ($\hat{\sigma}_r$ is defined in \eqref{eq:crosstermdef} and copied below in \eqref{eq:crosstermdeftriangleproof}). To do this, we show below that for the efficiently computable CPTP map $\mathcal{S}' = \{\frac{1}{\sqrt{2}}O_{b',c,d,\tau}\}_{b',c,d,\tau}$ ($O_{b',c,d,\tau}$ is introduced in \eqref{eq:defsuperoperatorO}):
\begin{equation}\label{eq:crosstermdifference}
\sum_{r\in \{0,1\}}(-1)^r\sigma_{r,k}^C = \mathcal{S}'(\sum_{r\in \{0,1\}}(-1)^r\hat{\sigma}_{r,k})
\end{equation}
Therefore, if the quantity in \eqref{eq:crosstermbeforesimplify} is non negligible, $\hat{\sigma}_0$ is computationally distinguishable from $\hat{\sigma}_1$ by using the CPTP map $\mathcal{O}\mathcal{S}'$. 

We now prove \eqref{eq:crosstermdifference}, beginning with the expression for $\sigma_{r,k}^C$ in \eqref{eq:sigmacrosstermproof}. First, we observe that
\begin{equation}
    \sigma_{1,k}^C = -(X\otimes \mcI)\sigma_{0,k}^C(X\otimes \mcI)
\end{equation}
Therefore: 
\begin{eqnarray}
\sigma_{0,k}^C - \sigma_{1,k}^C &=& \sum_{r\in\{0,1\}} (X^r\otimes \mcI)\sigma_{0,k}^C(X^r\otimes \mcI)\\
&=& \sum_{\substack{b',c,r\in \{0,1\},\tau\\ d\in\{0,1\}^w, y\in R_{c,d,k}}} \delta_{d,y}O_{b',c\oplus r,d,\tau}(\rho_{y,k}^C)O_{b',c\oplus r,d,\tau}^\dagger\\
&=& \sum_{\substack{b',c\in \{0,1\},\tau\\ d\in\{0,1\}^w,y\in T_k}} O_{b',c,d,\tau}(\rho_{y,k}^C)O_{b',c,d,\tau}^\dagger\label{eq:claimstateind20}
\end{eqnarray}
The third equality follows since $\delta_{d,y} = \frac{1}{2}$ if $d\notin \dset_{k,0,x_{0,y}} \cap \dset_{k,1,x_{1,y}}$ and $1$ if $d\in \dset_{k,0,x_{0,y}} \cap \dset_{k,1,x_{1,y}}$ and (from \eqref{eq:defsetr})
\begin{equation}
R_{c,d,k} = \{y\in T_k | (d\in \dset_{k,0,x_{0,y}} \cap \dset_{k,1,x_{1,y}} \land d\cdot (\inj(x_{0,y})\oplus \inj(x_{1,y}))= c) \lor (d\notin \dset_{k,0,x_{0,y}} \cap \dset_{k,1,x_{1,y}})\}
\end{equation}
Recall the state $\hat{\sigma}_{r,k}$ from \eqref{eq:crosstermdef}:
\begin{equation}\label{eq:crosstermdeftriangleproof}
\hat{\sigma}_{r,k} = (Z^r\otimes \mcI)(\sum_{y\in T_k}\rho_{y,k})(Z^r\otimes \mcI)
\end{equation}
We use the equality $\rho_{y,k} = \rho_{y,k}^D + \rho_{y,k}^C$ from \eqref{eq:componentbreakdown}:
\begin{eqnarray}
\hat{\sigma}_{0,k} - \hat{\sigma}_{1,k} &=& \sum_{y\in T_k}(\rho_{y,k}^D + \rho_{y,k}^C) - (Z\otimes \mcI) \sum_{y\in T_k}(\rho_{y,k}^D + \rho_{y,k}^C)(Z\otimes \mcI)\\
&=& 2\sum_{y\in T_k}\rho_{y,k}^C \label{eq:crosstermgoodstate}
\end{eqnarray}
Plugging the equality in \eqref{eq:crosstermgoodstate} into \eqref{eq:claimstateind20} yields \eqref{eq:crosstermdifference}.

\end{proof}




\subsection{Proof of Measurement Protocol Soundness}\label{sec:measprotocolsoundness}
In this section, we combine the results from Sections \ref{sec:proverbehavior} to \ref{sec:generaltotrivialhadamard} to prove soundness of the measurement protocol. To do so, we will first transition from general to perfect provers, and then from perfect to trivial provers. The final result is stated in the following lemma:
\begin{lem}{\textbf{Soundness of Protocol \ref{prot:measprotocol}}}\label{lem:measprotocolsoundness}
For a prover $\Prov$ in Protocol \ref{prot:measprotocol}, let $1 - p_{h,H}$ be the probability that the verifier accepts $\Prov$ on basis choice $h$ in the Hadamard round and $1 - p_{h,T}$ be the probability that the verifier accepts $\Prov$ in the test round. There exists a state $\rho$, a prover $\Prov'$ and a negligible function $\mu$ such that for all $h$, $\TV{D_{\Prov,h}}{ D_{\Prov',h}}\leq p_{h,H} + \sqrt{p_{h,T}} + \mu$ and $D_{\Prov',h}$ is computationally indistinguishable from the distribution $D_{\rho,h}$ which results from measuring $\rho$ in the basis determined by $h$.
\end{lem}
Throughout this section, we will use notation introduced in Section \ref{sec:notation}. We first state the two claims required to prove Lemma \ref{lem:measprotocolsoundness}, and then assume correctness of these claims to prove Lemma \ref{lem:measprotocolsoundness}. In the rest of this section, we prove the required claims. 

The first required claim (proven in Section \ref{sec:proofofgeneraltoperfect}) transitions from a general prover to a \textit{perfect} prover, i.e. a prover who is always accepted in the test run by the verifier:

\begin{claim}\label{cl:generaltoperfect}
For a prover $\Prov$ in Protocol \ref{prot:measprotocol}, let $1 - p_{h,H}$ be the probability that the verifier accepts $\Prov$ on basis choice $h$ in the Hadamard round and $1 - p_{h,T}$ be the probability that the verifier accepts $\Prov$ on basis choice $h$ in the test round. There exists a perfect prover $\Prov'$ and a negligible function $\mu$ such that for all $h$, $\TV{D_{\Prov,h}^C}{D_{\Prov',h}}\leq p_{h,H} + \sqrt{p_{h,T}} + \mu$.
\end{claim}
The second required claim (proven in Section \ref{sec:proofofperfecttotrivial}) transitions from a perfect prover to a trivial prover (as defined in Definition \ref{def:trivial}, a trivial prover's Hadamard round attack commutes with standard basis measurement on the $n$ committed qubits):
\begin{claim}\label{cl:perfecttotrivialprover}
For all perfect provers $\Prov$, there exists a trivial prover $\hat{\Prov}$ such that for all $h$, $D_{\Prov,h}$ is computationally indistinguishable from $D_{\hat{\Prov},h}$. 
\end{claim}
Assuming Claim \ref{cl:generaltoperfect} and Claim \ref{cl:perfecttotrivialprover}, the proof of Lemma \ref{lem:measprotocolsoundness} is straightforward:

\begin{proofof}{\Le{lem:measprotocolsoundness}}
We begin with a prover $\Prov$ and apply Claim \ref{cl:generaltoperfect} to transition to a perfect prover $\Prov'$ such that for all $h$, $\TV{D_{\Prov,h}^C}{D_{\Prov',h}}\leq p_{h,H} + \sqrt{p_{h,T}} + \mu$ for a negligible function $\mu$. Combining Claim \ref{cl:perfecttotrivialprover} and Lemma \ref{lem:trivialprover} tells us that there exists a state $\rho$ such that for all $h$, $D_{\Prov',h}$ is computationally indistinguishable from $D_{\rho,h}$. In more detail, Claim \ref{cl:perfecttotrivialprover} shows that there exists a trivial prover $\hat{\Prov}$ such that for all $h$, $D_{\Prov',h}$ is computationally indistinguishable from $D_{\hat{\Prov},h}$. Next, Lemma \ref{lem:trivialprover} shows that there exists a state $\rho$ such that for all $h$, $D_{\hat{\Prov},h}$ is computationally indistinguishable from $D_{\rho,h}$. 

\end{proofof}

\subsubsection{General to Perfect Prover (Proof of \Cl{cl:generaltoperfect})}\label{sec:proofofgeneraltoperfect}
Before proceeding to the proof of Claim \ref{cl:generaltoperfect}, we provide some intuition. The prover $\Prov'$ in the statement of the claim is created by conditioning $\Prov$ on acceptance in the test round. This argument is straightforward, but does have one delicate aspect: we need to ensure that the perfect prover is still efficient, even though we have conditioned on acceptance in the test run. This is taken care of by recalling that in the test round, the verifier computes the procedure CHK$_{\mathcal{F}}$ or CHK$_{\mathcal{G}}$ (see Protocol \ref{prot:measprotocol}). By definition (Definitions \ref{def:trapdoorclawfree} and \ref{def:trapdoorinjective}), both of these procedures require only the function key and not the trapdoor, which implies that the procedures can be computed efficiently by the prover. 
 
\begin{proof}
We begin by observing that by definition of total variation distance (see \eqref{eq:deftotalvariation}):
\begin{equation}
    \TV{D_{\Prov,h}^C}{ D_{\Prov,h}} = p_{h,H}
\end{equation}
It follows by \eqref{eq:equivalenceofdistandstate} that $\T{\sigma_{\Prov,h}^C}{\sigma_{\Prov,h}} = p_{h,H}$. The rest of this proof focuses on showing that there exists a perfect prover $\Prov'$ and a negligible function $\mu$ such that for all $h$ for which $1 - p_{h,T}$ is non negligible, the trace distance between $\sigma_{\Prov,h}$ and $\sigma_{\Prov',h}$ is $\leq \sqrt{p_{h,T}} + \mu$. We note that for all $h$ for which $1 - p_{h,T}$ is negligible, the trace distance bound of $\sqrt{p_{h,T}} + \mu$ is trivial and therefore satisfied. It follows by the triangle inequality that for all $h$, the trace distance between  $\sigma_{\Prov,h}^C$ and $\sigma_{\Prov',h}$ is $\leq p_{h,H} + \sqrt{p_{h,T}} + \mu$ (to complete the proof of the claim, recall from \eqref{eq:equivalenceofdistandstate} that $\T{\sigma_{\Prov,h}^C}{\sigma_{\Prov',h} } = \TV{D_{\Prov,h}^C}{ D_{\Prov',h} }$).
\\~\\
\textbf{Distance from perfect prover (between $\sigma_{\Prov,h}$ and $\sigma_{\Prov',h}$)}
By \Cl{cl:generalprover}, there exist unitaries $U_0,U$ such that the prover $\Prov$ is characterized by $(U_0,U)$ and the state $\sigma_{\Prov,h}$ can be created by following steps 1-4 in the statement of the claim. The state after application of $U_0$ (step 1 of Claim \ref{cl:generalprover}) is:  
\begin{equation}\label{eq:gentoperfstarting}
\sigma_{\Prov,h,0} = \sum_{k'} D_{\Ver,h}(k') \cdot U_0(\ket{0}\bra{0}^{\otimes e} \otimes \ket{k'}\bra{k'}) U_0^\dagger
\end{equation}
The state above is a mixed state over the verifier's choice of the function key $k'$, which is sampled according to the distribution $D_{\Ver,h}$ (see Section \ref{sec:notation} for a notation reminder). To create $\sigma_{\Prov,h}$ (i.e. the state resulting from the Hadamard round defined in \eqref{eq:verifierstate}), the prover will measure the second to last register (obtaining $y'\in \sY^n$), apply his attack $U$, measure all committed qubits and preimage registers in the Hadamard basis, send the results to the verifier and discard all other qubits. The verifier will decode the appropriate registers and discard all other measurement results (as described in Protocol \ref{prot:measprotocol}). Note that for all provers $\Prov, \Prov'$:
\begin{eqnarray}\label{eq:tracedistupperboundperfectextended}
\T{\sigma_{\Prov,h}}{\sigma_{\Prov',h}} \leq \T{\sigma_{\Prov,h,0}}{\sigma_{\Prov',h,0}}
\end{eqnarray}
This is because the operators described above which are applied to $\sigma_{\Prov,h,0}$ to create $\sigma_{\Prov,h}$ represent a CPTP map. 

We now construct a perfect prover $\Prov'$ and provide an upper bound for $\T{\sigma_{\Prov,h,0}}{ \sigma_{\Prov',h,0}}$. We begin by partitioning the state $\sigma_{\Prov,h,0}$ from \eqref{eq:gentoperfstarting} according to acceptance or rejection by the verifier in the test round:
\begin{equation}
\sigma_{\Prov,h,0} = \sum_{k'} D_{\Ver,h}(k') \cdot (\ket{\psi_{ACC,k'}} + \ket{\psi_{REJ,k'}})(\bra{\psi_{ACC,k'}} + \bra{\psi_{REJ,k'}})\otimes \ket{k'}\bra{k'}
\end{equation}
where 
\begin{equation}\label{eq:u0breakdown}
    U_0\ket{0}^{\otimes e}\ket{k'} = U_{0,k'}(\ket{0}^{\otimes e})\otimes \ket{k'} = (\ket{\psi_{ACC,k'}} + \ket{\psi_{REJ,k'}})\otimes \ket{k'}
\end{equation}
The first equality in \eqref{eq:u0breakdown} is given in \eqref{eq:u0trivialcontrol} in Claim \ref{cl:generalprover}, $\ket{\psi_{ACC,k'}}$ (resp. $\ket{\psi_{REJ,k'}}$) is the part of the state $U_{0,k'}(\ket{0}^{\otimes e})$ which will be accepted (resp. rejected) by the verifier in the test round for function choice $k'$, and $\braket{\psi_{ACC,k'}}{\psi_{REJ,k'}} = 0$. Consider the following state:
\begin{equation}\label{eq:sigmaperfect}
\sigma_{perfect} = \frac{1}{1 - p_{h,T}}\sum_{k'} D_{\Ver,h}(k') \cdot (\ket{\psi_{ACC,k'}})(\bra{\psi_{ACC,k'}})\otimes \ket{k'}\bra{k'}
\end{equation}
The trace distance between $\sigma_{perfect}$ and $\sigma_{\Prov,h,0}$ is at most $\sqrt{p_{h,T}}$. We show below that there exists a CPTP map $\mathcal{S}_0$, a perfect prover $\Prov'$ characterized by $(\mathcal{S}_0,U)$ and a negligible function $\mu$ such that for all $h$ for which $1 - p_{h,T}$ is non negligible, $\T{\sigma_{\Prov',h,0}}{\sigma_{perfect}} \leq \mu$. It follows by the triangle inequality that:
\begin{eqnarray}
\T{\sigma_{\Prov,h,0}}{\sigma_{\Prov',h,0}} &\leq& \T{\sigma_{\Prov,h,0}}{\sigma_{perfect}} + \T{\sigma_{perfect}}{\sigma_{\Prov',h,0}}\\
&\leq& \sqrt{p_{h,T}} + \mu
\end{eqnarray}
It then follows by \eqref{eq:tracedistupperboundperfectextended} that $\T{\sigma_{\Prov,h}}{\sigma_{\Prov',h}} \leq \sqrt{p_{h,T}} + \mu$, which completes the proof of Claim \ref{cl:generaltoperfect}.

We now describe the CPTP map $\mathcal{S}_0$. To do so, we will require the unitary $V$, which consists of first applying $U_0$, then applying the verifier's test round check (the algorithm CHK$_{\mathcal{F}}$) in superposition, and storing the result in an extra register. The result of applying $V$ is: 
\begin{eqnarray}\label{eq:inputtoV}
 V(\sum_{k'} D_{\Ver,h}(k')\cdot  \ket{0}\bra{0}^{\otimes e}\otimes \ket{k'}\bra{k'})V^{\dagger}
 \end{eqnarray}
 \begin{equation}
= \sum_{k'} D_{\Ver,h}(k')\cdot  (\ket{1}\ket{\psi_{ACC,k'}} + \ket{0}\ket{\psi_{REJ,k'}})(\ket{1}\ket{\psi_{ACC,k'}} + \ket{0}\ket{\psi_{REJ,k'}})^\dagger \otimes \ket{k'}\bra{k'}
\end{equation}
where the first register contains the result of the algorithm CHK$_{\mathcal{F}}$. If the first register is measured to obtain a bit $b$, the probability of obtaining $b = 1$ is $1 - p_{h,T}$. Therefore, for all $h$ for which $1 - p_{h,T}$ is non negligible, there exists a polynomial $n'_h$ such that if $V$ is applied $n'_h$ times to the input state in \eqref{eq:inputtoV} the probability of never obtaining $b = 1$ over $n'_h$ iterations is negligible. 

Let $n' = \max_h n'_h$, where the maximum is taken over all $h$ such that $1 - p_{h,T}$ is non negligible. The CPTP map $\mathcal{S}_0$ will apply the unitary $V$ followed by measurement of the first qubit to the input $\ket{0}^{\otimes e}\ket{k'}$ at most $n'$ times or until measuring $b = 1$. In the case that the measurement result $b = 1$ is never obtained, the CPTP map $\mathcal{S}_0$ applies a unitary operator which would be applied by an honest prover (described in Protocol \ref{prot:measprotocolprover}). Since the prover in Protocol \ref{prot:measprotocolprover} is perfect (see Lemma \ref{lem:measprotocolcorrectness}), it follows that the prover $\Prov'$ characterized by $(\mathcal{S}_0,U)$ is perfect.

Due to the choice of $n'$, it follows that there exists a negligible function $\mu$ such that for all $h$ for which $1 - p_{h,T}$ is non negligible
\begin{eqnarray}\label{eq:resultofsuperoperators0}
 \sigma_{\Prov',h,0} = \mathcal{S}_0(\sum_{k'} D_{\Ver,h}(k')\cdot  (\ket{0}\bra{0}^{\otimes e}\otimes \ket{k'}\bra{k'})) = (1-\mu)\sigma_{perfect} + \mu\sigma_{fail}
\end{eqnarray}
where $\sigma_{fail}$ is the state created in the case that all $n'$ applications of $V$ do not yield 1 as the measurement result. It follows by the convexity of trace distance that for all $h$ for which $1 - p_{h,T}$ is non negligible the trace distance between $\sigma_{\Prov',h,0}$ in \eqref{eq:resultofsuperoperators0} and $\sigma_{perfect}$ is $\leq\mu$.

\end{proof}

\subsubsection{Perfect to Trivial Prover (Proof of \Cl{cl:perfecttotrivialprover})}\label{sec:proofofperfecttotrivial}
In this section, we prove Claim \ref{cl:perfecttotrivialprover} as follows. First, in Claim \ref{cl:reductiontosinglestandard}, we prove a statement analogous to Lemma \ref{lem:singlesecurity} for the standard basis: we show that for an index $j$ such that $h_j = 0$, we can replace the prover's Hadamard round attack with an attack which acts $X$-trivially on the $j^{th}$ committed qubit to obtain the same distribution over measurements. The proof of Claim \ref{cl:reductiontosinglestandard} is quite straightforward, since the prover's Hadamard round attack has no effect on the standard basis measurement obtained by the verifier; this measurement is obtained from the commitment $y_j$. To prove Claim \ref{cl:perfecttotrivialprover}, we sequentially apply Lemma \ref{lem:singlesecurity} and Claim \ref{cl:reductiontosinglestandard} to each of the $n$ committed qubits, ending with an attack which commutes with standard basis measurement on all of the committed qubits. 

\begin{claim}\label{cl:reductiontosinglestandard}
Let $1\leq j\leq n$. Let $\mathcal{S} = \{B_{\tau}\}_{\tau}$ and $\mathcal{S}_j = \{B'_{j,x,\tau}\}_{x\in \{0,1\},\tau}$ be CPTP maps written in terms of their Kraus operators: 
\begin{equation}
B_{\tau} = \sum_{x,z\in\{0,1\}} X^xZ^z\otimes B_{jxz\tau}
\end{equation}
\begin{equation}
B'_{j,x,\tau} = \sum_{z\in\{0,1\}} Z^z\otimes B_{jxz\tau}
\end{equation}
where $B_{\tau}$ and $B'_{j,x,\tau}$ have been rearranged so that $X^xZ^z$ and $Z^z$ act on the $j^{th}$ qubit. For a unitary operator $U_0$, let $\Prov$ be a perfect prover characterized by $(U_0,\mathcal{S})$ (see Claim \ref{cl:generalprover} and notation \ref{notation:superoperator}). Let $\Prov_j$ be a perfect prover characterized by $(U_0,\mathcal{S}_j)$. If $h_j = 0$, $D_{\Prov,h}=D_{\Prov_j,h}$. 
\end{claim}
\begin{proofof}{\textbf{Claim \ref{cl:reductiontosinglestandard}}}
For convenience, we will assume that $h_1 = 0$ and we will prove the claim for $j = 1$. The proof is identical for all other values of $j$. The key observation is that we can change the CPTP map $\mathcal{S}$ to act in an arbitrary way on the first qubit (which is also the first committed qubit), as long as its action on the remaining qubits is unchanged. This is because the measurement of the first committed qubit will be ignored by the verifier; the verifier will obtain the standard basis measurement for the first qubit from the commitment $y_1$. In other words, the first committed qubit is traced out after application of $\mathcal{S}$ and standard basis measurement (it is not part of the distribution $D_{\Prov,h}$).

To prove the claim, we will show that the distribution over measurement results remains the same if $\mathcal{S} = \{B_{\tau}\}_{\tau}$ is replaced with $\mathcal{S}_1 = \{B'_{1,x,\tau}\}_{x\in \{0,1\},\tau}$. Our first step is to replace $\mathcal{S}$ with $\mathcal{S}^{(1)} = \{\frac{1}{\sqrt{2}}(Z^r\otimes \mcI)B_{\tau}\}_{r\in\{0,1\},\tau}$. As described above, this change has no impact since the measurement of the first committed qubit is ignored. Therefore, if we let $\Prov^{(1)}$ be a perfect prover characterized by $(U_0,\mathcal{S}^{(1)})$, $D_{\Prov,h} = D_{\Prov^{(1)},h}$. Next, we replace $\mathcal{S}^{(1)}$ with $\mathcal{S}^{(2)} = \{\frac{1}{\sqrt{2}}(Z^r\otimes \mcI)B_{\tau}(Z^r\otimes \mcI)\}_{r\in\{0,1\},\tau}$. Observe that the added $Z$ operator (acting prior to $B_{\tau}$) has no effect on the state of a perfect prover: it acts on the first committed qubit, which must be in a standard basis state since $h_1 = 0$ (the function key $k_1$ corresponds to a pair of injective functions $g_{k,0},g_{k,1}$ for which there is only one valid preimage for each commitment string $y$ - see equation \ref{eq:stateformatoverviewstandard} for more details). It follows that if we let $\Prov^{(2)}$ be the prover characterized by $(U_0,\mathcal{S}^{(2)})$, $D_{\Prov^{(1)},h} = D_{\Prov^{(2)},h}$. Finally, since $\mathcal{S}^{(2)}$ is followed by Hadamard basis measurement, we can apply Corollary \ref{corol:diagonalizingmeas} to replace $\mathcal{S}^{(2)}$ with $\mathcal{S}_1$, showing that $D_{\Prov^{(2)},h} = D_{\Prov_1,h}$ (therefore $D_{\Prov,h} = D_{\Prov_1,h}$). \end{proofof}

Using Lemma \ref{lem:singlesecurity} and Claim \ref{cl:reductiontosinglestandard}, we now prove Claim \ref{cl:perfecttotrivialprover}, which is copied here for reference: 
\\~\\
\begin{statement}{\textbf{Claim \ref{cl:perfecttotrivialprover}}}
\textit{
For all perfect provers $\Prov$, there exists a trivial prover $\hat{\Prov}$ such that for all $h$, $D_{\Prov,h}$ is computationally indistinguishable from $D_{\hat{\Prov},h}$.} 
\end{statement}

\begin{proofof}{\textbf{Claim \ref{cl:perfecttotrivialprover}}}
Fix a basis choice $h\in \{0,1\}^n$. Assume for convenience that for indices $1\leq i\leq t$, $h_i = 1$ and for indices $t+1\leq i\leq n$, $h_i = 0$. We apply Lemma \ref{lem:singlesecurity} $t$ times, beginning with the prover $\Prov$ (who is characterized by $(U_0,U)$). Let $\hat{\Prov}_0 = \Prov$ and let $\hat{\Prov}_j$ be the prover characterized by $(U_0, \{U_{x}\}_{x\in \{0,1\}^j})$: 
\begin{eqnarray}
U &=& \sum_{x,z\in \{0,1\}^j} X^xZ^z\otimes U^j_{xz}\\
U_{x} &=& \sum_{z\in\{0,1\}^j} Z^{z}\otimes U^j_{xz}
\end{eqnarray}
where $X^xZ^z$ acts on qubits $1,\ldots,j$ (which are also committed qubits $1,\ldots, j$). The $i^{th}$ application of Lemma \ref{lem:singlesecurity} shows that $D_{\hat{\Prov}_{i-1},h}$ is computationally indistinguishable from $D_{\hat{\Prov}_i,h}$. It follows by the triangle inequality that $D_{\Prov,h}$ is computationally indistinguishable from $D_{\hat{\Prov}_t,h}$. We complete the proof of Claim \ref{cl:perfecttotrivialprover} by applying Claim \ref{cl:reductiontosinglestandard} $n - t$ times, which tells us that $D_{\hat{\Prov}_{t},h} = D_{\hat{\Prov},h}$, where $\hat{\Prov} = \hat{\Prov}_{n}$ is a trivial prover characterized by $(U_0, \{U_x\}_{x\in\{0,1\}^n})$. 

\end{proofof}




\section{Extension of Measurement Protocol to a Verification Protocol for \BQP}\label{sec:meastoqpip}

In this section, we extend the measurement protocol (Protocol \ref{prot:measprotocol}) to a \QPIP$_0$, as described in Section \ref{sec:overviewmeastoqpip} of the overview. To do so, we will use the \QPIP$_1$ protocol in \cite{mf2016}. We begin by presenting this protocol.

\subsection{Morimae-Fitzsimons Protocol}\label{sec:mf2016}
Most of this description is taken directly from \cite{mf2016}; please see \cite{mf2016} for a full description of the protocol. Let $L$ be a language in \BQP. Since \BQP\ is contained in \QMA, for all instances $x$, there exists a local Hamiltonian $H$ such that 
\begin{enumerate}
\item if $x\in L$, then the ground energy of $H$ is $\leq a$
\item if $x\notin L$, then the ground energy of $H$ is $\geq b$
\end{enumerate}
where $b - a\geq \frac{1}{poly(|x|)}$. It is known that $H$ can be a 2-local Hamiltonian with only $X$ and $Z$ operators (\cite{xzhamiltonian}). 

Let us write the 2-local Hamiltonian as $H = \sum_S d_S S$, where $d_S$ is a real number and $S$ is a tensor product of Pauli operators, where only two operators are $Z$ or $X$ and others are $\mcI$. We define the rescaled Hamiltonian:
\begin{equation}\label{eq:rescaledhamiltonian}
H' =  \sum_S \pi_S P_S,
\end{equation}
where $\pi_S = \frac{|d_S|}{\sum_S |d_S|}\geq 0$ and $P_S = \frac{\mcI + sign(d_S)S}{2}$.

We now present the protocol:
\begin{protocol}\label{prot:mf}{\cite{mf2016}}
This protocol is used to verify that an instance $x\in L$ for a language $L\in $ \BQP. Let $H$ be the Hamiltonian which corresponds to $x$, and define $H'$ as in \ref{eq:rescaledhamiltonian}.
\begin{enumerate}
\item The verifier randomly chooses $S$ with probability $\pi_S$. 
\item The prover sends the verifier a state $\rho$, sending one qubit at a time.
\item The verifier measures $S$ by performing single qubit measurements of only two qubits of $\rho$ in the $X$ or $Z$ basis, discarding all other qubits of $\rho$ without measuring them.
\item The verifier computes the product of the measurement results. If the verifier obtains the result $-sign(d_S)$ the verifier accepts. 
\end{enumerate}
\end{protocol}

An honest prover would simply send the ground state of $H'$. 
\begin{thm}\label{thm:mfsoundness}{\cite{mf2016}}
Protocol \ref{prot:mf} is a \QPIP$_1$ for all languages in \BQP\ with completeness $c$ and soundness $s$, where $c - s$ is inverse polynomial in $|x|$. 
\end{thm}
\begin{proofof}{\Th{thm:mfsoundness}}
It was shown (in \cite{mns2015}) that the acceptance probability of Protocol \ref{prot:mf} is
\begin{equation}
p_{acc} = 1 - \frac{1}{\sum_S 2|d_S|}(\tr(H\rho) + \sum_S |d_S|)
\end{equation}
which is
\begin{equation}
p_{acc}\geq \frac{1}{2} - \frac{a}{\sum_S 2|d_S|}
\end{equation}
when $x\in L$, and
\begin{equation}
p_{acc}\leq \frac{1}{2} - \frac{b}{\sum_S 2|d_S|}
 \end{equation}
 when $x\notin L$. Their difference is $\frac{1}{poly(|x|)}$.
\end{proofof}

We require the following version of Protocol \ref{prot:mf}, which amplifies the completeness/soundness gap by repeating the protocol polynomially many times: 
\begin{protocol}\label{prot:mfamp}{\cite{mf2016}}
This protocol is used to verify that an instance $x\in L$ for a language $L\in $ \BQP. Let $H$ be the Hamiltonian which corresponds to $x$, and define $H'$ as in \ref{eq:rescaledhamiltonian}. Let $k'$ be a polynomial in $|x|$. 
\begin{enumerate}
\item The verifier randomly chooses $S_1,\ldots,S_{k'}$ independently, each with probability $\pi_{S_i}$. 
\item The prover sends the verifier a state $\rho'$, sending one qubit at a time.
\item The verifier measures each $S_i$ (for $1\leq i\leq k'$) by performing single qubit measurements of only two qubits of $\rho'$ in the $X$ or $Z$ basis, discarding all other qubits of $\rho'$ without measuring them.
\item The verifier computes the product of the measurement results for each $S_i$. If the verifier obtains the result $-sign(d_{S_i})$ more than $ \frac{1}{2} - \frac{a+b}{2\sum_S 2|d_S|}$ of the time the verifier accepts. 
\end{enumerate}
\end{protocol}
In Protocol \ref{prot:mfamp}, an honest prover would simply send $k'$ copies of the ground state of $H'$. 
\begin{thm}\label{thm:mfsoundnessamp}{\cite{mf2016}}
Protocol \ref{prot:mf} is a \QPIP$_1$ for all languages in \BQP\ with completeness $1 - \mu$ and soundness $\mu$, where $\mu$ is negligible in the size of the instance. 
\end{thm}
\begin{proof}
This theorem follows from \Th{thm:mfsoundness}; since $c - s$ is $\frac{1}{poly(|x|)}$ the verifier can distinguish the case where $x\in L$ from the case where $x \notin L$ with probability of error bounded to be exponentially small with only polynomially many repetitions.
\end{proof}

\subsection{Extending the Measurement Protocol}\label{sec:meastoqpipprotocol}

We now present the extension of the measurement protocol (Protocol \ref{prot:measprotocol}) to a \QPIP$_0$. To do this, we use the structure of Protocol \ref{prot:mfamp}, but we replace the prover sending qubits to the verifier to perform the measurement (i.e. steps 2 and 3 of Protocol \ref{prot:mfamp}) with Protocol \ref{prot:measprotocol}. Assume that in Protocol \ref{prot:mfamp}, $n$ qubits are sent from the prover to the verifier. The \QPIP\ protocol is as follows: 
\begin{protocol}{\textbf{Measurement Protocol \QPIP}}\label{prot:meastoqpip}
This protocol is used to verify that an instance $x\in L$ for a language $L\in $ \BQP. 
\begin{enumerate}
\item The verifier performs step 1 of Protocol \ref{prot:mfamp}, which partially defines a basis choice $h\in \{0,1\}^n$. For all undefined $h_i$, the verifier sets $h_i = 0$.  
\item The prover and verifier run Protocol \ref{prot:measprotocol} on basis choice $h$. The verifier accepts or rejects as specified in Protocol \ref{prot:measprotocol}.
\item In the case of a Hadamard round of Protocol \ref{prot:measprotocol}, the verifier performs step 4 of Protocol \ref{prot:mfamp} using the measurement results obtained from Protocol \ref{prot:measprotocol}. 
\end{enumerate}
\end{protocol}
We now prove the following theorem, which is the main result of this paper (and was stated earlier as Theorem \ref{thm:oneprover}): 
\begin{thm}\label{thm:lwetoqpip}
Protocol \ref{prot:meastoqpip} is a \QPIP$_0$ for all languages in \BQP\ with completeness negligibly close to $1$ and soundness negligibly close to $\frac{3}{4}$. 
\end{thm}
We require the completeness and soundness guarantees of Protocol \ref{prot:measprotocol}, copied below for reference. Both claims use notation from Section \ref{sec:notation}.
\\~\\
\begin{statement}{\textbf{Lemma \ref{lem:measprotocolcorrectness} \textit{Completeness of Protocol \ref{prot:measprotocol}}}}
\textit{For all $n$ qubit states $\rho$ and for all basis choices $h\in \{0,1\}^n$, the prover $\Prov$ described in Protocol \ref{prot:measprotocolprover} is a perfect prover ($\Prov$ is accepted by the verifier in a test run for basis choice $h$ with perfect probability). There exists a negligible function $\mu$ such that in the Hadamard round for basis choice $h$, the verifier accepts $\Prov$ with probability $\geq 1 - \mu$ and $\TV{D_{\Prov,h}^C}{D_{\rho,h}}\leq \mu$. }
\end{statement}

\begin{statement}{\textbf{Lemma \ref{lem:measprotocolsoundness} \textit{Soundness of Protocol \ref{prot:measprotocol}}}}
\textit{For a prover $\Prov$ in Protocol \ref{prot:measprotocol}, let $1 - p_{h,H}$ be the probability that the verifier accepts $\Prov$ on basis choice $h$ in the Hadamard round and $1 - p_{h,T}$ be the probability that the verifier accepts $\Prov$ in the test round. There exists a state $\rho$, a prover $\Prov'$ and a negligible function $\mu$ such that for all $h$, $\TV{D^C_{\Prov,h}}{ D_{\Prov',h}}\leq p_{h,H} + \sqrt{p_{h,T}} + \mu$ and $D_{\Prov',h}$ is computationally indistinguishable from the distribution $D_{\rho,h}$ which results from measuring $\rho$ in the basis determined by $h$.}
\end{statement}
\newline
\begin{proofof}{\textbf{\Th{thm:lwetoqpip}}}
We will require some notation. For all provers $\Prov$, let $E_{\Prov,h}^H$ be the event that the verifier accepts in a Hadamard round of Protocol \ref{prot:measprotocol} with basis choice $h$ while interacting with $\Prov$, let $E_{\Prov,h}^T$ be the same event in a test round, and let $E_{\Prov,h}$ be the event that the verifier accepts in step 3 of Protocol \ref{prot:meastoqpip}. Let $v_h$ be the probability that the verifier chooses basis choice $h$ in step 1 of Protocol \ref{prot:meastoqpip}. As a reminder, in step 3 of Protocol \ref{prot:meastoqpip} (which is step 4 of Protocol \ref{prot:mfamp}), the verifier determines whether or not to accept by computing the product of relevant measurement results; the verifier's decision is a function of the basis choice $h$, the measurement result and the \BQP\ instance, but we leave off the dependence on the \BQP\ instance for convenience. For a distribution $D$ over $n$ bit strings and a basis choice $h$, let $\tilde{p}_{h}(D)$ be the probability that the verifier rejects an $n$ bit string drawn from $D$ for basis choice $h$ in step 3 of Protocol \ref{prot:meastoqpip}. 
\\~\\
\textbf{Completeness }
Recall that in Protocol \ref{prot:mfamp}, an honest prover sends polynomially many copies of the ground state for the Hamiltonian corresponding to an instance $x \in L$ (where $L\in $ \BQP). Let the entire state sent by the prover be $\rho$, and assume it contains $n$ qubits. To compute the completeness parameter of Protocol \ref{prot:meastoqpip}, we will consider the prover $\Prov$ for the state $\rho$, as described in Protocol \ref{prot:measprotocolprover}, and upper bound the probability that the verifier rejects in Protocol \ref{prot:meastoqpip}. To do this, we need to upper bound the probability that the verifier rejects in step 2 of Protocol \ref{prot:meastoqpip} (i.e. in Protocol \ref{prot:measprotocol}) or the verifier rejects in step 3 of Protocol \ref{prot:meastoqpip}:  
\begin{eqnarray}
1 - c &=& \frac{1}{2}\sum_{h\in\{0,1\}^n} v_h(\Pr[\bar{E}^T_{\Prov,h}] + \Pr[\bar{E}^H_{\Prov,h}\cup \bar{E}_{\Prov,h}]) \\
&\leq& \frac{1}{2}\sum_{h\in\{0,1\}^n} v_h(\Pr[\bar{E}^H_{\Prov,h}] +  \Pr[\bar{E}_{\Prov,h}]) \\
&\leq& \frac{1}{2}\mu +  \frac{1}{2}\sum_{h\in\{0,1\}^n} v_h \Pr[\bar{E}_{\Prov,h}] \label{eq:measqpipcompleteness}
\end{eqnarray}
The last two expressions follow due to Lemma \ref{lem:measprotocolcorrectness}: we know that for all basis choices $h\in \{0,1\}^n$, the prover $\Prov$ described in Protocol \ref{prot:measprotocolprover} is accepted by the verifier in a test round with perfect probability and is accepted in a Hadamard round with probability $\geq 1 - \mu$ for a negligible function $\mu$. $\Pr[\bar{E}_{\Prov,h}]$ is the probability that the verifier rejects the distribution $D^C_{\Prov,h}$ for basis choice $h$ in step 3 of Protocol \ref{prot:meastoqpip}:
\begin{equation}\label{eq:completenessprobplugin}
\Pr[\bar{E}_{\Prov,h}] = \tilde{p}_h(D^C_{\Prov,h})
\end{equation}
Recall from Section \ref{sec:notation} that for a density matrix on $n$ qubits and basis choice $h$, we let $D_{\rho,h}$ be the distribution obtained by measuring $\rho$ in the basis corresponding to $h$. It follows by Lemma \ref{lem:projectionprob} and Lemma \ref{lem:measprotocolcorrectness} that
\begin{equation}\label{eq:completenessprobdifference}
\tilde{p}_h(D^C_{\Prov,h}) - \tilde{p}_h(D_{\rho,h})  \leq \TV{D^C_{\Prov,h}}{D_{\rho,h}} \leq \mu
\end{equation}
Due to the completeness parameter of Protocol \ref{prot:mfamp} (see Section \ref{sec:mf2016}), there exists a negligible function $\mu_C$ such that:
\begin{equation}\label{eq:correctcompleteness}
\sum_h v_h\tilde{p}_h(D_{\rho,h}) \leq \mu_C
\end{equation}
We now use \eqref{eq:completenessprobplugin}, \eqref{eq:completenessprobdifference} and \eqref{eq:correctcompleteness} to wrap up the calculation of the completeness parameter $c$ (continuing from \eqref{eq:measqpipcompleteness}): 
\begin{eqnarray}
1 - c &\leq& \frac{1}{2} \mu +  \frac{1}{2}\sum_{h\in\{0,1\}^n} v_h \tilde{p}_h(D_{\Prov,h}^C)\\
&\leq& \mu + \frac{1}{2}\mu_C \label{eq:secondtolastcompleteness}
\end{eqnarray}
Therefore, the completeness parameter $c$ is negligibly close to $1$. 
\\~\\
\textbf{Soundness }
To compute the soundness parameter, we will fix an arbitrary prover $\Prov$ and upper bound the probability that the verifier accepts in Protocol \ref{prot:meastoqpip} for an instance $x\notin L$. To do so, we need to upper bound the probability that the verifier accepts in step 2 of Protocol \ref{prot:meastoqpip} (i.e. in Protocol \ref{prot:measprotocol}) and the verifier accepts in step 3 of Protocol \ref{prot:meastoqpip}. The intuition here is that as long as there exists a state $\rho$ such that for all $h$, $D_{\Prov,h}^C$ is close (computationally) to $D_{\rho,h}$, the soundness parameter should be close to the soundness parameter of Protocol \ref{prot:mfamp}, which is negligible. We will rely on Lemma \ref{lem:measprotocolsoundness} (we also use the same notation used in Lemma \ref{lem:measprotocolsoundness}):
\begin{eqnarray}
s &=& \sum_{h\in\{0,1\}^n} v_h(\frac{1}{2}\Pr[E_{\Prov,h}^T] + \frac{1}{2}\Pr[E_{\Prov,h}^H\cap E_{\Prov,h}])\\
&=& \sum_{h\in\{0,1\}^n} v_h(\frac{1}{2}(1-p_{h,T}) + \frac{1}{2}\Pr[E_{\Prov,h}^H]\Pr[E_{\Prov,h}|E_{\Prov,h}^H])\\
&=& \sum_{h\in\{0,1\}^n} v_h(\frac{1}{2}(1-p_{h,T}) + \frac{1}{2}(1-p_{h,H})(1-\tilde{p}_h(D_{\Prov,h}^C)))\label{eq:measqpipsoundness}
\end{eqnarray}
where the last equality follows because $\Pr[E_{\Prov,h}|E_{\Prov,h}^H]$ is the probability that the verifier accepts a string drawn from the distribution $D^C_{\Prov,h}$ for basis choice $h$ in step 3 of Protocol \ref{prot:meastoqpip}. 

\Le{lem:measprotocolsoundness} guarantees the existence of a state $\rho$, a prover $\Prov'$ and a negligible function $\mu$ such that for all $h$, $\TV{D^C_{\Prov,h}}{ D_{\Prov',h}}\leq p_{h,H} + \sqrt{p_{h,T}} + \mu$ and $D_{\Prov',h}$ is computationally indistinguishable from $D_{\rho,h}$. By Lemma \ref{lem:projectionprob} and \Le{lem:measprotocolsoundness}:
\begin{equation}
\tilde{p}_h(D_{\Prov',h}) - \tilde{p}_h(D_{\Prov,h}^C) \leq \TV{D_{\Prov,h}^C}{D_{\Prov',h}} \leq p_{h,H} + \sqrt{p_{h,T}} + \mu
\end{equation}
We now return to the calculation of the soundness parameter of Protocol \ref{prot:meastoqpip} in \eqref{eq:measqpipsoundness}: 
\begin{eqnarray}
s &\leq&  \sum_{h\in \{0,1\}^n} v_h(\frac{1}{2}(1-p_{h,T}) + \frac{1}{2}(1-p_{h,H})(1-\tilde{p}_h(D_{\Prov',h}) + p_{h,H} + \sqrt{p_{h,T}} + \mu))\\
&\leq&  \frac{1}{2}\mu + \frac{1}{2}\sum_{h\in \{0,1\}^n} v_h(1-p_{h,T} + (1-p_{h,H})(p_{h,H} + \sqrt{p_{h,T}})) + \frac{1}{2}\sum_{h\in\{0,1\}^n} v_h(1-\tilde{p}_h(D_{\Prov',h}))\nonumber\\
&\leq& \frac{1}{2}\mu  + \frac{3}{4} + \frac{1}{2}\sum_{h\in\{0,1\}^n} v_h(1-\tilde{p}_h(D_{\Prov',h}))\label{eq:secondtolastsoundness}
\end{eqnarray}
where the last inequality follows since for all $p\in[0,1], \sqrt{p} - p \leq \frac{1}{4}$. 

Next, Lemma \ref{lem:measprotocolsoundness} guarantees that for all $h$, $D_{\Prov',h}$ and $D_{\rho,h}$ are computationally indistinguishable. It follows that for all $h$:  
\begin{equation}\label{eq:compindtosound}
\tilde{p}_h(D_{\rho,h}) - \tilde{p}_h(D_{\Prov',h}) \leq \mu_h
\end{equation}
To see this implication, assume there did exist an $h\in \{0,1\}^n$ such that the difference in \eqref{eq:compindtosound} was non negligible. Then $D_{\Prov',h}$ and $D_{\rho,h}$ could be distinguished by computing whether or not the verifier would reject for basis choice $h$ in step 3 of Protocol \ref{prot:meastoqpip}, which is step 4 of Protocol \ref{prot:mfamp}. This is because the computational indistinguishability of $D_{\Prov',h}$ and $D_{\rho,h}$ holds even if $h$ is known; the indistinguishability is due to the hardcore bit and injective invariance properties of the extended trapdoor claw-free family (Definition \ref{def:extendedtrapdoorclawfree}). 

Due to the soundness parameter of Protocol \ref{prot:mfamp}, we also know that there exists a negligible function $\mu_S$ such that: 
\begin{equation}
\sum_{h\in\{0,1\}^n} v_h (1-\tilde{p}_h(D_{\rho,h})) \leq \mu_S
\end{equation}
We return to calculating the soundness parameter of the \QPIP\ (continuing from \eqref{eq:secondtolastsoundness}):
\begin{eqnarray}
s &\leq&  \frac{1}{2}\mu + \frac{3}{4} +  \frac{1}{2}\sum_{h\in \{0,1\}^n} v_h (\mu_h + 1 - \tilde{p}_h(D_{\rho,h}))\\
&\leq& \frac{1}{2}\mu + \frac{3}{4} + \frac{1}{2}\max_{h\in\{0,1\}^n} \mu_h + \frac{1}{2}\mu_S
\end{eqnarray}
Therefore, the soundness parameter of Protocol \ref{prot:meastoqpip} is negligibly close to $\frac{3}{4}$.  
\end{proofof}

\section{Extended Trapdoor Claw-Free Family from LWE}\label{sec:lweextended}

\subsection{Parameters}\label{sec:lweparam}
We will use the same parameters used in \cite{oneproverrandomness}. Let $\lambda$ be the security parameter. All other parameters are functions of $\lambda$. Let $q\geq 2$ be a prime integer. 
Let $\ell,n,m,w\geq 1$ be polynomially bounded functions of $\lambda$ and $B_L, B_V, B_P$ be positive integers such that the following conditions hold:
\begin{equation}\label{eq:assumptionscopy}
    \begin{minipage}{0.9\textwidth}
\begin{enumerate}
\item  $n = \Omega(\ell \log q)$ and $m = \Omega(n\log q)$,
\item $w=n\lceil \log q\rceil$,
\item $B_P = \frac{q}{2C_T\sqrt{mn\log q}}$, for $C_T$ the universal constant in Theorem~\ref{thm:trapdoor},
\item $ 2\sqrt{n} \leq B_L < B_V < B_P$,
\item The ratios $\frac{B_P}{B_V}$ and $\frac{B_V}{B_L}$ are both super-polynomial  in $\lambda$.
\end{enumerate}
    \end{minipage}
  \end{equation}
Given a choice of parameters satisfying all conditions above, we describe the function family $\mathcal{F}_{\lwe}$ (taken from \cite{oneproverrandomness}). Let $\sX = \mZ_q^n$ and $\sY = \mZ_q^m$. 
The key space is $\mathcal{K}_{\mathcal{F}_{\lwe}} = \mZ_q^{m\times n} \times \mZ_q^m$. For $b\in \{0,1\}$, $x\in \sX$ and key $k = (\*A,\*A\*s + \*e)$,  the density $f_{k,b}(x) $ is defined as
\begin{equation}\label{eq:defFLWE}
  \forall y \in \sY,\quad   (f_{k,b}(x))(y) = D_{\mZ_q^m,B_P}(y - \*Ax - b\cdot \*A\*s)\;,
\end{equation}
where the definition of $D_{\mZ_q^m,B_P}$ is given in~\eqref{eq:d-bounded-def}. Note that $f_{k,b}$ is well-defined given $k$, as for our choice of parameters $k$ uniquely specifies $s$. The following theorem was proven in \cite{oneproverrandomness} (for an intuitive description of the proof of this theorem, see page 4 of \cite{oneproverrandomness}):
\\~\\
\begin{statement}{\textbf{Theorem 22 } \cite{oneproverrandomness}}
\textit{
For any choice of parameters satisfying the conditions~\eqref{eq:assumptionscopy}, the function family $\mathcal{F}_{\lwe}$ is a noisy trapdoor claw-free family under the hardness assumption $\lwe_{\ell,q,D_{\mZ_q,B_L}}$.}
\end{statement}
\newline
We will prove the following theorem:
\begin{thm}\label{thm:lwetcfextended}
For any choice of parameters satisfying the conditions~\eqref{eq:assumptionscopy}, the function family $\mathcal{F}_{\lwe}$ is an extended trapdoor claw-free family under the hardness assumption $\lwe_{\ell,q,D_{\mZ_q,B_L}}$. 
\end{thm}
We begin by providing a trapdoor injective family and then use this family to show that $\mathcal{F}_{\lwe}$ is an extended trapdoor claw-free function family.

\subsection{Trapdoor Injective Family from LWE}
\label{sec:lweinjective}
We now describe the trapdoor injective family $\mathcal{G}_{\lwe}$. Let $\sX = \mZ_q^n$ and $\sY = \mZ_q^m$. 
The key space is $\mathcal{K}_{\mathcal{G}_{\lwe}} = \mZ_q^{m\times n} \times \mZ_q^m$. For $b\in \{0,1\}$, $x\in \sX$ and key $k = (\*A,\*u)$,  the density $g_{k,b}(x) $ is defined as
\begin{equation}\label{eq:defGLWE}
  \forall y \in \sY,\quad   (g_{k,b}(x))(y) = D_{\mZ_q^m,B_P}(y - \*Ax - b\cdot \*u)\;,
\end{equation}
where the definition of $D_{\mZ_q^m,B_P}$ is given in~\eqref{eq:d-bounded-def}. The three properties required for a trapdoor injective family, as specified in Definition~\ref{def:trapdoorinjective}, are verified in the following subsections, providing a proof of the following theorem.

\begin{thm}\label{thm:lwetif}
For any choice of parameters satisfying the conditions~\eqref{eq:assumptionscopy}, the function family $\mathcal{G}_{\lwe}$ is a trapdoor injective family under the hardness assumption $\lwe_{\ell,q,D_{\mZ_q,B_L}}$. 
\end{thm}

\subsubsection{Efficient Function Generation}

GEN$_{\mathcal{G}_{\lwe}}$ is defined as follows. First, the procedure samples a random $\*A\in \mZ_q^{m\times n}$, together with trapdoor information $t_{\*A}$. This is done using the procedure $\GenTrap(1^n,1^m,q)$ from Theorem~\ref{thm:trapdoor}. The trapdoor allows the evaluation of an inversion algorithm $\Invert$  that, on input $\*A$, $t_{\*A}$ and $b=\*A\*s + \*e$ returns $\*s$ and $\*e$ as long as $\|\*e\|\leq \frac{q}{C_T\sqrt{n\log q}}$. Moreover, the distribution on matrices $\*A$ returned by $\GenTrap$ is negligibly close to the uniform distribution on $\mZ_q^{m\times n}$.

Next, the sampling procedure selects $\*u\in \mZ_q^m$ uniformly at random. By using the trapdoor $t_{\*A}$, the sampling procedures checks if there exist $\*s,\*e$ such that $\|\*e\|\leq \frac{q}{C_T\sqrt{n\log q}}$ and $\*A\*s + \*e = \*u$. If so, the sampling algorithm discards $\*u$ and samples $\*u$ again. Since $\*u$ is discarded with negligible probability, the distribution over $\*u$ is negligibly close to the uniform distribution. GEN$_{\mathcal{G}_{\lwe}}$ returns $k = (\*A,\*u)$ and $t_k = t_{\*A}$.

\subsubsection{Trapdoor Injective Functions}\label{sec:trapdoorinjectivereq}

It follows from~\eqref{eq:defGLWE} and the definition of the distribution $D_{\mZ_q^m,B_P}$ in~\eqref{eq:d-bounded-def} that for any key $k=(\*A,\*u)\in \mathcal{K}_{\mathcal{G}_{\lwe}}$ and for all $x\in \sX$,
\begin{eqnarray}
\supp(g_{k,0}(x)) &=& \big\{y = \*Ax + \*e_0\,| \; \|\*e_0\|\leq B_P\sqrt{m}\big\}\;,\\
\supp(g_{k,1}(x)) &=& \big\{y = \*Ax + \*e_0 + \*u \,| \;\|\*e_0\|\leq B_P\sqrt{m}\big\}\;.
\end{eqnarray}
Since there do not exist $\*s,\*e$ such that $\|\*e\|  \,\leq\, \frac{q}{C_T\sqrt{n\log q}}$ and $\*u = A\*s + \*e$, the intersection $\supp(g_{k,0}(x))\cap \supp(g_{k,1}(x))$ is empty as long as 
\begin{equation}\label{eq:trapdoorinjectiverequirement}
		 B_P \leq \frac{q}{2C_T\sqrt{mn\log q}}\;.
\end{equation}
The procedure $\textrm{INV}_{\mathcal{G}_{\lwe}}$ takes as input the trapdoor $t_{\*A}$ and $y\in \sY$. It first runs the algorithm $\Invert$ on input $y$. If $\Invert$ outputs $\*s_0,\*e_0$ such that $y = \*A\*s_0 + \*e_0$ and $  \|\*e_0\|  \,\leq\, B_P\sqrt{m}$, the procedure $\textrm{INV}_{\mathcal{F}_{\lwe}}$ outputs $(0,\*s_0)$. Otherwise, it runs the algorithm $\Invert$ on input $y - \*u$ to obtain $\*s_0,\*e_0$ and outputs $(1,\*s_0)$ if $y - \*u = \*A\*s_0 + \*e_0$. Using Theorem~\ref{thm:trapdoor}, this procedure returns the unique correct outcome provided $y - b\cdot\*u= \*A\*s_0+\*e_0$ for some $\*e_0$ such that $  \|\*e_0\|  \,\leq\, \frac{q}{C_T\sqrt{n\log q}}$. Due to \eqref{eq:trapdoorinjectiverequirement}, this condition is satisfied for all $y\in \supp(f_{k,b}(x))$.

\subsubsection{Efficient Range Superposition} 
The procedures CHK$_{\mathcal{G}_{\lwe}}$ and SAMP$_{\mathcal{G}_{\lwe}}$ are the same as the procedures CHK$_{\mathcal{F}_{\lwe}}$ and SAMP$_{\mathcal{F}_{\lwe}}$.

\subsection{Injective Invariance}
We now show that $\mathcal{F}_{\textrm{LWE}}$ as given in \eqref{eq:defFLWE} is injective invariant with respect to $\mathcal{G}_{\textrm{LWE}}$ (see Definition \ref{def:injectiveinvariant}). To show this, we need to show that for all \BQP\ attackers $\mathcal{A}$, the distributions produced by $\textrm{GEN}_{\mathcal{F}_{\textrm{LWE}}}$ and $\textrm{GEN}_{\mathcal{G}_{\textrm{LWE}}}$ are computationally indistinguishable. This is equivalent to proving the hardness of LWE (as defined in Definition \ref{def:lwehardness}) with a binary secret, as seen in the following lemma:  

\begin{lem}\label{lem:lwebinaryindist}
Assume a choice of parameters satisfying the conditions \eqref{eq:assumptionscopy}. Assume the hardness assumption $\lwe_{\ell,q,D_{\mZ_q,B_L}}$ holds. Then the distributions
\begin{equation}\label{eq:initialbinarylwe}
\mathcal{D}_{0} = ((\*A, \*A\*s + \*e)\leftarrow  \textrm{GEN}_{\mathcal{F}_{\textrm{LWE}}}(1^{\lambda}))
\end{equation}
and
\begin{equation}\label{eq:finalbinarylwedist}
\mathcal{D}_{1} = ((\*A, \*u)\leftarrow \textrm{GEN}_{\mathcal{G}_{\textrm{LWE}}}(1^{\lambda}))
\end{equation}
are computationally indistinguishable. 

\end{lem}
The hardness of LWE with a binary secret is well studied, and the above lemma is implied by several results, starting with \cite{robustnesslwe} and improved in \cite{BLPRS13}. To be precise, it is also immediately implied by Theorem B.5 of \cite{lwr}. 

\subsection{Extended Trapdoor Claw-Free Family}
We have already shown that $\mathcal{F}_{\textrm{LWE}}$ is injective invariant. To show that $\mathcal{F}_{\lwe}$ is an extended trapdoor claw-free family, we now prove the second condition (the hardcore bit condition) of Definition \ref{def:extendedtrapdoorclawfree}. The two proofs in this section use the same ideas as the proofs of Lemmas 4.3 and 4.4 in \cite{oneproverrandomness}; for this reason, it is best to read those proofs first. We will be using the same notation as \cite{oneproverrandomness}: let $\sX=\mZ_q^n$, $w=n\lceil \log q \rceil$ and $\inj:\sX\to\{0,1\}^w$ be the map such that $\inj(x)$ returns the binary representation of $x\in\sX$. 

The key point, which we prove in Lemma \ref{lem:lweadaptivehardcoreG}, is that the inner product appearing in the definition of $H'_{k,d}$ (in condition 2 of Definition \ref{def:extendedtrapdoorclawfree}) is equal to the inner product $\hat{d}\cdot s$ (for $\hat{d}\in\{0,1\}^n$) if $d = J(\hat{\*d})$. We first show in Lemma \ref{lem:lweadaptiveleakageG} that producing an inner product $\hat{d}\cdot s$ is computationally difficult given $\*A,\*A\*s + \*e$. Next, in Lemma \ref{lem:lweadaptivehardcoreG}, we use Lemma \ref{lem:lweadaptiveleakageG} to prove condition 2 of Definition \ref{def:extendedtrapdoorclawfree}. 

We will be using Lemma 4.6 (\cite{oneproverrandomness}) in the following proof, so it is included here for reference:
\begin{lem}[\cite{oneproverrandomness}]\label{lem:hardcore-1}
Let $q$ be a prime, $\ell,n\geq 1$ integers, and $\*C\in \mZ_q^{\ell\times n}$ a uniformly random matrix. With probability at least $1-q^\ell\cdot 2^{-\frac{n}{8}}$ over the choice of $\*C$ the following holds. For a fixed $\*C$, all $\*v\in\mZ_q^\ell$ and $\hat{d}\in \{0,1\}^n\setminus\{0^n\}$, the distribution of $(\hat{d}\cdot s \bmod 2)$, where $s$ is uniform in $\{0,1\}^n$ conditioned on $\*C\*s = \*v$, is within  statistical distance $O(q^{\frac{3\ell}{2}} \cdot 2^{-\frac{n}{40}})$ of the uniform distribution over $\{0,1\}$. 
\end{lem}

\begin{lem}\label{lem:lweadaptiveleakageG}
Assume a choice of parameters satisfying the conditions~\eqref{eq:assumptionscopy}. Assume the hardness assumption $\lwe_{\ell,q,D_{\mZ_q,B_L}}$ holds. For all $\hat{d}\in\{0,1\}^n\setminus\{0^n\}$, the distributions
\begin{eqnarray}\label{eq:l14-1G}
\mathcal{D}_0 \,=\, \big( (\*A,\*A\*s+\*e) \leftarrow \textrm{GEN}_{\mathcal{F}_{\lwe}}(1^{\lambda}) ,\; \hat{d} \cdot s \mod 2\big) 
\end{eqnarray}
and
\begin{eqnarray}\label{eq:l14-2G}
\mathcal{D}_1 \,=\, \big( (\*A,\*A\*s+\*e) \leftarrow \textrm{GEN}_{\mathcal{F}_{\lwe}}(1^{\lambda}) , \; r \big)\;, 
\end{eqnarray}
where $r\leftarrow_U \{0,1\}$, are computationally indistinguishable. 
\end{lem}
\begin{proof}
This proof is a simpler version of the proof of Lemma 4.4 in \cite{oneproverrandomness}. We first transition from $\mathcal{D}_{0}$ to the following computationally indistinguishable distribution:
\begin{equation}
\mathcal{D}^{(1)} = (\*B\*C+ \*F, \*B\*C\*s + \*e, \hat{d}\cdot s\mod 2) 
\end{equation}
where $\tilde{\*A} = \*B\*C + \*F \leftarrow \lossy(1^n,1^m,1^\ell,q,D_{\mZ_q,B_L})$ is sampled from a lossy sampler. The transition from $\mathcal{D}_0$ to $\mathcal{D}^{(1)}$ requires two steps: the first step is to replace $\*A$ with a lossy matrix $\tilde{\*A} = \*B\*C + \*F$ to obtain a computationally indistinguishable distribution, and the second step is to remove the term $\*F\*s$ from the lossy LWE sample $\tilde{\*A}\*s + \*e$. To justify the second step, note that $\*s$ is binary and the entries of $\*F$ are taken from a $B_L$-bounded distribution, implying that $\|\*F\*s\|\leq n\sqrt{m}B_L$. Lemma \ref{lem:distributiondistance} then implies that the removal of the term $\*F\*s$ results in a distribution within negligible statistical distance. For a more detailed description of this transition, see the proof of Lemma 4.4 in \cite{oneproverrandomness} (the transition from $\mathcal{D}_0$ to $\mathcal{D}^{(1)}$ is the same as the transition from (31) to (43) in the proof of Lemma 4.4 in \cite{oneproverrandomness}).

Next, we apply Lemma \ref{lem:hardcore-1} (Lemma 4.6 from \cite{oneproverrandomness}) to $\mathcal{D}^{(1)}$ to replace $\hat{d}\cdot s\mod 2$ with a uniformly random bit $r$, resulting in the following statistically indistinguishable distribution:
\begin{equation}
\mathcal{D}^{(2)} = (\*B\*C + \*F, \*B\*C\*s + \*e, r) 
\end{equation}
Computational indistinguishability between $\mathcal{D}^{(2)}$ and $\mathcal{D}_1$ follows similarly to between $\mathcal{D}^{(1)}$ and $\mathcal{D}_0$.
\end{proof}

The following lemma implies the second condition of Definition \ref{def:extendedtrapdoorclawfree}, by letting $d$ in Definition \ref{def:extendedtrapdoorclawfree} be $J(\hat{\*d})$, for any $\hat{d}\in\{0,1\}^n\setminus\{0^n\}$.
\begin{lem}\label{lem:lweadaptivehardcoreG} 
Assume a choice of parameters satisfying the conditions~\eqref{eq:assumptionscopy}. Assume the hardness assumption $\lwe_{\ell,q,D_{\mZ_q,B_L}}$ holds. Let $s\in\{0,1\}^n$ and for $d\in\{0,1\}^w$ let \footnote{We write the set as $H'_{s,d}$ instead of $H'_{k,d}$ to emphasize the dependence on $s$.}
\begin{eqnarray}
H'_{s,d} &=& \big\{d\cdot(\inj(x)\oplus \inj( x - \*s ))\,|\; x\in \sX  \big\}\;,\label{eq:def-h-prime}\;.
\end{eqnarray}
Then for all $\hat{d}\in\{0,1\}^n\setminus\{0^n\}$ and for any quantum polynomial-time procedure 
$$\mathcal{A}:\, \mZ_q^{m\times n} \times \mZ_q^m \,\to\,\{0,1\}$$
 there exists a negligible function $\mu(\cdot)$ such that 
\begin{equation}\label{eq:lwebitattackeradvantageG}
\Big|\Pr_{(\*A,\*A\*s+\*e)\leftarrow \textrm{GEN}_{\mathcal{F}_{\lwe}}(1^{\lambda})}\big[\mathcal{A}(\*A,\*A\*s+\*e) \in H'_{s,J(\hat{\*d})}\big] - \frac{1}{2}\Big| \,\leq\, \mu(\lambda)\;.
\end{equation} 
\end{lem}

\begin{proof}
This proof is very similar to the proof of Lemma 4.3 in \cite{oneproverrandomness}. The proof is by contradiction. Assume that there exists $\hat{d}\in\{0,1\}^n$ and a quantum polynomial-time procedure $\mathcal{A}$ such that the left-hand side of~\eqref{eq:lwebitattackeradvantageG} is at least some non-negligible function $\eta(\lambda)$. We derive a contradiction by showing that for $\hat{d}$, the two distributions $\mathcal{D}_0$ and $\mathcal{D}_1$ in Lemma~\ref{lem:lweadaptiveleakageG} are computationally distinguishable, giving a contradiction. 

Let $(\*A,\*A\*s+\*e) \leftarrow \textrm{GEN}_{\mathcal{F}_{\lwe}}(1^{\lambda})$. To link $\mathcal{A}$ to the distributions in Lemma~\ref{lem:lweadaptiveleakageG} we relate the inner product condition in~\eqref{eq:def-h-prime} to the inner product $\hat{d} \cdot s$ appearing in~\eqref{eq:l14-1G}. This is based on the following claim.

\begin{claim}\label{cl:dependenceonsecretG}
For all $x\in \sX, \hat{d}\in\{0,1\}^n$ and $s\in\{0,1\}^n$ the following equality holds:
\begin{equation}\label{eq:def-hatdG}
J(\hat{\*d})\cdot(\inj(x)\oplus \inj( x - \*s )) \,=\,  \hat{d}\cdot s \;.
\end{equation}
\end{claim}

\begin{proof}
We do the proof in case $n=1$ and $w=\lceil\log q\rceil$, as the case of general $n$ follows by linearity. In this case $s$ is a single bit. If $s=0$ then both sides of~\eqref{eq:def-hatdG} evaluate to zero, so the equality holds trivially. If $s = 1$, then the least significant bit of $\inj(x)\oplus \inj( x - \*s )$ is $s$ and the least significant bit of $J(\hat{\*d}) = \hat{d}$. Since the remaining $w - 1$ bits of $J(\hat{\*d})$ are 0, the claim follows. 
\end{proof}\\

To conclude we construct a distinguisher $\mathcal{A}'$ between the distributions $\mathcal{D}_0$ and $\mathcal{D}_1$ in Lemma \ref{lem:lweadaptiveleakageG}. Consider two possible distinguishers, $\mathcal{A}'_u$ for $u\in \{0,1\}$. Given a sample $((\*A,\*A\*s+\*e),t)$, $\mathcal{A}'_u$ computes $c = \mathcal{A}(\*A,\*A\*s + \*e)$ and returns $0$ if $c=t\oplus u$, and $1$ otherwise. The sum of the advantages of $\mathcal{A}'_0$ and $\mathcal{A}'_1$ is:
\begin{eqnarray}
\sum_{u\in \{0,1\}} \Big|\Pr_{((\*A,\*A\*s+\*e),\hat{d}\cdot s)\leftarrow \mathcal{D}_0}\big[\mathcal{A}_u'((\*A,\*A\*s+\*e),\hat{d}\cdot s)=0\big]-\Pr_{((\*A,\*A\*s+\*e),r)\leftarrow \mathcal{D}_1}\big[\mathcal{A}_u'((\*A,\*A\*s+\*e),r)=0\big]\Big|\nonumber
\end{eqnarray}
\begin{eqnarray}
&=& \sum_{u\in \{0,1\}} \Big|\Pr_{(\*A,\*A\*s+\*e)\leftarrow \textrm{GEN}_{\mathcal{F}_{\lwe}}(1^{\lambda})}\big[\mathcal{A}(\*A,\*A\*s+\*e)=\hat{d}\cdot s \oplus u\big]-\Pr_{((\*A,\*A\*s+\*e),r)\leftarrow \mathcal{D}_1}\big[\mathcal{A}(\*A,\*A\*s+\*e)=r\oplus u\big]\Big|\nonumber\\
&=& \sum_{u\in \{0,1\}} \Big|\Pr_{(\*A,\*A\*s+\*e)\leftarrow \textrm{GEN}_{\mathcal{F}_{\lwe}}(1^{\lambda})}\big[\mathcal{A}(\*A,\*A\*s+\*e)=\hat{d}\cdot s \oplus u\big]-\frac{1}{2}\Big|\\
&\geq&  \Big|\Pr_{(\*A,\*A\*s+\*e)\leftarrow \textrm{GEN}_{\mathcal{F}_{\lwe}}(1^{\lambda})}\big[\mathcal{A}(\*A,\*A\*s+\*e)=\hat{d}\cdot s\big]-\Pr_{(\*A,\*A\*s+\*e)\leftarrow \textrm{GEN}_{\mathcal{F}_{\lwe}}(1^{\lambda})}\big[\mathcal{A}(\*A,\*A\*s+\*e)=\hat{d}\cdot s \oplus 1\big]\Big|\nonumber\\
&\geq&  2\Big|\Pr_{(\*A,\*A\*s+\*e)\leftarrow \textrm{GEN}_{\mathcal{F}_{\lwe}}(1^{\lambda})}\big[\mathcal{A}(\*A,\*A\*s+\*e)=\hat{d}\cdot s\big]-\frac{1}{2}\Big|
\end{eqnarray}
By Claim \ref{cl:dependenceonsecretG}, we can replace $\hat{d}\cdot s$ with $J(\hat{\*d})\cdot(\inj(x)\oplus \inj( x - \*s ))$ to obtain:
\begin{eqnarray}
 &=& 2\Big|\Pr_{(\*A,\*A\*s+\*e)\leftarrow \textrm{GEN}_{\mathcal{F}_{\lwe}}(1^{\lambda})}\big[\mathcal{A}(\*A,\*A\*s+\*e) \in H'_{s,J(\hat{\*d})}\big]-\frac{1}{2}\Big|\\
 &\geq& 2\eta(\lambda)
\end{eqnarray}
Therefore, at least one of $\mathcal{A}'_0$ or $\mathcal{A}'_1$ must successfully distinguish between $\mathcal{D}_0$ and $\mathcal{D}_1$ with advantage at least $\eta$, a contradiction with the statement of Lemma~\ref{lem:lweadaptiveleakageG}. 
\end{proof}

\section{Acknowledgments}
Thanks to Dorit Aharonov, Zvika Brakerski, Zeph Landau, Umesh Vazirani and Thomas Vidick for many useful discussions.
\bibliographystyle{alpha}
\bibliography{qpip}

\end{document}

%% file: def.tex
\newenvironment{proof}{\noindent\textit{Proof: }}{$\Box $}

\newcommand{\QMA}{{\sffamily QMA}}
\newcommand{\BQP}{{\sffamily BQP}}
\newcommand{\QPIP}{{\sffamily QPIP}}

\newcommand{\BPP}{{\sffamily BPP}}

\newcommand{\Th}[1]{Theorem~\ref{#1}}

\newcommand{\Le}[1]{Lemma~\ref{#1}}
\newcommand{\Cl}[1]{Claim~\ref{#1}}

\newcommand{\EqDef}{\stackrel{\mathrm{def}}{=}}

\newcommand{\bra}[1]{\left< #1\right|}
\newcommand{\ket}[1]{\left| #1\right>}
\newcommand{\braket}[2]{\left< #1|#2\right>}

\newcommand{\tr}{\mbox{Tr}}

\newcommand{\mZ}{\mathbbm{Z}}
\newcommand{\mN}{\mathbbm{N}}
\newcommand{\mR}{\mathbbm{R}}
\newcommand{\mbP}{\mathbbm{P}}

\newcommand{\mcI}{\mathcal{I}}

\newcommand{\mcL}{\mathcal{L}}

\newcommand{\mfC}{\mathfrak{C}}

\newcommand{\ignore}[1]{}

\newtheorem{thm}{Theorem}[section]

\newtheorem{deff}[thm]{Definition}

\newtheorem{protocol}[thm]{Protocol}
\newtheorem{claim}[thm]{Claim}
\newtheorem{lem}[thm]{Lemma}

\newtheorem{notation}[thm]{Notation}
\newtheorem{corol}[thm]{Corollary}

\newenvironment{proofof}[1]{\noindent\textit {Proof of} $\mathbf{#1}$:\hspace*{1em}}{$\Box $}
\newenvironment{statement}[1]{\noindent{\textbf {#1}\hspace*{0.5em}}}{$\\$}

\def\hpic #1 #2 {\mbox{$\begin{array}[c]{l}
      \epsfig{file=#1,height=#2} \end{array}$}}

\def\vpic #1 #2 {\mbox{$\begin{array}[c]{l}
      \epsfig{file=#1,width=#2} \end{array}$}}

\newcommand{\trnq}[1]{\left[ {#1} \right]_q}

\newcommand{\lwe}{\mathrm{LWE}}

\newcommand{\vc}[1]{\mathbf{{#1}}}
\newcommand{\abs}[1]{\left\vert {#1} \right\vert}

\newcommand{\sivp}{\mathrm{SIVP}}
\newcommand{\otild}{{\widetilde{O}}}

\newcommand{\T}[2]{\left\|#1 - #2\right\|_{tr}}
\def\*#1{\mathbf{#1}}

\newcommand{\GenTrap}{\textsc{GenTrap}}
\newcommand{\Invert}{\textsc{Invert}}
\newcommand{\lossy}{\textsc{lossy}}
\newcommand{\sR}{{\mathcal{R}}}
\newcommand{\sX}{\mathcal{X}}
\newcommand{\sY}{{\mathcal{Y}}}
\newcommand{\odots}{{\otimes{\ldots}\otimes}}

\newcommand{\Prov}{\mathds{P}}
\newcommand{\Ver}{\mathds{V}}
\newcommand{\dset}{G}

\newcommand{\TV}[2]{\left\|#1 - #2\right\|_{TV}}
\def\*#1{\mathbf{#1}}

\newcommand{\supp}{\textsc{Supp}}
\newcommand{\Es}[1]{\textsc{E}_{#1}}
\newcommand{\inj}{J}